%% file: main.tex
\documentclass{article}
\usepackage[utf8]{inputenc}
\usepackage{mathrsfs}
\usepackage{amsmath}
\usepackage{amssymb}
\usepackage{amsthm}
\usepackage{authblk}
\usepackage{tikz}
\usetikzlibrary{automata, positioning}
\usepackage{float}
\usepackage{tabularx,lipsum,environ}
\usepackage{algorithmic}
\usepackage[ruled,vlined]{algorithm2e}
\usepackage{url}

\usetikzlibrary{graphs, graphs.standard}

\newtheorem{fact}{Fact}
\newtheorem{theorem}{Theorem}
\newtheorem{lemma}{Lemma}

\newtheorem{question}{Question}

\usepackage{hyperref}
\hypersetup{
    colorlinks=true,
    linkcolor=blue,
    filecolor=magenta,      
    urlcolor=cyan,
}

\title{Asymmetric Cryptosystem Using Careful Synchronization}
\author{Jakub Ruszil}
\affil{Jagiellonian University\\ Cracow\\ Poland
}
\date{February 2023}

\begin{document}

\maketitle

\begin{abstract}
    We present public-private key cryptosystem which utilizes the fact that checking whether a partial automaton is carefully synchronizing is $PSPACE$-complete, even in the case of a binary alphabet.
\end{abstract}

\section{Introduction}
Cryptography is essential branch of mathematics since the ancient times. It's main purpose is to ensure the privacy of information between sender and receiver sent through a possibly observed channel. Nowadays we differ symmetric cryptography - where the key used to cipher the message is the same as the one to decipher it - and asymmetric, where the key to cipher the message is commonly known and the one to decipher it is known only to the receiver of the message. In other words asymmetric cryptography is referred to as a public key cryptography, or a public-private key cryptography. The idea of public key cryptography was first mentioned in a confidential report GCHQ \cite{ellis} (UK Government Communications Headquarters) and later independently by Diffie and Hellman in 1976 \cite{diffie-hellman} along with the first practical public key cryptosystem based on knapsack problem. The mostly known asymmetric cryptosystem (RSA) was invented by Rivest, Shamir and Adleman in 1978 \cite{rsa} and is applicable since then to encryption and digital signatures. \newline
The concept of synchronization of finite automata is essential in various areas of computer science. It consists in regaining control over a system by applying a specific set of input instructions. These instructions lead the system to a fixed state no matter in which state it was at the beginning. The idea of synchronization has been studied for many classes of complete deterministic finite automata (DFA) \cite{berlinkov2014,berlinkov2016,eppstein1990,kari2001,kari2003,pin1983,rystsov1997,szykula2018,cerny1964,trahtman2007,volkov2008,volkov2019} and non-deterministic finite automata \cite{imreh1999,ito2004}. One of the most famous longstanding open problems in automata theory, known as \v{C}ern\'{y} Conjecture, states that for a given synchronizing DFA with $n$ states one can always find a synchronizing word of length at most $(n-1)^2$. This conjecture was proven for numerous classes of automata, but the problem is still not solved in general case. The concept of synchronization has been also considered in coding theory \cite{biskup2009,jurgensen2008}, parts orienting in manufacturing \cite{eppstein1990,natarajan1986}, testing of reactive systems \cite{sandberg2005} and Markov Decision Processes \cite{doyen2014,doyen2019}.

Allowing no outgoing transitions from some states for certain letters helps us to model a system for which certain actions cannot be accomplished while being in a specified state. This leads to the problem of finding a synchronizing word for a finite automaton, where transition function is not defined for all states. Notice that this is the most frequent case, if we use automata to model real-world systems. In practice, it rarely happens that a real system can be modeled with a DFA where transition function is total. The transition function is usually a partial one. This fact motivated many researchers to investigate the properties of partial finite automata relevant to practical problems of synchronization.\newline
We know that, in general case, checking if a partial automaton can be synchronized is PSPACE-complete \cite{martyugin2010b} even for binary alphabet \cite{vorel2016} and those facts are essential in our latter considerations.
\newline 
In this paper we present a public key cryptosystem utilizing fact, that checking if the PFA is carefully synchronizing is PSPACE-complete. This is however not the first attempt of trying to develop asymmetric cryptosystems with the notion of finite automata. Public key cryptography on finite automata with output is discussed in \cite{Tao2009} and uses the notion of invertible automata to provide the hard computational problem, inevitable to design such cryptosystem.\newline
The paper is organized as follows. In the section \ref{preliminaries_section} we provide with the basic notions and facts about synchronization of automata. In the sections \ref{encryption_section} and \ref{decryption_section} we present basic method of encryption and decryption using our cryptosystem. In the section \ref{section:extensions} we state couple of additional improvements to ensure better security. Finally we conclude the paper in the section \ref{section:conclusions} along with possible further research to the topic.
\section{Preliminaries}
\label{preliminaries_section}
\emph{Partial finite automaton} (PFA) is an ordered tuple $\mathcal{A} = (Q,\Sigma, \delta)$, where $\Sigma$ is a finite set of letters, $Q$ is a finite set of states and $\delta:{Q \times \Sigma^*}\rightarrow{Q}$ is a transition function, possibly not everywhere defined. In this definition we omit initial and final states, since they are not relevant to the problem of synchronization. For $\emph{w} \in \Sigma^\ast$ and $\emph{q} \in Q$ we define $\delta(\emph{q},\emph{w})$ inductively: $\delta(\emph{q},\epsilon) = q$ and $\delta(\emph{q},\emph{aw}) = \delta(\delta(\emph{q},\emph{a}), \emph{w})$ for $a \in \Sigma$, where $\epsilon$ is the empty word and $\delta(\emph{q}, \emph{a})$ is defined. A word $\emph{w} \in \Sigma^\ast$ is called \emph{carefully synchronizing} if there exists $\overline{q} \in Q$ such that for every $\emph{q} \in Q$, $\delta(\emph{q}, \emph{w}) = \overline{q}$ and all transitions $\delta(q, w')$, where $w'$ is any prefix of $w$, are defined. A PFA is called \emph{carefully synchronizing} if it admits any carefully synchronizing word. For a given $\mathcal{A}$ we define its \emph{power automaton} (which is itself a PFA) as $\mathcal{P}(\mathcal{A}) = (2^Q, \Sigma, \tau)$, where $2^Q$ stands for the set of all subsets of $Q$, and $\Sigma$ is the same as in $\mathcal{A}$. The transition function $\tau:{2^Q \times \Sigma}\rightarrow{2^Q}$ is defined as follows. Let $Q' \subseteq Q$. For every $a \in \Sigma$ we define $\tau(Q',a) = \bigcup_{q \in Q'} \delta(q,a)$ if
$\delta(q,a)$ is defined for all states $q\in Q'$,
otherwise $\tau(Q',a)$ is not defined. We also note $Q.w$ as an action of a word $w$ on a set of states $Q$ under the function $\delta$. Let $S \subseteq Q$. Then we denote $S.w^{-1}$ as a preimage of $S$ under the action of a word $w$. \newline 
We note that the above concepts can also be considered for \textit{deterministic finite automata} (DFA), for which the transition function is total. We define an $a$-cluster to be a DFA $\mathcal{A} = (Q, \{a\}, \delta)$ such that the automaton is connected. In other words it means that such automaton is a cycle on letter $a$ with paths that leads to the states of that cycle.  The set of states that induce a cycle in the $a$-cluster is referred to as the \textit{center} of the cluster. The \textit{depth} of the cluster is the length of the longest path to the center of the cluster. If $q$ belongs to the center of the $a$-cluster, the \textit{branch} of the state $q$ are the states that has a path to $q$ that does not have any other state belonging to the center. \textit{Destination} of the branch is a state in the center that has an in-transition from the last state of the branch.  Example of the $a$-cluster is depicted on Figure \ref{fig:cluster}. \newline
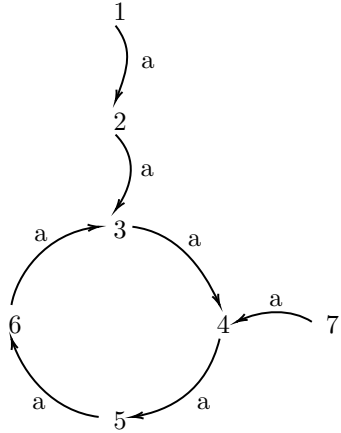
\begin{figure}[H]
    \centering
    \input{./cluster.tex}
    \caption{Example of the $a$-cluster}
    \label{fig:cluster}
\end{figure}
Center of that $a$-cluster is the set $\{3,4,5,6\}$, the depth is $2$ and there are two branches: $b_1 = \{1,2\}$ and $b_2 = \{7\}$. Destination of the branch $b_1$ is the state $3$ and of the branch $b_2$ is state $4$.
\newline

We define the sum of two automata $\mathcal{A} = (Q_1, \Sigma_1, \delta_1)$ and $\mathcal{B} = (Q_2, \Sigma_2, \delta_2)$ as $\mathcal{A} \cup \mathcal{B} = (Q_1 \cup Q_2, \Sigma_1 \cup \Sigma_2, \delta_1 \cup \delta_2)$. 
We can now state the obvious fact, useful to decide whether a given PFA is carefully synchronizing.
\begin{fact} 
\label{fact:1}
Let $\mathcal{A}$ be a PFA and $\mathcal{P(A)}$ be its power automaton. Then $\mathcal{A}$ is carefully synchronizing if and only if for some state $q \in Q$ there exists a path in $\mathcal{P(A)}$ from $Q$ to $\{q\}$. The shortest synchronizing word for $\mathcal{A}$ corresponds to the shortest such path in $\mathcal{P(A)}$.
\end{fact}
An example of a carefully synchronizing automaton $\mathcal{A}_{car}$ is depicted in Fig. \ref{fig:automaton}. One of its carefully synchronizing words is $aa(ba)^3bbab$.
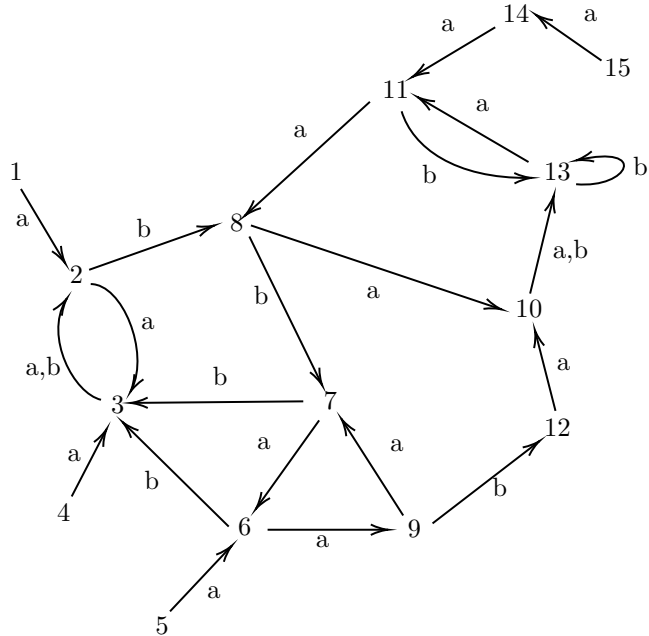
\begin{figure}[H]
    \centering
        \input{./automaton.tex}
    \caption{A carefully synchronizing automaton $\mathcal{A}_{car}$.}
    \label{fig:automaton}
\end{figure}
We recall the result of Vorel \cite{vorel2016} about the complexity of deciding whether a PFA is carefully synchronizing.

\begin{theorem}
Given a PFA $\mathcal{A} = (Q, \Sigma, \delta)$, checking if $\mathcal{A}$ is carefully synchronizing is $PSPACE$-complete even for $|\Sigma| = 2$.
\end{theorem}
Further we assume that $\Sigma = \{a,b\}$ and the letter $a$ is defined for all the states wherever not mentioned otherwise. Having that we can go to the description of our method.

\section{Basic encryption}
\label{encryption_section}
Let a plain text be the word $u\in \{0,1\}^*$. Choose a public key to be a carefully synchronizing PFA $\mathcal{A} = (Q, \Sigma, \delta)$ and a private key to be any word $w$ that carefully synchronizes $\mathcal{A}$. For simplicity of further statements we note $\mathcal{A}_i = (Q_i, \Sigma, \delta_i)$ to be isomorphic to $\mathcal{A}$ for any $i \in \mathbb{N}$. First we describe a construction that is a ciphertext. \newline
Define an automaton $\mathcal{P} = (\{p_1, p_2, .., p_{|u|+1}\}, \{0,1\}, \gamma)$ where $\gamma$ is defined as follows: for $i \in \{1, .., |u| + 1\}$ set $\gamma(p_{i-1}, u_i) = p_i$, where $u_i$ is $i$-th letter of a word $u$. In other words we encode our plaintext in the form of a directed path, where consecutive edges correspond to the consecutive letters of the word $u$. Encryption consists of four steps:
\begin{enumerate}
    \item \label{enc:point1}  Compute $\mathcal{B} = \bigcup_{i=1}^{|u| + 1}\mathcal{A}_i$ and denote $\bigcup_{i=1}^{|u| + 1} \delta_i = \rho$, $P = \bigcup_{i=1}^{|u| + 1} Q_i$
    \item \label{enc:point2} for any transition $(p_i, p_j)$ in $\mathcal{P}$, labelled with a letter $x \in \{0,1\}$ choose any pair of states $q^i \in Q_i$ and $q^j \in Q_j$, and set $\rho(q^i, x) = q^j$,
    \item \label{enc:point3} for all $i \in \{1, .., |u| + 1\}$ and for every letter $a \in \Sigma$, if $q^i \in Q_i$ and $\delta(q^i, a)$ is undefined, then choose any $j$ and any state $q^j \in Q_j$ and set $\rho(q^i, a) = q^j$,
    \item \label{enc:point4} for all $i \in \{1, .., |u| + 1\}$ choose $k_i \in \mathbb{N}$. Choose $k_i$ pairs $(q_p^i, q_r^i)$ and a letter $x \in \{0,1\}$ and define $\rho(q_p^i, x) = q_r^i$
\end{enumerate}
Automaton $\mathcal{B}$ is our ciphertext. It is straightforward from the construction, that computing such automaton is polynomial in terms of $Q, \Sigma$ and length of the plaintext. We also state two obvious observations.
\begin{fact}
\label{fact:2}
After removing letters $x \in \{0,1\}$ from automaton $\mathcal{B}$ we obtain a DFA over $\Sigma$.
\end{fact}

\begin{fact}
After removing letters $a \in \Sigma$ from automaton $\mathcal{B}$ we obtain a digraph labelled with letters $x \in \{0,1\}$ with longest path between the vertices of length 1.
\end{fact}

Procedure of encrypting the word $01$ is depicted on figures \ref{fig:first_step}, \ref{fig:second_step}, \ref{fig:third_step} and \ref{fig:fourth_step}. As a public key we take the automaton depicted on Figure \ref{fig:automaton}.
\begin{figure}[H]
    \centering
    \input{./first_step.tex}
    \caption{First step of encryption.}
    \label{fig:first_step}
\end{figure}
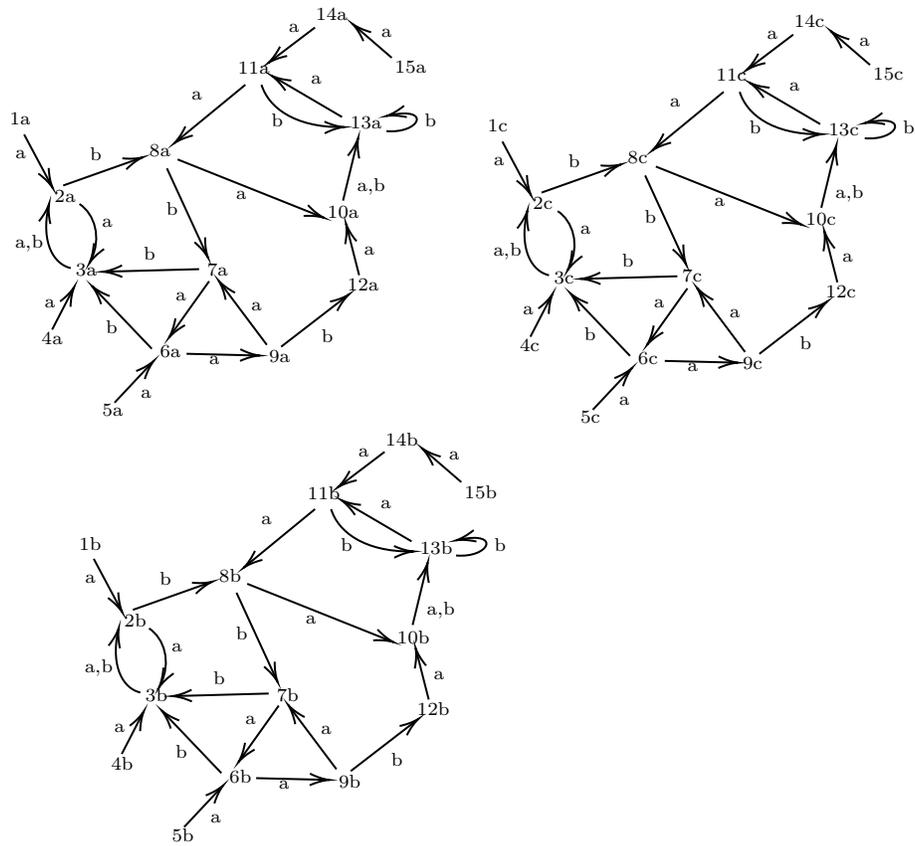
The first step involves summation of three copies of the public key that correspond to the three vertices of the word $01$ encoded as a labeled path. The first vertex of the path corresponds to the automaton induced by the states with suffix $a$, the second - by the states with the suffix $b$, and the third - by the suffix $c$.
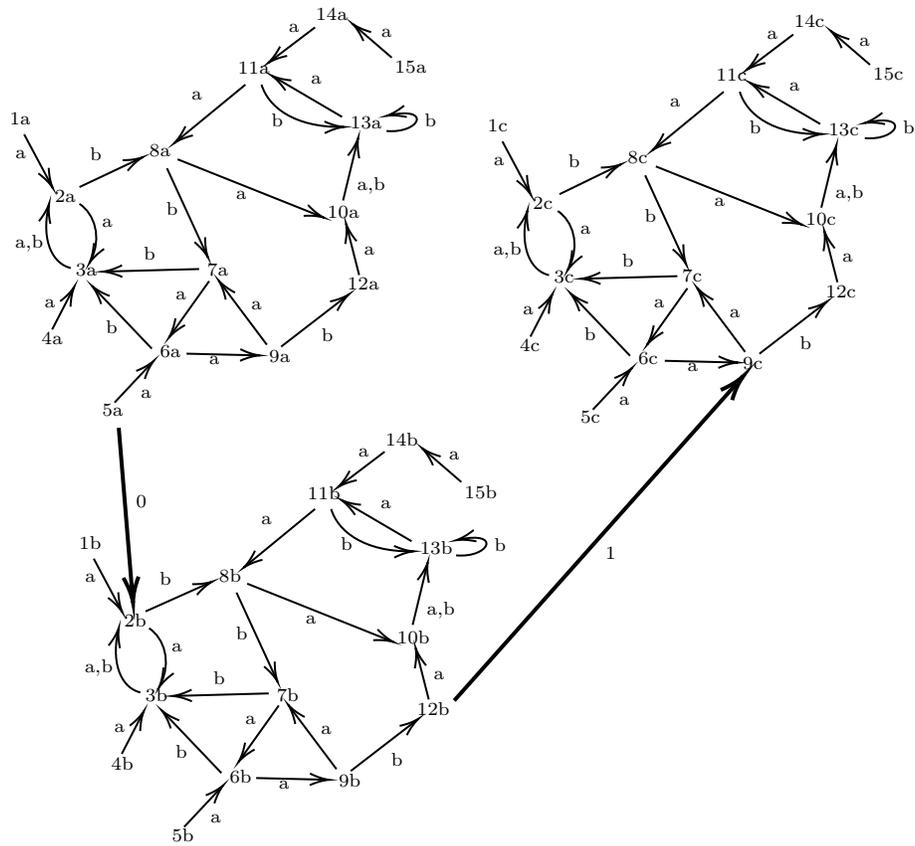
\begin{figure}[H]
    \centering
    \input{./second_step.tex}
    \caption{Second step of encryption.}
    \label{fig:second_step}
\end{figure}
The second step involves adding the transitions $0$ and $1$ to the states of automata that correspond to the in and out vertices of the transition. In the above example we define transition $\rho(5a, 0) = 2b$, which corresponds to the first transition of the encoded word, and $\rho(12b, 1) = 9c$, which corresponds to the second transition of the encoded word. Transitions added in this step are bolded.
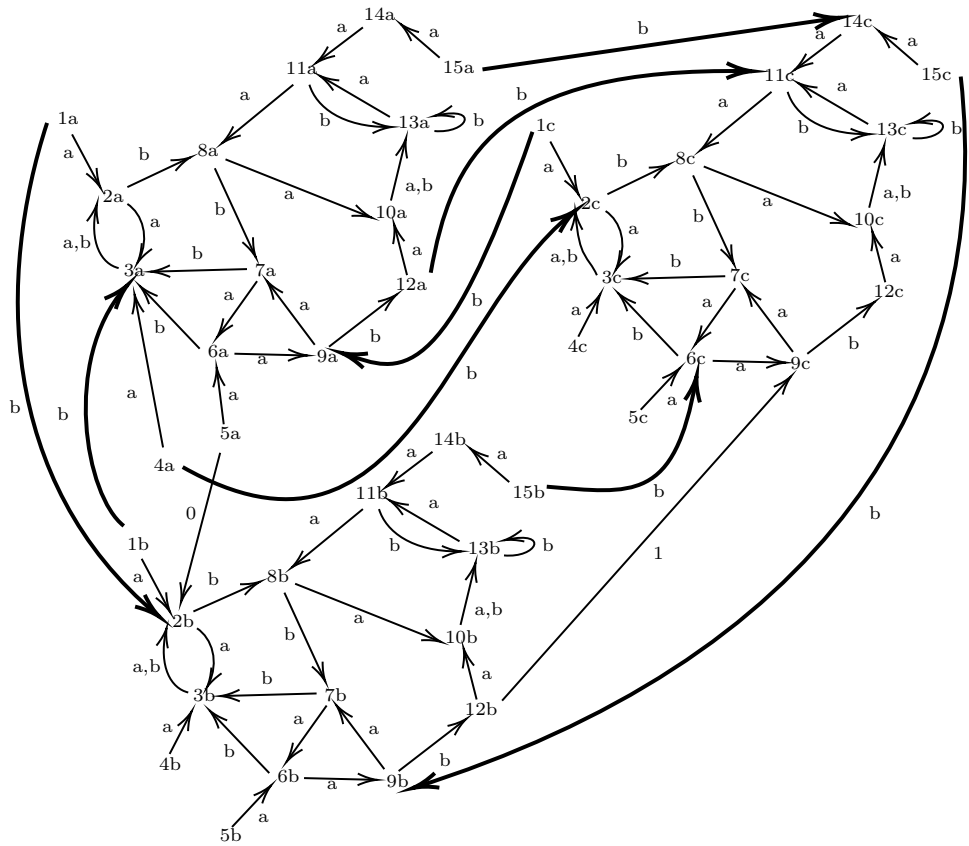
\begin{figure}[H]
    \centering
    \input{./third_step.tex}
    \caption{Third step of encryption.}
    \label{fig:third_step}
\end{figure}
The third step involves adding transitions from $\Sigma$ to those states in $\mathcal{B}$, which have undefined transitions for letters from $\Sigma$. In that case we add only $b$ letters. For example we defined $\rho(1a, b) = 2b$. We should act similarly for all states, for which $b$ is undefined, but we have only added some of the necessary transitions so the figure is readable.
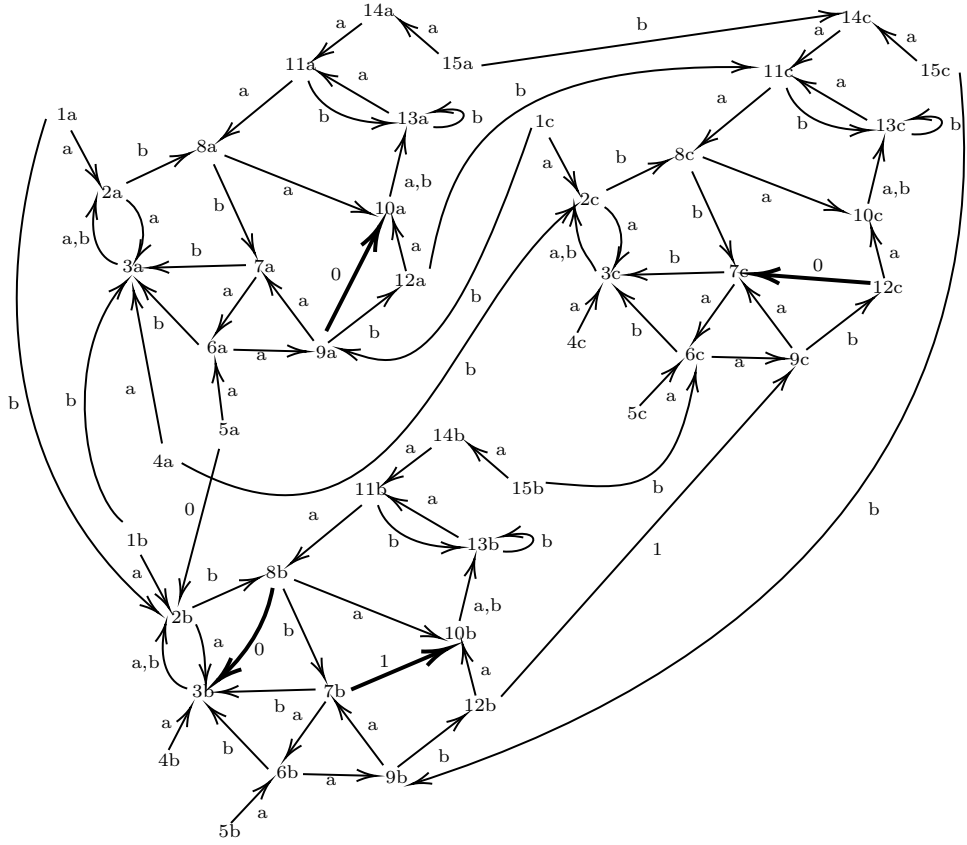
\begin{figure}[H]
    \centering
    \input{./fourth_step.tex}
    \caption{Fourth step of encryption.}
    \label{fig:fourth_step}
\end{figure}
The last step involves adding some number of transitions under letters from the alphabet $\{0,1\}$ between states belonging to the same copy of a public key in $\mathcal{B}$. In that case we have added transition $\rho(9a, 0) = 10a$ (first copy), transitions $\rho(7b, 1) = 10b$ and $\rho(8b, 1) = 3b$ (second copy) and transition $\rho(12c, 0) = 7c$ (third copy).
\section{Basic decryption}
\label{decryption_section}
For that section we assume that we have a ciphertext automaton $\mathcal{B} = (P, \Sigma, \rho)$ constructed from a public key $\mathcal{A} = (Q, \Sigma, \delta)$, and that we know a private key $w$ which is a carefully synchronizing word for the automaton $\mathcal{A}$. First we state a lemma.
\begin{lemma}
\label{lemma:1}
Let $Q.w = q_l$. After removing letters $x \in \{0,1\}$ from automaton $\mathcal{B}$ we have that $P.w = \{q_l^1, q_l^2, .., q_l^{|u| + 1}\}$.
\end{lemma}
\begin{proof}
It is immediate from construction, since we have not removed any transitions from $\Sigma$ within any $Q_i$, that for any $Q_i \in P$ holds $Q_i.w = q_l^i$, since $w$ carefully synchronizes $\mathcal{A}$ and each $Q_i$ on $\Sigma$ induces an isomorphic copy of $\mathcal{A}$. So we have, that $\{q_l^1, q_l^2, .., q_l^{|u| + 1}\} \subseteq P.w$. To prove that $P.w \subseteq \{q_l^1, q_l^2, .., q_l^{|u| + 1}\} \subseteq P.w$ it suffices to notice, that, from fact \ref{fact:2} automaton $\mathcal{B}$ is deterministic and for any prefix $w'$ of $w$ if $Q.w' = \{q_{k_1}, .., q_{k_s}\}$, then $Q_i.w' = \{q_{k_1}^i, .., q_{k_s}^i\}$.  
\end{proof}

\begin{lemma}
\label{lemma:2}
There exist an algorithm with $O(|P||w|)$ time complexity and $O(|P||w|)$ space complexity which computes a partition of $P$ on sets $Q_1$, $Q_2$, ..., $Q_{|u|+ 1}$. 
\end{lemma}
\begin{proof}
We describe a desired algorithm. Suppose we have an array with $|P|$ columns and $|w| + 1$ rows. Put every element of $P$ in a different column of a first row. Then we fill the $i$-th row by taking a state from the $(i-1)$-th row of the corresponding column and applying to it the $i$-th letter of a word $w$ until the end of the row. After this procedure, from lemma \ref{lemma:1}, the last row contains only the states from the set $\{q_l^1, q_l^2, .., q_l^{|u| + 1}\}$. We can now compute each $Q_i$ by taking those states from the first row that lie in the same columns as the state $q_l^i$.
\end{proof}

\noindent With these two lemmas we are ready to present a decryption method:
\begin{enumerate}
    \item \label{dec:point1} using Lemma \ref{lemma:1} compute the set $\{q_l^1, q_l^2, .., q_l^{|u| + 1}\}$,
    \item \label{dec:point2} using Lemma \ref{lemma:2} compute the partition of $P$ on sets $Q_1$, $Q_2$, ..., $Q_{|u|+ 1}$,
    \item \label{dec:point3} for every transition $x \in \{0,1\}$ in $\mathcal{B}$ if $x$ joins a states from different sets, say $Q_i$ and $Q_j$, then join $q_l^i$ and $q_l^j$ with transition $x$, otherwise remove the transition.
\end{enumerate}
Observe, that after applying that procedure to the ciphertext $\mathcal{B}$ we end up with a graph that was our plaintext, what can be concluded directly from the encryption procedure.
In general one can decipher the message only by knowing any carefully synchronizing word for $\mathcal{A}$ or computing every possible induced subautomaton isomorphic to $\mathcal{A}$.
\color{black}

\section{Extensions}
\label{section:extensions}
As the ciphertext which is a result of our encryption method consists of $n$ copies of isomorphic automaton with added transitions between those copies one can think of more "sophisticated" method of creating a ciphertext. As mentioned in the previous section a potential attacker can decipher the message computing every possible induced subautomaton isomorphic to a public key. However, the problem of determining for two given graphs say $G$ and $H$, whether $G$ has a copy of $H$ as an induced subgraph is $NP$-complete \cite{10.5555/574848}.  In this section we present two lemmas that can be used to obfuscate the ciphertext even more. The first one involves adding the state to the public key and the second one adding arbitrary number of $a$-clusters to the ciphertext. 

\begin{lemma}
\label{lemma:add_state}
Let $\mathcal{A} = (Q, \Sigma, \delta)$ be a PFA with carefully synchronizing word $w$. Further, let $q \in Q$ be such that there exists $p \in Q$ such that $q \in p.a^{-1}$. Let also $\mathcal{A}' = (Q \cup \{q'\}, \Sigma, \delta')$ where $\delta'$ is defined as $\delta$ on $Q$ and $\delta'(q',a) = q$. Then $w$ carefully synchronizes $\mathcal{A}'$.
\end{lemma}
\begin{proof}
Since $a$ is defined for all states of $Q'$, and $|\Sigma| = 2$, then the first letter of $w$ must be $a$. Let $w'$ be the word $w$ without the first letter. Since $\delta'(q',a) = q$ and we assumed that there exist $p \in Q$ such that $q \in p.a^{-1}$ it is straightforward, that $Q'.a = Q.a$. Since we have not added any other transitions to $\mathcal{A}'$ and $\delta'$ is defined as $\delta$ on $Q$, we obtain that $Q'.aw' = Q.aw' = Q.w$ and that concludes the proof. 
\end{proof}
For the next lemma we assume notation as in former part of the paper.
\begin{lemma}
\label{lemma:add_cycle}

Let $\mathcal{B} = \bigcup_{i=1}^k \mathcal{A}_k$ and $m \in \mathbb{N}$ be the smallest integer such that $Q.a^m = Q.a^{m+1}$. Define $B_i=Q_i.a^mb$ and let $\mathcal{C}_1 = (S_1, \{a\}, \eta_1), \ldots, \mathcal{C}_l = (S_l, \{a\}, \eta_l)$ be $a$-clusters with depth $1$ and centers $K_1, \ldots, K_l$ respectively. Let $\mathcal{B}'= \mathcal{B} \cup \bigcup_{i=1}^l \mathcal{C}_i = (P', \Sigma, \rho')$. If we define $b$ transitions for all states $q \in \bigcup_{i=1}^l K_i$ such that there exists $0 < j < k+1$ such that  $\rho'(q,b) \in B_j$ then $P'.w = \{q_l^1, ..., q_l^k\}$.
\end{lemma}
\begin{proof}
Since $a$ is the only letter defined for all states in $\mathcal{A}$ and $Q.a^m = Q.a^{m+1}$ then $w$ starts with a word $a^{m_1}b$ for $0 < m_1 < m+1$. Note $w = a^{m_1}bw'$ Observe that $Q.a^{i+1} \subseteq Q.a^i$ for all $i \geq 0$. From that we have, that $Q.a^m \subseteq Q.a^{m_1}$ and further for all copies of $\mathcal{A}$ in $\mathcal{B}'$ we obtain that $B_i \in Q_i.a^{m_1}.b$. Also, since the depth of any cluster $\mathcal{C}_i$ is $1$, we have that $P_j.a^{m_1} = K_j$ for all $0 < j < m+1$. Notice that $P' = \bigcup_{i=1}^k Q_i \cup \bigcup_{i=1}^l S_i$, so $$P'.w = \bigcup_{i=1}^k Q_i.w \cup \bigcup_{i=1}^l S_i.w = \bigcup_{i=1}^k Q_i.a^{m_1}bw' \cup \bigcup_{i=1}^l S_i.a^{m_1}bw'$$ which gives 
$$P'.w  = \bigcup_{i=1}^k B_iw' \cup \bigcup_{i=1}^l K_ibw'.$$
But we know, that for all $q \in \bigcup_{i=1} K_i$ there exists $0 < i < k+1$ such that  $\delta'(q,b) \in B_i$. From that we obtain $$P'.w  = \bigcup_{i=1}^k B_iw',$$ and since each $B_i = Q_i.a^mb$ then $B_i.w' = Q_i.a^mbw' = Q_i.w = \{q_l^i\}$ and that concludes the proof.
\end{proof}

Using these two lemmas we can move on to the description of the extended method of encryption and decryption. In the next two subsections we follow the notation provided in sections \ref{encryption_section} and \ref{decryption_section}.

\subsection{Extended encryption}
The extension consists of adding two stages between the \ref{enc:point1} and \ref{enc:point2} stage of encryption method, defining sets $Q_i'$ and substitute them for $Q_i$ in latter stages. Let us state two additional stages:
\begin{enumerate}
    \item \label{ext_enc:point1} add $l$ $a$-clusters with depth $1$ to automaton obtained in stage \ref{enc:point1} and define letters $b$ for centers of those clusters to fulfill assumptions of lemma \ref{lemma:add_cycle} in $\rho$ function (defined in section \ref{encryption_section})
    \item \label{ext_enc:point2}  for each copy $\mathcal{A}_i$ of public key in automaton obtained in previous stage add $k_i$ states and define transitions as in lemma \ref{lemma:add_state} and note the set of the added states in this stage as $A_i$ for each $\mathcal{A}_i$
\end{enumerate}
Now let us define sets $Q_i'$. For clusters $\mathcal{C}_1 = (S_1, \{a\}, \gamma_1), \ldots, \mathcal{C}_l = (S_l, \{a\}, \gamma_l)$ (from stage \ref{ext_enc:point1}) with centers $K_1, \ldots, K_l$ respectively we define sets $C_1, \ldots, C_{|u|+1}$, such that if for $q \in K_i$ it holds $\rho(q, b) \in B_j$ (notation from lemma \ref{lemma:add_cycle}), then $q$ and its branch belong to the set $C_j$. Then define $Q_i' = Q_i \cup A_i \cup C_i$. It is a simple exercise to prove that the sets $Q_1', \ldots, Q_{|u| + 1}'$ form a partition of $P = \bigcup_{i=1}^{|u|+1} Q_i \cup A_i \cup C_i$ which is the set of all states of our ciphertext. The latter stages remain as in section \ref{encryption_section}.

\subsection{Extended decryption}
Algorithm of deciphering is similar to the one described in section \ref{decryption_section}. We state lemmas being in a strict correspondence with those proven in section \ref{decryption_section}.
\begin{lemma}
\label{lemma:extended_1}
Let $\mathcal{B}$ be a ciphertext computed by extended encryption method using public key $\mathcal{A} = (Q, \Sigma, \delta)$ and $Q.w = q_l$. After removing letters $x \in \{0,1\}$ from automaton $\mathcal{B}$ we have that $P.w = \{q_l^1, q_l^2, \ldots, q_l^{|u| + 1}\}$.
\end{lemma}
\begin{proof}
Observe that after stage \ref{ext_enc:point1} we can apply Lemma \ref{lemma:add_cycle} and we obtain that $(\bigcup_{i=1}^{|u|+1} Q_i \cup C_i).w = \{q_l^1, q_l^2, .., q_l^{|u| + 1}\}$. Notice that after stage \ref{ext_enc:point2} we can apply Lemma \ref{lemma:add_state} to any copy of public key that was modified by that stage and also $P.w = \{q_l^1, q_l^2, .., q_l^{|u| + 1}\}$. The rest of the proof is similar to the proof of Lemma \ref{lemma:1}.
\end{proof}

\begin{lemma}
\label{lemma:extended_2}
There exist an algorithm with polynomial time complexity (depending on $|P|$ and $|w|$) which computes a partition of $P$ on sets $Q_1'$, $Q_2'$, ..., $Q_{|u|+ 1}'$. 
\end{lemma}
\begin{proof}
Using approach from the proof of Lemma \ref{lemma:2} we can compute similar matrix, say $M$, in time $O(|P||w|)$. From Lemma \ref{lemma:extended_2} we know the last row contains only the states from the set $\{q_l^1, q_l^2, 
\ldots, q_l^{|u| + 1}\}$ and we can compute sets $\bar{Q}_1, \ldots, \bar{Q}_{|u| + 1}$ such if column of the first row containing $q$ is the same as the column of the last row containing $q_l^i$, then $q \in \bar{Q}_i$. Notice that there are three cases, when $q \in \bar{Q}_i$:
\begin{itemize}
    \item $q \in Q_i$
    \item $q \in A_i$
    \item $q \in S_m$ such that there exist $p \in C_i \cap S_m$ (notation from Lemma 4)
\end{itemize}
First two cases are straightforward. To prove the theorem for the third case observe that if $q \notin Q_i \cup A_i$, then $q \notin A_j$ and $q \notin Q_k$ for any $j,k \neq i$ otherwise $\mathcal{B}$ would be non-deterministic. So we deduce that $q \in S_m$ for some $m$. For the sake of contradiction suppose that $C_i \cap S_m = \varnothing$. But that means, that $q.a^mb \in B_j$ for $j \neq i$ and further $q.ab^mw' = q.w = q_l^j$ what is a contradiction. From these considerations we are able to determine for each $i$ the sets $A_i$ and $Q_i$ that are subsets of the set $Q_i'$.
In order to compute the sets $C_i$ we first compute $S_1',.., S_n'$ inducing all $a$-clusters in $\mathcal{B}$ by removing $b, 0, 1$ transitions and determine all connected components of the resulting structure. Now we examine three cases for a cluster $S_j'$:
\begin{itemize}
    \item $S_j' \cap \bar{Q}_i = \varnothing$
    \item $S_j' \subseteq \bar{Q}_i$
    \item $S_j' \cap \bar{Q}_i \neq \varnothing$ and $S_j' \not\subset \bar{Q}_i$
\end{itemize}

Notice, that if the first case holds we know that no state of $S_j'$ belongs to $C_i$. If the second case holds, we must check if $S_j' \subseteq Q_i \cup A_i$. If this is not true, then we have found a cluster $\mathcal{C}_m$, such that for all $q \in K_m$ it holds $\rho(q, b) \in B_i$ and we determined the $a$-cluster that belongs to $C_i$. In the third case we know that some of the states of the cluster $S_j'$ are in $C_i$ and some are not. To compute those that are let us take the center of the $a$-cluster $S_j'$, say $K_j'$, and observe that $q \in C_i$ if, and only if $q \in \{p \in K_j': \rho(p, b) \in B_i\} = K_j''$ or $q$ belongs to some branch with destination in $K_j''$. That concludes the proof.
\end{proof}
Using two former lemmas, decryption method is similar as in \ref{decryption_section}. Extended step is depicted on Figure \ref{fig:extended_step}. If we choose the public key to be the automaton on Figure \ref{fig:automaton}, then notice, that in Lemma $\ref{lemma:add_cycle}$ we have $m=2$, and $Q.a^2b = \{2,3,7,12,13\}$.

\begin{figure}[H]
    \centering
    \input{./extended_step.tex}
    \caption{Extended step of encryption.}
    \label{fig:extended_step}
\end{figure}
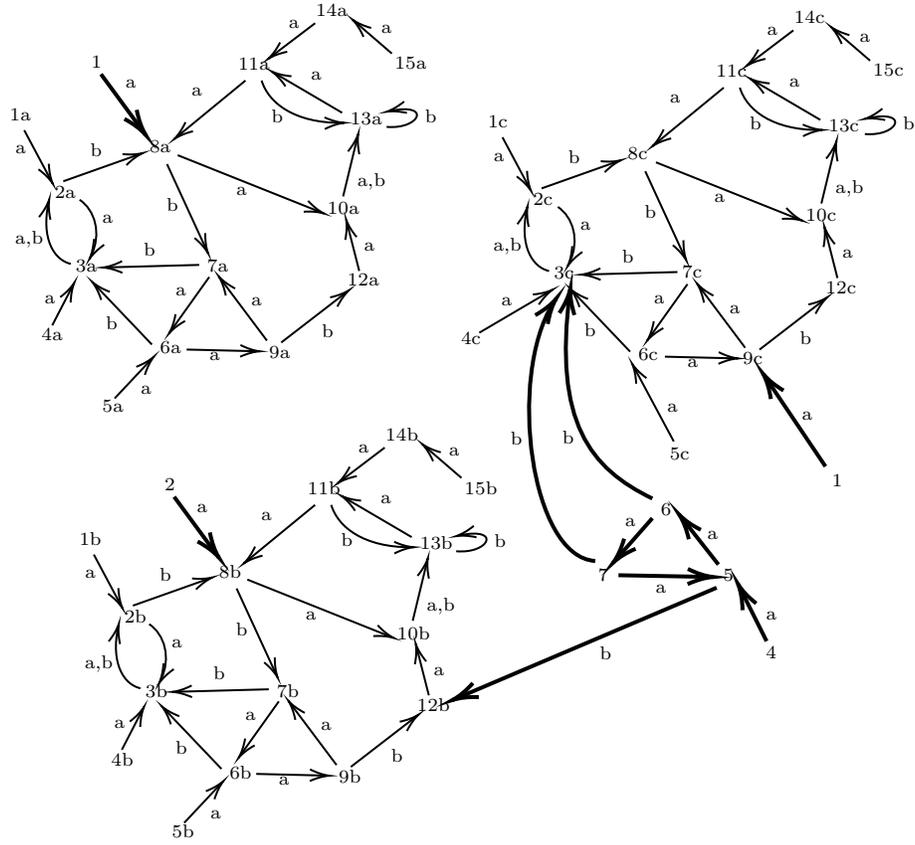
Observe that we can apply Lemma \ref{lemma:add_state} to states $1,2,3$. We also added cluster that consists of states $4,5,6,7$ so we can apply Lemma \ref{lemma:add_cycle}. In the former method we had that $Q_1 = \{1a,..,15a\}, Q_2 = \{1b,..,15b\},Q_3 = \{1c,..,15c\}$ and now $Q_1' = Q_1 \cup \{1\}, Q_2' = Q_2 \cup \{2,4,5\}, Q_3' = Q_3 \cup \{3,6,7\}$.
\section{Conclusions and further work}
\label{section:conclusions}
We proposed a method of utilizing careful synchronization to provide brand new public key cryptosystem. In sections \ref{encryption_section} and \ref{decryption_section} we presented core idea of our method and provided an example that illustrates it. As the ciphertext in that method consists of $n$ copies of the same automaton, those two sections are included to so the reader could understand the method presented in section \ref{section:extensions}.
\newline It should be also mentioned that lemmas \ref{lemma:add_cycle} and \ref{lemma:add_state} are only examples of extensions of that cryptosystem. Indeed, observe that Lemma \ref{lemma:add_cycle} provides a possibility to add "free" $a$-clusters to a ciphertext. The disadvantage of that extension is that we only can add $b$ transitions to states that are some specified states of the copy of the public key. It is possible also to add the extension, that allows us to define the $a$-clusters for which we can define $b$ transition outside of that specific sets $B_i$ but to whatever state we want, even the other added $a$-cluster. However this extension would cause that in Lemma \ref{lemma:extended_1} it would be only $\{q_l^1, .., q_l^{|u| + 1}\} \in P.w$ so the number of states added in such extensions would have been bounded by $\text{min}(|Q_1'|,.., |Q_l'|)$ and also demanded modifications in Lemma \ref{lemma:extended_2} so we omitted that extension.
\newline Observe also that point \ref{enc:point2} of encryption procedure can be modified in many ways. For example one can choose to define more than one transition between copies of automata and in decryption section choose the one that has odd or even number in a ciphertext.
We end up with several questions and open problems:
\begin{question}
What is the most reasonable way to define lacking transitions in point \ref{enc:point3}?
\end{question}
It is straightforward that if all lacking transitions in a copy of a public key are defined within the same copy, it would result with $|u| + 1$ connected automata which are not connected between each other, and that simplifies the attack on the cryptosystem. 
\begin{question}
What is the most reasonable way to define transitions in point \ref{enc:point4}?
\end{question}
We have defined step \ref{enc:point4} in an abstract way, so to investigate many versions of adding those "obfuscating" $\{0,1\}$ transitions.
\begin{question}
Find an algorithm that generates pairs of public and private keys.
\end{question}
We believe that the most promising approach will be to construct a PFA that is carefully synchronized by a given $w$. We also want to investigate if it is possible to design an algorithm that for a given word $w$ generates $n$ non-isomorphic PFA's that are carefully synchronized by $w$. Having that one could take as a public key a tuple of $n$ automata that are synchronized by the same word $w$. In that case, all methods presented in the paper would need only slight modifications to work properly. 
\bibliography{bibliography.bib}

\end{document}

%% file: cluster.tex
\tikzset{every picture/.style={line width=0.75pt}} 

\begin{tikzpicture}[x=0.75pt,y=0.75pt,yscale=-0.5,xscale=0.5]

\draw    (299,77) .. controls (316.73,100.15) and (312.14,118.93) .. (299.58,151.5) ;
\draw [shift={(299,153)}, rotate = 291.21] [color={rgb, 255:red, 0; green, 0; blue, 0 }  ][line width=0.75]    (10.93,-3.29) .. controls (6.95,-1.4) and (3.31,-0.3) .. (0,0) .. controls (3.31,0.3) and (6.95,1.4) .. (10.93,3.29)   ;
\draw    (299,187) .. controls (318.7,207.19) and (319.97,232.72) .. (299.93,261.67) ;
\draw [shift={(299,263)}, rotate = 305.45] [color={rgb, 255:red, 0; green, 0; blue, 0 }  ][line width=0.75]    (10.93,-3.29) .. controls (6.95,-1.4) and (3.31,-0.3) .. (0,0) .. controls (3.31,0.3) and (6.95,1.4) .. (10.93,3.29)   ;
\draw    (497,376) .. controls (476.42,362.77) and (447.2,363.47) .. (422.51,375.26) ;
\draw [shift={(421,376)}, rotate = 333.43] [color={rgb, 255:red, 0; green, 0; blue, 0 }  ][line width=0.75]    (10.93,-3.29) .. controls (6.95,-1.4) and (3.31,-0.3) .. (0,0) .. controls (3.31,0.3) and (6.95,1.4) .. (10.93,3.29)   ;
\draw    (316,280) .. controls (359.34,285.42) and (386.19,319.46) .. (403.23,357.27) ;
\draw [shift={(404,359)}, rotate = 246.18] [color={rgb, 255:red, 0; green, 0; blue, 0 }  ][line width=0.75]    (10.93,-3.29) .. controls (6.95,-1.4) and (3.31,-0.3) .. (0,0) .. controls (3.31,0.3) and (6.95,1.4) .. (10.93,3.29)   ;
\draw    (404,393) .. controls (391.2,445.7) and (354.13,465.89) .. (317.67,471.74) ;
\draw [shift={(316,472)}, rotate = 351.54] [color={rgb, 255:red, 0; green, 0; blue, 0 }  ][line width=0.75]    (10.93,-3.29) .. controls (6.95,-1.4) and (3.31,-0.3) .. (0,0) .. controls (3.31,0.3) and (6.95,1.4) .. (10.93,3.29)   ;
\draw    (282,472) .. controls (236.69,466.58) and (205.93,429.63) .. (194.51,394.6) ;
\draw [shift={(194,393)}, rotate = 72.78] [color={rgb, 255:red, 0; green, 0; blue, 0 }  ][line width=0.75]    (10.93,-3.29) .. controls (6.95,-1.4) and (3.31,-0.3) .. (0,0) .. controls (3.31,0.3) and (6.95,1.4) .. (10.93,3.29)   ;
\draw    (194,359) .. controls (198.95,329.79) and (227.42,286.38) .. (280.39,280.18) ;
\draw [shift={(282,280)}, rotate = 174.18] [color={rgb, 255:red, 0; green, 0; blue, 0 }  ][line width=0.75]    (10.93,-3.29) .. controls (6.95,-1.4) and (3.31,-0.3) .. (0,0) .. controls (3.31,0.3) and (6.95,1.4) .. (10.93,3.29)   ;

\draw (322,105) node [anchor=north west][inner sep=0.75pt]   [align=left] {a};
\draw (321,214) node [anchor=north west][inner sep=0.75pt]   [align=left] {a};
\draw (369,285) node [anchor=north west][inner sep=0.75pt]   [align=left] {a};
\draw (451,346) node [anchor=north west][inner sep=0.75pt]   [align=left] {a};
\draw (378,450) node [anchor=north west][inner sep=0.75pt]   [align=left] {a};
\draw (212,450) node [anchor=north west][inner sep=0.75pt]   [align=left] {a};
\draw (214,280) node [anchor=north west][inner sep=0.75pt]   [align=left] {a};
\draw (294,51) node [anchor=north west][inner sep=0.75pt]   [align=left] {1};
\draw (294,162) node [anchor=north west][inner sep=0.75pt]   [align=left] {2};
\draw (294,271) node [anchor=north west][inner sep=0.75pt]   [align=left] {3};
\draw (398,367) node [anchor=north west][inner sep=0.75pt]   [align=left] {4};
\draw (508,367) node [anchor=north west][inner sep=0.75pt]   [align=left] {7};
\draw (294,464) node [anchor=north west][inner sep=0.75pt]   [align=left] {5};
\draw (188,367) node [anchor=north west][inner sep=0.75pt]   [align=left] {6};

\end{tikzpicture}

%% file: automaton.tex
\tikzset{every picture/.style={line width=0.75pt}} 

\begin{tikzpicture}[x=0.75pt,y=0.75pt,yscale=-0.8,xscale=0.8]

\draw    (152,257) .. controls (171.99,257) and (189.12,299.78) .. (177.92,323.23) ;
\draw [shift={(177,325)}, rotate = 299.48] [color={rgb, 255:red, 0; green, 0; blue, 0 }  ][line width=0.75]    (10.93,-3.29) .. controls (6.95,-1.4) and (3.31,-0.3) .. (0,0) .. controls (3.31,0.3) and (6.95,1.4) .. (10.93,3.29)   ;
\draw    (158.5,330) .. controls (141.35,326.08) and (122.76,290.47) .. (136.14,267.4) ;
\draw [shift={(137,266)}, rotate = 123.11] [color={rgb, 255:red, 0; green, 0; blue, 0 }  ][line width=0.75]    (10.93,-3.29) .. controls (6.95,-1.4) and (3.31,-0.3) .. (0,0) .. controls (3.31,0.3) and (6.95,1.4) .. (10.93,3.29)   ;
\draw    (263.5,412) -- (337.5,412) ;
\draw [shift={(339.5,412)}, rotate = 180] [color={rgb, 255:red, 0; green, 0; blue, 0 }  ][line width=0.75]    (10.93,-3.29) .. controls (6.95,-1.4) and (3.31,-0.3) .. (0,0) .. controls (3.31,0.3) and (6.95,1.4) .. (10.93,3.29)   ;
\draw    (296,343) -- (257.16,397.37) ;
\draw [shift={(256,399)}, rotate = 305.54] [color={rgb, 255:red, 0; green, 0; blue, 0 }  ][line width=0.75]    (10.93,-3.29) .. controls (6.95,-1.4) and (3.31,-0.3) .. (0,0) .. controls (3.31,0.3) and (6.95,1.4) .. (10.93,3.29)   ;
\draw    (349,403) -- (311.08,343.69) ;
\draw [shift={(310,342)}, rotate = 57.41] [color={rgb, 255:red, 0; green, 0; blue, 0 }  ][line width=0.75]    (10.93,-3.29) .. controls (6.95,-1.4) and (3.31,-0.3) .. (0,0) .. controls (3.31,0.3) and (6.95,1.4) .. (10.93,3.29)   ;
\draw    (429,263) -- (443.53,202.94) ;
\draw [shift={(444,201)}, rotate = 103.6] [color={rgb, 255:red, 0; green, 0; blue, 0 }  ][line width=0.75]    (10.93,-3.29) .. controls (6.95,-1.4) and (3.31,-0.3) .. (0,0) .. controls (3.31,0.3) and (6.95,1.4) .. (10.93,3.29)   ;
\draw    (428,180) -- (359.73,140.01) ;
\draw [shift={(358,139)}, rotate = 30.36] [color={rgb, 255:red, 0; green, 0; blue, 0 }  ][line width=0.75]    (10.93,-3.29) .. controls (6.95,-1.4) and (3.31,-0.3) .. (0,0) .. controls (3.31,0.3) and (6.95,1.4) .. (10.93,3.29)   ;
\draw    (328,142) -- (250.48,212.65) ;
\draw [shift={(249,214)}, rotate = 317.65] [color={rgb, 255:red, 0; green, 0; blue, 0 }  ][line width=0.75]    (10.93,-3.29) .. controls (6.95,-1.4) and (3.31,-0.3) .. (0,0) .. controls (3.31,0.3) and (6.95,1.4) .. (10.93,3.29)   ;
\draw    (253,220) -- (409.1,272.36) ;
\draw [shift={(411,273)}, rotate = 198.54] [color={rgb, 255:red, 0; green, 0; blue, 0 }  ][line width=0.75]    (10.93,-3.29) .. controls (6.95,-1.4) and (3.31,-0.3) .. (0,0) .. controls (3.31,0.3) and (6.95,1.4) .. (10.93,3.29)   ;
\draw    (108,197) -- (134.98,242.28) ;
\draw [shift={(136,244)}, rotate = 239.22] [color={rgb, 255:red, 0; green, 0; blue, 0 }  ][line width=0.75]    (10.93,-3.29) .. controls (6.95,-1.4) and (3.31,-0.3) .. (0,0) .. controls (3.31,0.3) and (6.95,1.4) .. (10.93,3.29)   ;
\draw    (140,391) -- (161.08,350.28) ;
\draw [shift={(162,348.5)}, rotate = 117.37] [color={rgb, 255:red, 0; green, 0; blue, 0 }  ][line width=0.75]    (10.93,-3.29) .. controls (6.95,-1.4) and (3.31,-0.3) .. (0,0) .. controls (3.31,0.3) and (6.95,1.4) .. (10.93,3.29)   ;
\draw    (202,463.5) -- (239.63,423.46) ;
\draw [shift={(241,422)}, rotate = 133.22] [color={rgb, 255:red, 0; green, 0; blue, 0 }  ][line width=0.75]    (10.93,-3.29) .. controls (6.95,-1.4) and (3.31,-0.3) .. (0,0) .. controls (3.31,0.3) and (6.95,1.4) .. (10.93,3.29)   ;
\draw    (474,116) -- (434.14,88.15) ;
\draw [shift={(432.5,87)}, rotate = 34.95] [color={rgb, 255:red, 0; green, 0; blue, 0 }  ][line width=0.75]    (10.93,-3.29) .. controls (6.95,-1.4) and (3.31,-0.3) .. (0,0) .. controls (3.31,0.3) and (6.95,1.4) .. (10.93,3.29)   ;
\draw    (407,95) -- (356.21,125.47) ;
\draw [shift={(354.5,126.5)}, rotate = 329.04] [color={rgb, 255:red, 0; green, 0; blue, 0 }  ][line width=0.75]    (10.93,-3.29) .. controls (6.95,-1.4) and (3.31,-0.3) .. (0,0) .. controls (3.31,0.3) and (6.95,1.4) .. (10.93,3.29)   ;
\draw    (445,337) -- (432.49,287.94) ;
\draw [shift={(432,286)}, rotate = 75.7] [color={rgb, 255:red, 0; green, 0; blue, 0 }  ][line width=0.75]    (10.93,-3.29) .. controls (6.95,-1.4) and (3.31,-0.3) .. (0,0) .. controls (3.31,0.3) and (6.95,1.4) .. (10.93,3.29)   ;
\draw    (286,331) -- (179,331.98) ;
\draw [shift={(177,332)}, rotate = 359.47] [color={rgb, 255:red, 0; green, 0; blue, 0 }  ][line width=0.75]    (10.93,-3.29) .. controls (6.95,-1.4) and (3.31,-0.3) .. (0,0) .. controls (3.31,0.3) and (6.95,1.4) .. (10.93,3.29)   ;
\draw    (239,410) -- (173.42,345.4) ;
\draw [shift={(172,344)}, rotate = 44.57] [color={rgb, 255:red, 0; green, 0; blue, 0 }  ][line width=0.75]    (10.93,-3.29) .. controls (6.95,-1.4) and (3.31,-0.3) .. (0,0) .. controls (3.31,0.3) and (6.95,1.4) .. (10.93,3.29)   ;
\draw    (367.5,408) -- (433.42,357.22) ;
\draw [shift={(435,356)}, rotate = 142.39] [color={rgb, 255:red, 0; green, 0; blue, 0 }  ][line width=0.75]    (10.93,-3.29) .. controls (6.95,-1.4) and (3.31,-0.3) .. (0,0) .. controls (3.31,0.3) and (6.95,1.4) .. (10.93,3.29)   ;
\draw    (252,227) -- (297.12,319.2) ;
\draw [shift={(298,321)}, rotate = 243.92] [color={rgb, 255:red, 0; green, 0; blue, 0 }  ][line width=0.75]    (10.93,-3.29) .. controls (6.95,-1.4) and (3.31,-0.3) .. (0,0) .. controls (3.31,0.3) and (6.95,1.4) .. (10.93,3.29)   ;
\draw    (458,194) .. controls (491.66,196.97) and (503.76,168.58) .. (459.36,178.68) ;
\draw [shift={(458,179)}, rotate = 346.55] [color={rgb, 255:red, 0; green, 0; blue, 0 }  ][line width=0.75]    (10.93,-3.29) .. controls (6.95,-1.4) and (3.31,-0.3) .. (0,0) .. controls (3.31,0.3) and (6.95,1.4) .. (10.93,3.29)   ;
\draw    (348,148) .. controls (354.93,170.77) and (381.46,190.6) .. (428.57,190.02) ;
\draw [shift={(430,190)}, rotate = 178.81] [color={rgb, 255:red, 0; green, 0; blue, 0 }  ][line width=0.75]    (10.93,-3.29) .. controls (6.95,-1.4) and (3.31,-0.3) .. (0,0) .. controls (3.31,0.3) and (6.95,1.4) .. (10.93,3.29)   ;
\draw    (151,248) -- (228.11,220.67) ;
\draw [shift={(230,220)}, rotate = 160.48] [color={rgb, 255:red, 0; green, 0; blue, 0 }  ][line width=0.75]    (10.93,-3.29) .. controls (6.95,-1.4) and (3.31,-0.3) .. (0,0) .. controls (3.31,0.3) and (6.95,1.4) .. (10.93,3.29)   ;

\draw (461,82) node [anchor=north west][inner sep=0.75pt]   [align=left] {a};
\draw (371,88) node [anchor=north west][inner sep=0.75pt]   [align=left] {a};
\draw (278,155) node [anchor=north west][inner sep=0.75pt]   [align=left] {a};
\draw (324,257) node [anchor=north west][inner sep=0.75pt]   [align=left] {a};
\draw (393,139) node [anchor=north west][inner sep=0.75pt]   [align=left] {a};
\draw (441.5,227) node [anchor=north west][inner sep=0.75pt]   [align=left] {a,b};
\draw (444,302) node [anchor=north west][inner sep=0.75pt]   [align=left] {a};
\draw (103,211) node [anchor=north west][inner sep=0.75pt]   [align=left] {a};
\draw (182,276) node [anchor=north west][inner sep=0.75pt]   [align=left] {a};
\draw (109,301) node [anchor=north west][inner sep=0.75pt]   [align=left] {a,b};
\draw (339,353) node [anchor=north west][inner sep=0.75pt]   [align=left] {a};
\draw (255,352) node [anchor=north west][inner sep=0.75pt]   [align=left] {a};
\draw (292,414) node [anchor=north west][inner sep=0.75pt]   [align=left] {a};
\draw (223.5,445.75) node [anchor=north west][inner sep=0.75pt]   [align=left] {a};
\draw (135.5,359.5) node [anchor=north west][inner sep=0.75pt]   [align=left] {a};
\draw (227,307) node [anchor=north west][inner sep=0.75pt]   [align=left] {b};
\draw (184,373) node [anchor=north west][inner sep=0.75pt]   [align=left] {b};
\draw (403.5,377.5) node [anchor=north west][inner sep=0.75pt]   [align=left] {b};
\draw (253,256) node [anchor=north west][inner sep=0.75pt]   [align=left] {b};
\draw (492,175) node [anchor=north west][inner sep=0.75pt]   [align=left] {b};
\draw (359,180) node [anchor=north west][inner sep=0.75pt]   [align=left] {b};
\draw (179,213.75) node [anchor=north west][inner sep=0.75pt]   [align=left] {b};
\draw (99,179) node [anchor=north west][inner sep=0.75pt]   [align=left] {1};
\draw (137,244) node [anchor=north west][inner sep=0.75pt]   [align=left] {2};
\draw (163,326) node [anchor=north west][inner sep=0.75pt]   [align=left] {3};
\draw (129,394) node [anchor=north west][inner sep=0.75pt]   [align=left] {4};
\draw (191,465) node [anchor=north west][inner sep=0.75pt]   [align=left] {5};
\draw (243,403) node [anchor=north west][inner sep=0.75pt]   [align=left] {6};
\draw (297,323) node [anchor=north west][inner sep=0.75pt]   [align=left] {7};
\draw (350,404) node [anchor=north west][inner sep=0.75pt]   [align=left] {9};
\draw (238,211) node [anchor=north west][inner sep=0.75pt]   [align=left] {8};
\draw (418,265) node [anchor=north west][inner sep=0.75pt]   [align=left] {10};
\draw (436,340) node [anchor=north west][inner sep=0.75pt]   [align=left] {12};
\draw (334,127) node [anchor=north west][inner sep=0.75pt]   [align=left] {11};
\draw (436,179) node [anchor=north west][inner sep=0.75pt]   [align=left] {13};
\draw (410,79) node [anchor=north west][inner sep=0.75pt]   [align=left] {14};
\draw (474,113) node [anchor=north west][inner sep=0.75pt]   [align=left] {15};

\end{tikzpicture}

%% file: first_step.tex
\tikzset{every picture/.style={line width=0.75pt}} 

\begin{tikzpicture}[x=0.75pt,y=0.75pt,yscale=-0.9,xscale=0.9]

\draw    (74,162) .. controls (82.6,167.73) and (86.53,180.89) .. (80.48,193.39) ;
\draw [shift={(79.56,195.15)}, rotate = 299.43] [color={rgb, 255:red, 0; green, 0; blue, 0 }  ][line width=0.75]    (10.93,-3.29) .. controls (6.95,-1.4) and (3.31,-0.3) .. (0,0) .. controls (3.31,0.3) and (6.95,1.4) .. (10.93,3.29)   ;
\draw    (69.12,197.97) .. controls (59.55,195.78) and (52.44,183.75) .. (56.58,163.69) ;
\draw [shift={(57,161.81)}, rotate = 103.27] [color={rgb, 255:red, 0; green, 0; blue, 0 }  ][line width=0.75]    (10.93,-3.29) .. controls (6.95,-1.4) and (3.31,-0.3) .. (0,0) .. controls (3.31,0.3) and (6.95,1.4) .. (10.93,3.29)   ;
\draw    (134,246) -- (173,246.95) ;
\draw [shift={(175,247)}, rotate = 181.4] [color={rgb, 255:red, 0; green, 0; blue, 0 }  ][line width=0.75]    (10.93,-3.29) .. controls (6.95,-1.4) and (3.31,-0.3) .. (0,0) .. controls (3.31,0.3) and (6.95,1.4) .. (10.93,3.29)   ;
\draw    (146.68,205.32) -- (125.28,235.34) ;
\draw [shift={(124.12,236.97)}, rotate = 305.49] [color={rgb, 255:red, 0; green, 0; blue, 0 }  ][line width=0.75]    (10.93,-3.29) .. controls (6.95,-1.4) and (3.31,-0.3) .. (0,0) .. controls (3.31,0.3) and (6.95,1.4) .. (10.93,3.29)   ;
\draw    (179,241) -- (153.2,206.6) ;
\draw [shift={(152,205)}, rotate = 53.13] [color={rgb, 255:red, 0; green, 0; blue, 0 }  ][line width=0.75]    (10.93,-3.29) .. controls (6.95,-1.4) and (3.31,-0.3) .. (0,0) .. controls (3.31,0.3) and (6.95,1.4) .. (10.93,3.29)   ;
\draw    (221.7,160.11) -- (229.69,127.02) ;
\draw [shift={(230.16,125.07)}, rotate = 103.58] [color={rgb, 255:red, 0; green, 0; blue, 0 }  ][line width=0.75]    (10.93,-3.29) .. controls (6.95,-1.4) and (3.31,-0.3) .. (0,0) .. controls (3.31,0.3) and (6.95,1.4) .. (10.93,3.29)   ;
\draw    (221.13,113.21) -- (183.38,91.05) ;
\draw [shift={(181.65,90.04)}, rotate = 30.41] [color={rgb, 255:red, 0; green, 0; blue, 0 }  ][line width=0.75]    (10.93,-3.29) .. controls (6.95,-1.4) and (3.31,-0.3) .. (0,0) .. controls (3.31,0.3) and (6.95,1.4) .. (10.93,3.29)   ;
\draw    (166.99,95.12) -- (128.54,126.73) ;
\draw [shift={(127,128)}, rotate = 320.57] [color={rgb, 255:red, 0; green, 0; blue, 0 }  ][line width=0.75]    (10.93,-3.29) .. controls (6.95,-1.4) and (3.31,-0.3) .. (0,0) .. controls (3.31,0.3) and (6.95,1.4) .. (10.93,3.29)   ;
\draw    (129,137) -- (209.12,168.98) ;
\draw [shift={(210.98,169.72)}, rotate = 201.76] [color={rgb, 255:red, 0; green, 0; blue, 0 }  ][line width=0.75]    (10.93,-3.29) .. controls (6.95,-1.4) and (3.31,-0.3) .. (0,0) .. controls (3.31,0.3) and (6.95,1.4) .. (10.93,3.29)   ;
\draw    (43,123) -- (58.06,151.24) ;
\draw [shift={(59,153)}, rotate = 241.93] [color={rgb, 255:red, 0; green, 0; blue, 0 }  ][line width=0.75]    (10.93,-3.29) .. controls (6.95,-1.4) and (3.31,-0.3) .. (0,0) .. controls (3.31,0.3) and (6.95,1.4) .. (10.93,3.29)   ;
\draw    (58.69,232.44) -- (70.18,210.2) ;
\draw [shift={(71.1,208.43)}, rotate = 117.32] [color={rgb, 255:red, 0; green, 0; blue, 0 }  ][line width=0.75]    (10.93,-3.29) .. controls (6.95,-1.4) and (3.31,-0.3) .. (0,0) .. controls (3.31,0.3) and (6.95,1.4) .. (10.93,3.29)   ;
\draw    (93.66,273.42) -- (114.29,251.42) ;
\draw [shift={(115.66,249.96)}, rotate = 133.17] [color={rgb, 255:red, 0; green, 0; blue, 0 }  ][line width=0.75]    (10.93,-3.29) .. controls (6.95,-1.4) and (3.31,-0.3) .. (0,0) .. controls (3.31,0.3) and (6.95,1.4) .. (10.93,3.29)   ;
\draw    (249,80) -- (228.51,62.31) ;
\draw [shift={(227,61)}, rotate = 40.82] [color={rgb, 255:red, 0; green, 0; blue, 0 }  ][line width=0.75]    (10.93,-3.29) .. controls (6.95,-1.4) and (3.31,-0.3) .. (0,0) .. controls (3.31,0.3) and (6.95,1.4) .. (10.93,3.29)   ;
\draw    (206,63) -- (181.27,81.76) ;
\draw [shift={(179.68,82.97)}, rotate = 322.81] [color={rgb, 255:red, 0; green, 0; blue, 0 }  ][line width=0.75]    (10.93,-3.29) .. controls (6.95,-1.4) and (3.31,-0.3) .. (0,0) .. controls (3.31,0.3) and (6.95,1.4) .. (10.93,3.29)   ;
\draw    (230.72,201.93) -- (223.88,175.05) ;
\draw [shift={(223.39,173.11)}, rotate = 75.73] [color={rgb, 255:red, 0; green, 0; blue, 0 }  ][line width=0.75]    (10.93,-3.29) .. controls (6.95,-1.4) and (3.31,-0.3) .. (0,0) .. controls (3.31,0.3) and (6.95,1.4) .. (10.93,3.29)   ;
\draw    (141.04,198.54) -- (89,199.95) ;
\draw [shift={(87,200)}, rotate = 358.45] [color={rgb, 255:red, 0; green, 0; blue, 0 }  ][line width=0.75]    (10.93,-3.29) .. controls (6.95,-1.4) and (3.31,-0.3) .. (0,0) .. controls (3.31,0.3) and (6.95,1.4) .. (10.93,3.29)   ;
\draw    (114.53,243.18) -- (82.38,209.45) ;
\draw [shift={(81,208)}, rotate = 46.38] [color={rgb, 255:red, 0; green, 0; blue, 0 }  ][line width=0.75]    (10.93,-3.29) .. controls (6.95,-1.4) and (3.31,-0.3) .. (0,0) .. controls (3.31,0.3) and (6.95,1.4) .. (10.93,3.29)   ;
\draw    (187.01,242.05) -- (223.5,213.89) ;
\draw [shift={(225.08,212.67)}, rotate = 142.34] [color={rgb, 255:red, 0; green, 0; blue, 0 }  ][line width=0.75]    (10.93,-3.29) .. controls (6.95,-1.4) and (3.31,-0.3) .. (0,0) .. controls (3.31,0.3) and (6.95,1.4) .. (10.93,3.29)   ;
\draw    (123,142) -- (145.16,190.18) ;
\draw [shift={(146,192)}, rotate = 245.3] [color={rgb, 255:red, 0; green, 0; blue, 0 }  ][line width=0.75]    (10.93,-3.29) .. controls (6.95,-1.4) and (3.31,-0.3) .. (0,0) .. controls (3.31,0.3) and (6.95,1.4) .. (10.93,3.29)   ;
\draw    (246.05,121.12) .. controls (264.75,122.77) and (271.65,107.23) .. (247.94,112.22) ;
\draw [shift={(246.05,112.64)}, rotate = 346.53] [color={rgb, 255:red, 0; green, 0; blue, 0 }  ][line width=0.75]    (10.93,-3.29) .. controls (6.95,-1.4) and (3.31,-0.3) .. (0,0) .. controls (3.31,0.3) and (6.95,1.4) .. (10.93,3.29)   ;
\draw    (176.01,95.12) .. controls (179.88,107.86) and (194.58,118.97) .. (220.65,118.88) ;
\draw [shift={(222.26,118.86)}, rotate = 178.8] [color={rgb, 255:red, 0; green, 0; blue, 0 }  ][line width=0.75]    (10.93,-3.29) .. controls (6.95,-1.4) and (3.31,-0.3) .. (0,0) .. controls (3.31,0.3) and (6.95,1.4) .. (10.93,3.29)   ;
\draw    (64.89,151.63) -- (107.57,136.48) ;
\draw [shift={(109.45,135.81)}, rotate = 160.45] [color={rgb, 255:red, 0; green, 0; blue, 0 }  ][line width=0.75]    (10.93,-3.29) .. controls (6.95,-1.4) and (3.31,-0.3) .. (0,0) .. controls (3.31,0.3) and (6.95,1.4) .. (10.93,3.29)   ;
\draw    (342,166) .. controls (350.6,171.73) and (354.53,184.89) .. (348.48,197.39) ;
\draw [shift={(347.56,199.15)}, rotate = 299.43] [color={rgb, 255:red, 0; green, 0; blue, 0 }  ][line width=0.75]    (10.93,-3.29) .. controls (6.95,-1.4) and (3.31,-0.3) .. (0,0) .. controls (3.31,0.3) and (6.95,1.4) .. (10.93,3.29)   ;
\draw    (337.12,201.97) .. controls (327.55,199.78) and (320.44,187.75) .. (324.58,167.69) ;
\draw [shift={(325,165.81)}, rotate = 103.27] [color={rgb, 255:red, 0; green, 0; blue, 0 }  ][line width=0.75]    (10.93,-3.29) .. controls (6.95,-1.4) and (3.31,-0.3) .. (0,0) .. controls (3.31,0.3) and (6.95,1.4) .. (10.93,3.29)   ;
\draw    (402,250) -- (441,250.95) ;
\draw [shift={(443,251)}, rotate = 181.4] [color={rgb, 255:red, 0; green, 0; blue, 0 }  ][line width=0.75]    (10.93,-3.29) .. controls (6.95,-1.4) and (3.31,-0.3) .. (0,0) .. controls (3.31,0.3) and (6.95,1.4) .. (10.93,3.29)   ;
\draw    (414.68,209.32) -- (393.28,239.34) ;
\draw [shift={(392.12,240.97)}, rotate = 305.49] [color={rgb, 255:red, 0; green, 0; blue, 0 }  ][line width=0.75]    (10.93,-3.29) .. controls (6.95,-1.4) and (3.31,-0.3) .. (0,0) .. controls (3.31,0.3) and (6.95,1.4) .. (10.93,3.29)   ;
\draw    (447,245) -- (421.2,210.6) ;
\draw [shift={(420,209)}, rotate = 53.13] [color={rgb, 255:red, 0; green, 0; blue, 0 }  ][line width=0.75]    (10.93,-3.29) .. controls (6.95,-1.4) and (3.31,-0.3) .. (0,0) .. controls (3.31,0.3) and (6.95,1.4) .. (10.93,3.29)   ;
\draw    (489.7,164.11) -- (497.69,131.02) ;
\draw [shift={(498.16,129.07)}, rotate = 103.58] [color={rgb, 255:red, 0; green, 0; blue, 0 }  ][line width=0.75]    (10.93,-3.29) .. controls (6.95,-1.4) and (3.31,-0.3) .. (0,0) .. controls (3.31,0.3) and (6.95,1.4) .. (10.93,3.29)   ;
\draw    (489.13,117.21) -- (451.38,95.05) ;
\draw [shift={(449.65,94.04)}, rotate = 30.41] [color={rgb, 255:red, 0; green, 0; blue, 0 }  ][line width=0.75]    (10.93,-3.29) .. controls (6.95,-1.4) and (3.31,-0.3) .. (0,0) .. controls (3.31,0.3) and (6.95,1.4) .. (10.93,3.29)   ;
\draw    (434.99,99.12) -- (396.54,130.73) ;
\draw [shift={(395,132)}, rotate = 320.57] [color={rgb, 255:red, 0; green, 0; blue, 0 }  ][line width=0.75]    (10.93,-3.29) .. controls (6.95,-1.4) and (3.31,-0.3) .. (0,0) .. controls (3.31,0.3) and (6.95,1.4) .. (10.93,3.29)   ;
\draw    (397,141) -- (477.12,172.98) ;
\draw [shift={(478.98,173.72)}, rotate = 201.76] [color={rgb, 255:red, 0; green, 0; blue, 0 }  ][line width=0.75]    (10.93,-3.29) .. controls (6.95,-1.4) and (3.31,-0.3) .. (0,0) .. controls (3.31,0.3) and (6.95,1.4) .. (10.93,3.29)   ;
\draw    (311,127) -- (326.06,155.24) ;
\draw [shift={(327,157)}, rotate = 241.93] [color={rgb, 255:red, 0; green, 0; blue, 0 }  ][line width=0.75]    (10.93,-3.29) .. controls (6.95,-1.4) and (3.31,-0.3) .. (0,0) .. controls (3.31,0.3) and (6.95,1.4) .. (10.93,3.29)   ;
\draw    (326.69,236.44) -- (338.18,214.2) ;
\draw [shift={(339.1,212.43)}, rotate = 117.32] [color={rgb, 255:red, 0; green, 0; blue, 0 }  ][line width=0.75]    (10.93,-3.29) .. controls (6.95,-1.4) and (3.31,-0.3) .. (0,0) .. controls (3.31,0.3) and (6.95,1.4) .. (10.93,3.29)   ;
\draw    (361.66,277.42) -- (382.29,255.42) ;
\draw [shift={(383.66,253.96)}, rotate = 133.17] [color={rgb, 255:red, 0; green, 0; blue, 0 }  ][line width=0.75]    (10.93,-3.29) .. controls (6.95,-1.4) and (3.31,-0.3) .. (0,0) .. controls (3.31,0.3) and (6.95,1.4) .. (10.93,3.29)   ;
\draw    (517,84) -- (496.51,66.31) ;
\draw [shift={(495,65)}, rotate = 40.82] [color={rgb, 255:red, 0; green, 0; blue, 0 }  ][line width=0.75]    (10.93,-3.29) .. controls (6.95,-1.4) and (3.31,-0.3) .. (0,0) .. controls (3.31,0.3) and (6.95,1.4) .. (10.93,3.29)   ;
\draw    (474,67) -- (449.27,85.76) ;
\draw [shift={(447.68,86.97)}, rotate = 322.81] [color={rgb, 255:red, 0; green, 0; blue, 0 }  ][line width=0.75]    (10.93,-3.29) .. controls (6.95,-1.4) and (3.31,-0.3) .. (0,0) .. controls (3.31,0.3) and (6.95,1.4) .. (10.93,3.29)   ;
\draw    (498.72,205.93) -- (491.88,179.05) ;
\draw [shift={(491.39,177.11)}, rotate = 75.73] [color={rgb, 255:red, 0; green, 0; blue, 0 }  ][line width=0.75]    (10.93,-3.29) .. controls (6.95,-1.4) and (3.31,-0.3) .. (0,0) .. controls (3.31,0.3) and (6.95,1.4) .. (10.93,3.29)   ;
\draw    (409.04,202.54) -- (357,203.95) ;
\draw [shift={(355,204)}, rotate = 358.45] [color={rgb, 255:red, 0; green, 0; blue, 0 }  ][line width=0.75]    (10.93,-3.29) .. controls (6.95,-1.4) and (3.31,-0.3) .. (0,0) .. controls (3.31,0.3) and (6.95,1.4) .. (10.93,3.29)   ;
\draw    (382.53,247.18) -- (350.38,213.45) ;
\draw [shift={(349,212)}, rotate = 46.38] [color={rgb, 255:red, 0; green, 0; blue, 0 }  ][line width=0.75]    (10.93,-3.29) .. controls (6.95,-1.4) and (3.31,-0.3) .. (0,0) .. controls (3.31,0.3) and (6.95,1.4) .. (10.93,3.29)   ;
\draw    (455.01,246.05) -- (491.5,217.89) ;
\draw [shift={(493.08,216.67)}, rotate = 142.34] [color={rgb, 255:red, 0; green, 0; blue, 0 }  ][line width=0.75]    (10.93,-3.29) .. controls (6.95,-1.4) and (3.31,-0.3) .. (0,0) .. controls (3.31,0.3) and (6.95,1.4) .. (10.93,3.29)   ;
\draw    (391,146) -- (413.16,194.18) ;
\draw [shift={(414,196)}, rotate = 245.3] [color={rgb, 255:red, 0; green, 0; blue, 0 }  ][line width=0.75]    (10.93,-3.29) .. controls (6.95,-1.4) and (3.31,-0.3) .. (0,0) .. controls (3.31,0.3) and (6.95,1.4) .. (10.93,3.29)   ;
\draw    (514.05,125.12) .. controls (532.75,126.77) and (539.65,111.23) .. (515.94,116.22) ;
\draw [shift={(514.05,116.64)}, rotate = 346.53] [color={rgb, 255:red, 0; green, 0; blue, 0 }  ][line width=0.75]    (10.93,-3.29) .. controls (6.95,-1.4) and (3.31,-0.3) .. (0,0) .. controls (3.31,0.3) and (6.95,1.4) .. (10.93,3.29)   ;
\draw    (444.01,99.12) .. controls (447.88,111.86) and (462.58,122.97) .. (488.65,122.88) ;
\draw [shift={(490.26,122.86)}, rotate = 178.8] [color={rgb, 255:red, 0; green, 0; blue, 0 }  ][line width=0.75]    (10.93,-3.29) .. controls (6.95,-1.4) and (3.31,-0.3) .. (0,0) .. controls (3.31,0.3) and (6.95,1.4) .. (10.93,3.29)   ;
\draw    (332.89,155.63) -- (375.57,140.48) ;
\draw [shift={(377.45,139.81)}, rotate = 160.45] [color={rgb, 255:red, 0; green, 0; blue, 0 }  ][line width=0.75]    (10.93,-3.29) .. controls (6.95,-1.4) and (3.31,-0.3) .. (0,0) .. controls (3.31,0.3) and (6.95,1.4) .. (10.93,3.29)   ;
\draw    (113,400) .. controls (121.6,405.73) and (125.53,418.89) .. (119.48,431.39) ;
\draw [shift={(118.56,433.15)}, rotate = 299.43] [color={rgb, 255:red, 0; green, 0; blue, 0 }  ][line width=0.75]    (10.93,-3.29) .. controls (6.95,-1.4) and (3.31,-0.3) .. (0,0) .. controls (3.31,0.3) and (6.95,1.4) .. (10.93,3.29)   ;
\draw    (108.12,435.97) .. controls (98.55,433.78) and (91.44,421.75) .. (95.58,401.69) ;
\draw [shift={(96,399.81)}, rotate = 103.27] [color={rgb, 255:red, 0; green, 0; blue, 0 }  ][line width=0.75]    (10.93,-3.29) .. controls (6.95,-1.4) and (3.31,-0.3) .. (0,0) .. controls (3.31,0.3) and (6.95,1.4) .. (10.93,3.29)   ;
\draw    (173,484) -- (212,484.95) ;
\draw [shift={(214,485)}, rotate = 181.4] [color={rgb, 255:red, 0; green, 0; blue, 0 }  ][line width=0.75]    (10.93,-3.29) .. controls (6.95,-1.4) and (3.31,-0.3) .. (0,0) .. controls (3.31,0.3) and (6.95,1.4) .. (10.93,3.29)   ;
\draw    (185.68,443.32) -- (164.28,473.34) ;
\draw [shift={(163.12,474.97)}, rotate = 305.49] [color={rgb, 255:red, 0; green, 0; blue, 0 }  ][line width=0.75]    (10.93,-3.29) .. controls (6.95,-1.4) and (3.31,-0.3) .. (0,0) .. controls (3.31,0.3) and (6.95,1.4) .. (10.93,3.29)   ;
\draw    (218,479) -- (192.2,444.6) ;
\draw [shift={(191,443)}, rotate = 53.13] [color={rgb, 255:red, 0; green, 0; blue, 0 }  ][line width=0.75]    (10.93,-3.29) .. controls (6.95,-1.4) and (3.31,-0.3) .. (0,0) .. controls (3.31,0.3) and (6.95,1.4) .. (10.93,3.29)   ;
\draw    (260.7,398.11) -- (268.69,365.02) ;
\draw [shift={(269.16,363.07)}, rotate = 103.58] [color={rgb, 255:red, 0; green, 0; blue, 0 }  ][line width=0.75]    (10.93,-3.29) .. controls (6.95,-1.4) and (3.31,-0.3) .. (0,0) .. controls (3.31,0.3) and (6.95,1.4) .. (10.93,3.29)   ;
\draw    (260.13,351.21) -- (222.38,329.05) ;
\draw [shift={(220.65,328.04)}, rotate = 30.41] [color={rgb, 255:red, 0; green, 0; blue, 0 }  ][line width=0.75]    (10.93,-3.29) .. controls (6.95,-1.4) and (3.31,-0.3) .. (0,0) .. controls (3.31,0.3) and (6.95,1.4) .. (10.93,3.29)   ;
\draw    (205.99,333.12) -- (167.54,364.73) ;
\draw [shift={(166,366)}, rotate = 320.57] [color={rgb, 255:red, 0; green, 0; blue, 0 }  ][line width=0.75]    (10.93,-3.29) .. controls (6.95,-1.4) and (3.31,-0.3) .. (0,0) .. controls (3.31,0.3) and (6.95,1.4) .. (10.93,3.29)   ;
\draw    (168,375) -- (248.12,406.98) ;
\draw [shift={(249.98,407.72)}, rotate = 201.76] [color={rgb, 255:red, 0; green, 0; blue, 0 }  ][line width=0.75]    (10.93,-3.29) .. controls (6.95,-1.4) and (3.31,-0.3) .. (0,0) .. controls (3.31,0.3) and (6.95,1.4) .. (10.93,3.29)   ;
\draw    (82,361) -- (97.06,389.24) ;
\draw [shift={(98,391)}, rotate = 241.93] [color={rgb, 255:red, 0; green, 0; blue, 0 }  ][line width=0.75]    (10.93,-3.29) .. controls (6.95,-1.4) and (3.31,-0.3) .. (0,0) .. controls (3.31,0.3) and (6.95,1.4) .. (10.93,3.29)   ;
\draw    (97.69,470.44) -- (109.18,448.2) ;
\draw [shift={(110.1,446.43)}, rotate = 117.32] [color={rgb, 255:red, 0; green, 0; blue, 0 }  ][line width=0.75]    (10.93,-3.29) .. controls (6.95,-1.4) and (3.31,-0.3) .. (0,0) .. controls (3.31,0.3) and (6.95,1.4) .. (10.93,3.29)   ;
\draw    (132.66,511.42) -- (153.29,489.42) ;
\draw [shift={(154.66,487.96)}, rotate = 133.17] [color={rgb, 255:red, 0; green, 0; blue, 0 }  ][line width=0.75]    (10.93,-3.29) .. controls (6.95,-1.4) and (3.31,-0.3) .. (0,0) .. controls (3.31,0.3) and (6.95,1.4) .. (10.93,3.29)   ;
\draw    (288,318) -- (267.51,300.31) ;
\draw [shift={(266,299)}, rotate = 40.82] [color={rgb, 255:red, 0; green, 0; blue, 0 }  ][line width=0.75]    (10.93,-3.29) .. controls (6.95,-1.4) and (3.31,-0.3) .. (0,0) .. controls (3.31,0.3) and (6.95,1.4) .. (10.93,3.29)   ;
\draw    (245,301) -- (220.27,319.76) ;
\draw [shift={(218.68,320.97)}, rotate = 322.81] [color={rgb, 255:red, 0; green, 0; blue, 0 }  ][line width=0.75]    (10.93,-3.29) .. controls (6.95,-1.4) and (3.31,-0.3) .. (0,0) .. controls (3.31,0.3) and (6.95,1.4) .. (10.93,3.29)   ;
\draw    (269.72,439.93) -- (262.88,413.05) ;
\draw [shift={(262.39,411.11)}, rotate = 75.73] [color={rgb, 255:red, 0; green, 0; blue, 0 }  ][line width=0.75]    (10.93,-3.29) .. controls (6.95,-1.4) and (3.31,-0.3) .. (0,0) .. controls (3.31,0.3) and (6.95,1.4) .. (10.93,3.29)   ;
\draw    (180.04,436.54) -- (128,437.95) ;
\draw [shift={(126,438)}, rotate = 358.45] [color={rgb, 255:red, 0; green, 0; blue, 0 }  ][line width=0.75]    (10.93,-3.29) .. controls (6.95,-1.4) and (3.31,-0.3) .. (0,0) .. controls (3.31,0.3) and (6.95,1.4) .. (10.93,3.29)   ;
\draw    (153.53,481.18) -- (121.38,447.45) ;
\draw [shift={(120,446)}, rotate = 46.38] [color={rgb, 255:red, 0; green, 0; blue, 0 }  ][line width=0.75]    (10.93,-3.29) .. controls (6.95,-1.4) and (3.31,-0.3) .. (0,0) .. controls (3.31,0.3) and (6.95,1.4) .. (10.93,3.29)   ;
\draw    (226.01,480.05) -- (262.5,451.89) ;
\draw [shift={(264.08,450.67)}, rotate = 142.34] [color={rgb, 255:red, 0; green, 0; blue, 0 }  ][line width=0.75]    (10.93,-3.29) .. controls (6.95,-1.4) and (3.31,-0.3) .. (0,0) .. controls (3.31,0.3) and (6.95,1.4) .. (10.93,3.29)   ;
\draw    (162,380) -- (184.16,428.18) ;
\draw [shift={(185,430)}, rotate = 245.3] [color={rgb, 255:red, 0; green, 0; blue, 0 }  ][line width=0.75]    (10.93,-3.29) .. controls (6.95,-1.4) and (3.31,-0.3) .. (0,0) .. controls (3.31,0.3) and (6.95,1.4) .. (10.93,3.29)   ;
\draw    (285.05,359.12) .. controls (303.75,360.77) and (310.65,345.23) .. (286.94,350.22) ;
\draw [shift={(285.05,350.64)}, rotate = 346.53] [color={rgb, 255:red, 0; green, 0; blue, 0 }  ][line width=0.75]    (10.93,-3.29) .. controls (6.95,-1.4) and (3.31,-0.3) .. (0,0) .. controls (3.31,0.3) and (6.95,1.4) .. (10.93,3.29)   ;
\draw    (215.01,333.12) .. controls (218.88,345.86) and (233.58,356.97) .. (259.65,356.88) ;
\draw [shift={(261.26,356.86)}, rotate = 178.8] [color={rgb, 255:red, 0; green, 0; blue, 0 }  ][line width=0.75]    (10.93,-3.29) .. controls (6.95,-1.4) and (3.31,-0.3) .. (0,0) .. controls (3.31,0.3) and (6.95,1.4) .. (10.93,3.29)   ;
\draw    (103.89,389.63) -- (146.57,374.48) ;
\draw [shift={(148.45,373.81)}, rotate = 160.45] [color={rgb, 255:red, 0; green, 0; blue, 0 }  ][line width=0.75]    (10.93,-3.29) .. controls (6.95,-1.4) and (3.31,-0.3) .. (0,0) .. controls (3.31,0.3) and (6.95,1.4) .. (10.93,3.29)   ;

\draw (241.35,62.13) node [anchor=north west][inner sep=0.75pt]  [font=\scriptsize] [align=left] {a};
\draw (189.58,59.52) node [anchor=north west][inner sep=0.75pt]  [font=\scriptsize] [align=left] {a};
\draw (135.13,97.38) node [anchor=north west][inner sep=0.75pt]  [font=\scriptsize] [align=left] {a};
\draw (160.08,153.02) node [anchor=north west][inner sep=0.75pt]  [font=\scriptsize] [align=left] {a};
\draw (201.99,88.34) node [anchor=north west][inner sep=0.75pt]  [font=\scriptsize] [align=left] {a};
\draw (227.93,145.59) node [anchor=north west][inner sep=0.75pt]  [font=\scriptsize] [align=left] {a,b};
\draw (231.76,184.45) node [anchor=north west][inner sep=0.75pt]  [font=\scriptsize] [align=left] {a};
\draw (36.42,130.03) node [anchor=north west][inner sep=0.75pt]  [font=\scriptsize] [align=left] {a};
\draw (84.98,168.76) node [anchor=north west][inner sep=0.75pt]  [font=\scriptsize] [align=left] {a};
\draw (36.19,177.89) node [anchor=north west][inner sep=0.75pt]  [font=\scriptsize] [align=left] {a,b};
\draw (168.54,215.27) node [anchor=north west][inner sep=0.75pt]  [font=\scriptsize] [align=left] {a};
\draw (126.16,209.71) node [anchor=north west][inner sep=0.75pt]  [font=\scriptsize] [align=left] {a};
\draw (145.03,245.75) node [anchor=north west][inner sep=0.75pt]  [font=\scriptsize] [align=left] {a};
\draw (106.66,264.69) node [anchor=north west][inner sep=0.75pt]  [font=\scriptsize] [align=left] {a};
\draw (52.75,213.95) node [anchor=north west][inner sep=0.75pt]  [font=\scriptsize] [align=left] {a};
\draw (108.36,184.28) node [anchor=north west][inner sep=0.75pt]  [font=\scriptsize] [align=left] {b};
\draw (87.11,225.58) node [anchor=north west][inner sep=0.75pt]  [font=\scriptsize] [align=left] {b};
\draw (208.04,230.36) node [anchor=north west][inner sep=0.75pt]  [font=\scriptsize] [align=left] {b};
\draw (121.03,159.46) node [anchor=north west][inner sep=0.75pt]  [font=\scriptsize] [align=left] {b};
\draw (265.83,109.68) node [anchor=north west][inner sep=0.75pt]  [font=\scriptsize] [align=left] {b};
\draw (179.82,109.51) node [anchor=north west][inner sep=0.75pt]  [font=\scriptsize] [align=left] {b};
\draw (78.29,128.58) node [anchor=north west][inner sep=0.75pt]  [font=\scriptsize] [align=left] {b};
\draw (33.17,108.94) node [anchor=north west][inner sep=0.75pt]  [font=\scriptsize] [align=left] {1a};
\draw (58.43,152.37) node [anchor=north west][inner sep=0.75pt]  [font=\scriptsize] [align=left] {2a};
\draw (70.26,194.02) node [anchor=north west][inner sep=0.75pt]  [font=\scriptsize] [align=left] {3a};
\draw (51.09,232.44) node [anchor=north west][inner sep=0.75pt]  [font=\scriptsize] [align=left] {4a};
\draw (85.06,272.57) node [anchor=north west][inner sep=0.75pt]  [font=\scriptsize] [align=left] {5a};
\draw (117.39,239.53) node [anchor=north west][inner sep=0.75pt]  [font=\scriptsize] [align=left] {6a};
\draw (143.81,193.89) node [anchor=north west][inner sep=0.75pt]  [font=\scriptsize] [align=left] {7a};
\draw (178.57,242.23) node [anchor=north west][inner sep=0.75pt]  [font=\scriptsize] [align=left] {9a};
\draw (111.57,127.03) node [anchor=north west][inner sep=0.75pt]  [font=\scriptsize] [align=left] {8a};
\draw (211.35,161.54) node [anchor=north west][inner sep=0.75pt]  [font=\scriptsize] [align=left] {10a};
\draw (222.5,201.93) node [anchor=north west][inner sep=0.75pt]  [font=\scriptsize] [align=left] {12a};
\draw (161.19,80.56) node [anchor=north west][inner sep=0.75pt]  [font=\scriptsize] [align=left] {11a};
\draw (224.13,111.21) node [anchor=north west][inner sep=0.75pt]  [font=\scriptsize] [align=left] {13a};
\draw (204.84,50.43) node [anchor=north west][inner sep=0.75pt]  [font=\scriptsize] [align=left] {14a};
\draw (249.08,80.04) node [anchor=north west][inner sep=0.75pt]  [font=\scriptsize] [align=left] {15a};
\draw (509.35,67.13) node [anchor=north west][inner sep=0.75pt]  [font=\scriptsize] [align=left] {a};
\draw (457.58,63.52) node [anchor=north west][inner sep=0.75pt]  [font=\scriptsize] [align=left] {a};
\draw (403.13,101.38) node [anchor=north west][inner sep=0.75pt]  [font=\scriptsize] [align=left] {a};
\draw (428.08,157.02) node [anchor=north west][inner sep=0.75pt]  [font=\scriptsize] [align=left] {a};
\draw (469.99,92.34) node [anchor=north west][inner sep=0.75pt]  [font=\scriptsize] [align=left] {a};
\draw (495.93,149.59) node [anchor=north west][inner sep=0.75pt]  [font=\scriptsize] [align=left] {a,b};
\draw (499.76,188.45) node [anchor=north west][inner sep=0.75pt]  [font=\scriptsize] [align=left] {a};
\draw (304.42,134.03) node [anchor=north west][inner sep=0.75pt]  [font=\scriptsize] [align=left] {a};
\draw (352.98,172.76) node [anchor=north west][inner sep=0.75pt]  [font=\scriptsize] [align=left] {a};
\draw (304.19,181.89) node [anchor=north west][inner sep=0.75pt]  [font=\scriptsize] [align=left] {a,b};
\draw (436.54,219.27) node [anchor=north west][inner sep=0.75pt]  [font=\scriptsize] [align=left] {a};
\draw (394.16,213.71) node [anchor=north west][inner sep=0.75pt]  [font=\scriptsize] [align=left] {a};
\draw (413.03,249.75) node [anchor=north west][inner sep=0.75pt]  [font=\scriptsize] [align=left] {a};
\draw (374.66,268.69) node [anchor=north west][inner sep=0.75pt]  [font=\scriptsize] [align=left] {a};
\draw (320.75,217.95) node [anchor=north west][inner sep=0.75pt]  [font=\scriptsize] [align=left] {a};
\draw (376.36,188.28) node [anchor=north west][inner sep=0.75pt]  [font=\scriptsize] [align=left] {b};
\draw (355.11,229.58) node [anchor=north west][inner sep=0.75pt]  [font=\scriptsize] [align=left] {b};
\draw (476.04,234.36) node [anchor=north west][inner sep=0.75pt]  [font=\scriptsize] [align=left] {b};
\draw (389.03,163.46) node [anchor=north west][inner sep=0.75pt]  [font=\scriptsize] [align=left] {b};
\draw (533.83,113.68) node [anchor=north west][inner sep=0.75pt]  [font=\scriptsize] [align=left] {b};
\draw (447.82,113.51) node [anchor=north west][inner sep=0.75pt]  [font=\scriptsize] [align=left] {b};
\draw (346.29,132.58) node [anchor=north west][inner sep=0.75pt]  [font=\scriptsize] [align=left] {b};
\draw (301.17,112.94) node [anchor=north west][inner sep=0.75pt]  [font=\scriptsize] [align=left] {1c};
\draw (326.43,156.37) node [anchor=north west][inner sep=0.75pt]  [font=\scriptsize] [align=left] {2c};
\draw (338.26,198.02) node [anchor=north west][inner sep=0.75pt]  [font=\scriptsize] [align=left] {3c};
\draw (319.09,236.44) node [anchor=north west][inner sep=0.75pt]  [font=\scriptsize] [align=left] {4c};
\draw (353.06,276.57) node [anchor=north west][inner sep=0.75pt]  [font=\scriptsize] [align=left] {5c};
\draw (385.39,243.53) node [anchor=north west][inner sep=0.75pt]  [font=\scriptsize] [align=left] {6c};
\draw (410,197) node [anchor=north west][inner sep=0.75pt]  [font=\scriptsize] [align=left] {7c};
\draw (444,246) node [anchor=north west][inner sep=0.75pt]  [font=\scriptsize] [align=left] {9c};
\draw (379.57,131.03) node [anchor=north west][inner sep=0.75pt]  [font=\scriptsize] [align=left] {8c};
\draw (479.35,165.54) node [anchor=north west][inner sep=0.75pt]  [font=\scriptsize] [align=left] {10c};
\draw (490.5,205.93) node [anchor=north west][inner sep=0.75pt]  [font=\scriptsize] [align=left] {12c};
\draw (429.19,84.56) node [anchor=north west][inner sep=0.75pt]  [font=\scriptsize] [align=left] {11c};
\draw (492.13,115.21) node [anchor=north west][inner sep=0.75pt]  [font=\scriptsize] [align=left] {13c};
\draw (472.84,54.43) node [anchor=north west][inner sep=0.75pt]  [font=\scriptsize] [align=left] {14c};
\draw (517.08,84.04) node [anchor=north west][inner sep=0.75pt]  [font=\scriptsize] [align=left] {15c};
\draw (279.35,299.13) node [anchor=north west][inner sep=0.75pt]  [font=\scriptsize] [align=left] {a};
\draw (228.58,297.52) node [anchor=north west][inner sep=0.75pt]  [font=\scriptsize] [align=left] {a};
\draw (174.13,335.38) node [anchor=north west][inner sep=0.75pt]  [font=\scriptsize] [align=left] {a};
\draw (199.08,391.02) node [anchor=north west][inner sep=0.75pt]  [font=\scriptsize] [align=left] {a};
\draw (240.99,326.34) node [anchor=north west][inner sep=0.75pt]  [font=\scriptsize] [align=left] {a};
\draw (266.93,383.59) node [anchor=north west][inner sep=0.75pt]  [font=\scriptsize] [align=left] {a,b};
\draw (270.76,422.45) node [anchor=north west][inner sep=0.75pt]  [font=\scriptsize] [align=left] {a};
\draw (75.42,368.03) node [anchor=north west][inner sep=0.75pt]  [font=\scriptsize] [align=left] {a};
\draw (123.98,406.76) node [anchor=north west][inner sep=0.75pt]  [font=\scriptsize] [align=left] {a};
\draw (75.19,415.89) node [anchor=north west][inner sep=0.75pt]  [font=\scriptsize] [align=left] {a,b};
\draw (207.54,453.27) node [anchor=north west][inner sep=0.75pt]  [font=\scriptsize] [align=left] {a};
\draw (165.16,447.71) node [anchor=north west][inner sep=0.75pt]  [font=\scriptsize] [align=left] {a};
\draw (184.03,483.75) node [anchor=north west][inner sep=0.75pt]  [font=\scriptsize] [align=left] {a};
\draw (145.66,502.69) node [anchor=north west][inner sep=0.75pt]  [font=\scriptsize] [align=left] {a};
\draw (91.75,451.95) node [anchor=north west][inner sep=0.75pt]  [font=\scriptsize] [align=left] {a};
\draw (147.36,422.28) node [anchor=north west][inner sep=0.75pt]  [font=\scriptsize] [align=left] {b};
\draw (126.11,463.58) node [anchor=north west][inner sep=0.75pt]  [font=\scriptsize] [align=left] {b};
\draw (247.04,468.36) node [anchor=north west][inner sep=0.75pt]  [font=\scriptsize] [align=left] {b};
\draw (160.03,397.46) node [anchor=north west][inner sep=0.75pt]  [font=\scriptsize] [align=left] {b};
\draw (304.83,347.68) node [anchor=north west][inner sep=0.75pt]  [font=\scriptsize] [align=left] {b};
\draw (218.82,347.51) node [anchor=north west][inner sep=0.75pt]  [font=\scriptsize] [align=left] {b};
\draw (117.29,366.58) node [anchor=north west][inner sep=0.75pt]  [font=\scriptsize] [align=left] {b};
\draw (72.17,346.94) node [anchor=north west][inner sep=0.75pt]  [font=\scriptsize] [align=left] {1b};
\draw (97.43,390.37) node [anchor=north west][inner sep=0.75pt]  [font=\scriptsize] [align=left] {2b};
\draw (109.26,432.02) node [anchor=north west][inner sep=0.75pt]  [font=\scriptsize] [align=left] {3b};
\draw (90.09,470.44) node [anchor=north west][inner sep=0.75pt]  [font=\scriptsize] [align=left] {4b};
\draw (124.06,510.57) node [anchor=north west][inner sep=0.75pt]  [font=\scriptsize] [align=left] {5b};
\draw (156.39,477.53) node [anchor=north west][inner sep=0.75pt]  [font=\scriptsize] [align=left] {6b};
\draw (182.81,431.89) node [anchor=north west][inner sep=0.75pt]  [font=\scriptsize] [align=left] {7b};
\draw (217.57,480.23) node [anchor=north west][inner sep=0.75pt]  [font=\scriptsize] [align=left] {9b};
\draw (150.57,365.03) node [anchor=north west][inner sep=0.75pt]  [font=\scriptsize] [align=left] {8b};
\draw (250.35,399.54) node [anchor=north west][inner sep=0.75pt]  [font=\scriptsize] [align=left] {10b};
\draw (261.5,439.93) node [anchor=north west][inner sep=0.75pt]  [font=\scriptsize] [align=left] {12b};
\draw (200.19,318.56) node [anchor=north west][inner sep=0.75pt]  [font=\scriptsize] [align=left] {11b};
\draw (263.13,349.21) node [anchor=north west][inner sep=0.75pt]  [font=\scriptsize] [align=left] {13b};
\draw (243.84,288.43) node [anchor=north west][inner sep=0.75pt]  [font=\scriptsize] [align=left] {14b};
\draw (288.08,318.04) node [anchor=north west][inner sep=0.75pt]  [font=\scriptsize] [align=left] {15b};

\end{tikzpicture}

%% file: second_step.tex
\tikzset{every picture/.style={line width=0.75pt}} 

\begin{tikzpicture}[x=0.75pt,y=0.75pt,yscale=-0.9,xscale=0.9]

\draw    (94,169.43) .. controls (102.6,175.16) and (106.53,188.33) .. (100.48,200.82) ;
\draw [shift={(99.56,202.58)}, rotate = 299.43] [color={rgb, 255:red, 0; green, 0; blue, 0 }  ][line width=0.75]    (10.93,-3.29) .. controls (6.95,-1.4) and (3.31,-0.3) .. (0,0) .. controls (3.31,0.3) and (6.95,1.4) .. (10.93,3.29)   ;
\draw    (89.12,205.41) .. controls (79.55,203.21) and (72.44,191.19) .. (76.58,171.12) ;
\draw [shift={(77,169.24)}, rotate = 103.27] [color={rgb, 255:red, 0; green, 0; blue, 0 }  ][line width=0.75]    (10.93,-3.29) .. controls (6.95,-1.4) and (3.31,-0.3) .. (0,0) .. controls (3.31,0.3) and (6.95,1.4) .. (10.93,3.29)   ;
\draw    (154,253.43) -- (193,254.38) ;
\draw [shift={(195,254.43)}, rotate = 181.4] [color={rgb, 255:red, 0; green, 0; blue, 0 }  ][line width=0.75]    (10.93,-3.29) .. controls (6.95,-1.4) and (3.31,-0.3) .. (0,0) .. controls (3.31,0.3) and (6.95,1.4) .. (10.93,3.29)   ;
\draw    (166.68,212.75) -- (145.28,242.77) ;
\draw [shift={(144.12,244.4)}, rotate = 305.49] [color={rgb, 255:red, 0; green, 0; blue, 0 }  ][line width=0.75]    (10.93,-3.29) .. controls (6.95,-1.4) and (3.31,-0.3) .. (0,0) .. controls (3.31,0.3) and (6.95,1.4) .. (10.93,3.29)   ;
\draw    (199,248.43) -- (173.2,214.03) ;
\draw [shift={(172,212.43)}, rotate = 53.13] [color={rgb, 255:red, 0; green, 0; blue, 0 }  ][line width=0.75]    (10.93,-3.29) .. controls (6.95,-1.4) and (3.31,-0.3) .. (0,0) .. controls (3.31,0.3) and (6.95,1.4) .. (10.93,3.29)   ;
\draw    (241.7,167.54) -- (249.69,134.45) ;
\draw [shift={(250.16,132.51)}, rotate = 103.58] [color={rgb, 255:red, 0; green, 0; blue, 0 }  ][line width=0.75]    (10.93,-3.29) .. controls (6.95,-1.4) and (3.31,-0.3) .. (0,0) .. controls (3.31,0.3) and (6.95,1.4) .. (10.93,3.29)   ;
\draw    (241.13,120.64) -- (203.38,98.48) ;
\draw [shift={(201.65,97.47)}, rotate = 30.41] [color={rgb, 255:red, 0; green, 0; blue, 0 }  ][line width=0.75]    (10.93,-3.29) .. controls (6.95,-1.4) and (3.31,-0.3) .. (0,0) .. controls (3.31,0.3) and (6.95,1.4) .. (10.93,3.29)   ;
\draw    (186.99,102.56) -- (148.54,134.16) ;
\draw [shift={(147,135.43)}, rotate = 320.57] [color={rgb, 255:red, 0; green, 0; blue, 0 }  ][line width=0.75]    (10.93,-3.29) .. controls (6.95,-1.4) and (3.31,-0.3) .. (0,0) .. controls (3.31,0.3) and (6.95,1.4) .. (10.93,3.29)   ;
\draw    (149,144.43) -- (229.12,176.41) ;
\draw [shift={(230.98,177.15)}, rotate = 201.76] [color={rgb, 255:red, 0; green, 0; blue, 0 }  ][line width=0.75]    (10.93,-3.29) .. controls (6.95,-1.4) and (3.31,-0.3) .. (0,0) .. controls (3.31,0.3) and (6.95,1.4) .. (10.93,3.29)   ;
\draw    (63,130.43) -- (78.06,158.67) ;
\draw [shift={(79,160.43)}, rotate = 241.93] [color={rgb, 255:red, 0; green, 0; blue, 0 }  ][line width=0.75]    (10.93,-3.29) .. controls (6.95,-1.4) and (3.31,-0.3) .. (0,0) .. controls (3.31,0.3) and (6.95,1.4) .. (10.93,3.29)   ;
\draw    (78.69,239.88) -- (90.18,217.64) ;
\draw [shift={(91.1,215.86)}, rotate = 117.32] [color={rgb, 255:red, 0; green, 0; blue, 0 }  ][line width=0.75]    (10.93,-3.29) .. controls (6.95,-1.4) and (3.31,-0.3) .. (0,0) .. controls (3.31,0.3) and (6.95,1.4) .. (10.93,3.29)   ;
\draw    (113.66,280.85) -- (134.29,258.86) ;
\draw [shift={(135.66,257.4)}, rotate = 133.17] [color={rgb, 255:red, 0; green, 0; blue, 0 }  ][line width=0.75]    (10.93,-3.29) .. controls (6.95,-1.4) and (3.31,-0.3) .. (0,0) .. controls (3.31,0.3) and (6.95,1.4) .. (10.93,3.29)   ;
\draw    (269,87.43) -- (248.51,69.74) ;
\draw [shift={(247,68.43)}, rotate = 40.82] [color={rgb, 255:red, 0; green, 0; blue, 0 }  ][line width=0.75]    (10.93,-3.29) .. controls (6.95,-1.4) and (3.31,-0.3) .. (0,0) .. controls (3.31,0.3) and (6.95,1.4) .. (10.93,3.29)   ;
\draw    (226,70.43) -- (201.27,89.2) ;
\draw [shift={(199.68,90.41)}, rotate = 322.81] [color={rgb, 255:red, 0; green, 0; blue, 0 }  ][line width=0.75]    (10.93,-3.29) .. controls (6.95,-1.4) and (3.31,-0.3) .. (0,0) .. controls (3.31,0.3) and (6.95,1.4) .. (10.93,3.29)   ;
\draw    (250.72,209.36) -- (243.88,182.48) ;
\draw [shift={(243.39,180.54)}, rotate = 75.73] [color={rgb, 255:red, 0; green, 0; blue, 0 }  ][line width=0.75]    (10.93,-3.29) .. controls (6.95,-1.4) and (3.31,-0.3) .. (0,0) .. controls (3.31,0.3) and (6.95,1.4) .. (10.93,3.29)   ;
\draw    (161.04,205.97) -- (109,207.38) ;
\draw [shift={(107,207.43)}, rotate = 358.45] [color={rgb, 255:red, 0; green, 0; blue, 0 }  ][line width=0.75]    (10.93,-3.29) .. controls (6.95,-1.4) and (3.31,-0.3) .. (0,0) .. controls (3.31,0.3) and (6.95,1.4) .. (10.93,3.29)   ;
\draw    (134.53,250.62) -- (102.38,216.88) ;
\draw [shift={(101,215.43)}, rotate = 46.38] [color={rgb, 255:red, 0; green, 0; blue, 0 }  ][line width=0.75]    (10.93,-3.29) .. controls (6.95,-1.4) and (3.31,-0.3) .. (0,0) .. controls (3.31,0.3) and (6.95,1.4) .. (10.93,3.29)   ;
\draw    (207.01,249.49) -- (243.5,221.32) ;
\draw [shift={(245.08,220.1)}, rotate = 142.34] [color={rgb, 255:red, 0; green, 0; blue, 0 }  ][line width=0.75]    (10.93,-3.29) .. controls (6.95,-1.4) and (3.31,-0.3) .. (0,0) .. controls (3.31,0.3) and (6.95,1.4) .. (10.93,3.29)   ;
\draw    (143,149.43) -- (165.16,197.62) ;
\draw [shift={(166,199.43)}, rotate = 245.3] [color={rgb, 255:red, 0; green, 0; blue, 0 }  ][line width=0.75]    (10.93,-3.29) .. controls (6.95,-1.4) and (3.31,-0.3) .. (0,0) .. controls (3.31,0.3) and (6.95,1.4) .. (10.93,3.29)   ;
\draw    (266.05,128.55) .. controls (284.75,130.2) and (291.65,114.67) .. (267.94,119.65) ;
\draw [shift={(266.05,120.07)}, rotate = 346.53] [color={rgb, 255:red, 0; green, 0; blue, 0 }  ][line width=0.75]    (10.93,-3.29) .. controls (6.95,-1.4) and (3.31,-0.3) .. (0,0) .. controls (3.31,0.3) and (6.95,1.4) .. (10.93,3.29)   ;
\draw    (196.01,102.56) .. controls (199.88,115.29) and (214.58,126.4) .. (240.65,126.31) ;
\draw [shift={(242.26,126.29)}, rotate = 178.8] [color={rgb, 255:red, 0; green, 0; blue, 0 }  ][line width=0.75]    (10.93,-3.29) .. controls (6.95,-1.4) and (3.31,-0.3) .. (0,0) .. controls (3.31,0.3) and (6.95,1.4) .. (10.93,3.29)   ;
\draw    (94,160) -- (127.64,144.1) ;
\draw [shift={(129.45,143.24)}, rotate = 154.7] [color={rgb, 255:red, 0; green, 0; blue, 0 }  ][line width=0.75]    (10.93,-3.29) .. controls (6.95,-1.4) and (3.31,-0.3) .. (0,0) .. controls (3.31,0.3) and (6.95,1.4) .. (10.93,3.29)   ;
\draw    (362,173.43) .. controls (370.6,179.16) and (374.53,192.33) .. (368.48,204.82) ;
\draw [shift={(367.56,206.58)}, rotate = 299.43] [color={rgb, 255:red, 0; green, 0; blue, 0 }  ][line width=0.75]    (10.93,-3.29) .. controls (6.95,-1.4) and (3.31,-0.3) .. (0,0) .. controls (3.31,0.3) and (6.95,1.4) .. (10.93,3.29)   ;
\draw    (357.12,209.41) .. controls (347.55,207.21) and (340.44,195.19) .. (344.58,175.12) ;
\draw [shift={(345,173.24)}, rotate = 103.27] [color={rgb, 255:red, 0; green, 0; blue, 0 }  ][line width=0.75]    (10.93,-3.29) .. controls (6.95,-1.4) and (3.31,-0.3) .. (0,0) .. controls (3.31,0.3) and (6.95,1.4) .. (10.93,3.29)   ;
\draw    (422,257.43) -- (461,258.38) ;
\draw [shift={(463,258.43)}, rotate = 181.4] [color={rgb, 255:red, 0; green, 0; blue, 0 }  ][line width=0.75]    (10.93,-3.29) .. controls (6.95,-1.4) and (3.31,-0.3) .. (0,0) .. controls (3.31,0.3) and (6.95,1.4) .. (10.93,3.29)   ;
\draw    (434.68,216.75) -- (413.28,246.77) ;
\draw [shift={(412.12,248.4)}, rotate = 305.49] [color={rgb, 255:red, 0; green, 0; blue, 0 }  ][line width=0.75]    (10.93,-3.29) .. controls (6.95,-1.4) and (3.31,-0.3) .. (0,0) .. controls (3.31,0.3) and (6.95,1.4) .. (10.93,3.29)   ;
\draw    (467,252.43) -- (441.2,218.03) ;
\draw [shift={(440,216.43)}, rotate = 53.13] [color={rgb, 255:red, 0; green, 0; blue, 0 }  ][line width=0.75]    (10.93,-3.29) .. controls (6.95,-1.4) and (3.31,-0.3) .. (0,0) .. controls (3.31,0.3) and (6.95,1.4) .. (10.93,3.29)   ;
\draw    (509.7,171.54) -- (517.69,138.45) ;
\draw [shift={(518.16,136.51)}, rotate = 103.58] [color={rgb, 255:red, 0; green, 0; blue, 0 }  ][line width=0.75]    (10.93,-3.29) .. controls (6.95,-1.4) and (3.31,-0.3) .. (0,0) .. controls (3.31,0.3) and (6.95,1.4) .. (10.93,3.29)   ;
\draw    (509.13,124.64) -- (471.38,102.48) ;
\draw [shift={(469.65,101.47)}, rotate = 30.41] [color={rgb, 255:red, 0; green, 0; blue, 0 }  ][line width=0.75]    (10.93,-3.29) .. controls (6.95,-1.4) and (3.31,-0.3) .. (0,0) .. controls (3.31,0.3) and (6.95,1.4) .. (10.93,3.29)   ;
\draw    (454.99,106.56) -- (416.54,138.16) ;
\draw [shift={(415,139.43)}, rotate = 320.57] [color={rgb, 255:red, 0; green, 0; blue, 0 }  ][line width=0.75]    (10.93,-3.29) .. controls (6.95,-1.4) and (3.31,-0.3) .. (0,0) .. controls (3.31,0.3) and (6.95,1.4) .. (10.93,3.29)   ;
\draw    (417,148.43) -- (497.12,180.41) ;
\draw [shift={(498.98,181.15)}, rotate = 201.76] [color={rgb, 255:red, 0; green, 0; blue, 0 }  ][line width=0.75]    (10.93,-3.29) .. controls (6.95,-1.4) and (3.31,-0.3) .. (0,0) .. controls (3.31,0.3) and (6.95,1.4) .. (10.93,3.29)   ;
\draw    (331,134.43) -- (346.06,162.67) ;
\draw [shift={(347,164.43)}, rotate = 241.93] [color={rgb, 255:red, 0; green, 0; blue, 0 }  ][line width=0.75]    (10.93,-3.29) .. controls (6.95,-1.4) and (3.31,-0.3) .. (0,0) .. controls (3.31,0.3) and (6.95,1.4) .. (10.93,3.29)   ;
\draw    (346.69,243.88) -- (358.18,221.64) ;
\draw [shift={(359.1,219.86)}, rotate = 117.32] [color={rgb, 255:red, 0; green, 0; blue, 0 }  ][line width=0.75]    (10.93,-3.29) .. controls (6.95,-1.4) and (3.31,-0.3) .. (0,0) .. controls (3.31,0.3) and (6.95,1.4) .. (10.93,3.29)   ;
\draw    (381.66,284.85) -- (402.29,262.86) ;
\draw [shift={(403.66,261.4)}, rotate = 133.17] [color={rgb, 255:red, 0; green, 0; blue, 0 }  ][line width=0.75]    (10.93,-3.29) .. controls (6.95,-1.4) and (3.31,-0.3) .. (0,0) .. controls (3.31,0.3) and (6.95,1.4) .. (10.93,3.29)   ;
\draw    (537,91.43) -- (516.51,73.74) ;
\draw [shift={(515,72.43)}, rotate = 40.82] [color={rgb, 255:red, 0; green, 0; blue, 0 }  ][line width=0.75]    (10.93,-3.29) .. controls (6.95,-1.4) and (3.31,-0.3) .. (0,0) .. controls (3.31,0.3) and (6.95,1.4) .. (10.93,3.29)   ;
\draw    (494,74.43) -- (469.27,93.2) ;
\draw [shift={(467.68,94.41)}, rotate = 322.81] [color={rgb, 255:red, 0; green, 0; blue, 0 }  ][line width=0.75]    (10.93,-3.29) .. controls (6.95,-1.4) and (3.31,-0.3) .. (0,0) .. controls (3.31,0.3) and (6.95,1.4) .. (10.93,3.29)   ;
\draw    (518.72,213.36) -- (511.88,186.48) ;
\draw [shift={(511.39,184.54)}, rotate = 75.73] [color={rgb, 255:red, 0; green, 0; blue, 0 }  ][line width=0.75]    (10.93,-3.29) .. controls (6.95,-1.4) and (3.31,-0.3) .. (0,0) .. controls (3.31,0.3) and (6.95,1.4) .. (10.93,3.29)   ;
\draw    (429.04,209.97) -- (377,211.38) ;
\draw [shift={(375,211.43)}, rotate = 358.45] [color={rgb, 255:red, 0; green, 0; blue, 0 }  ][line width=0.75]    (10.93,-3.29) .. controls (6.95,-1.4) and (3.31,-0.3) .. (0,0) .. controls (3.31,0.3) and (6.95,1.4) .. (10.93,3.29)   ;
\draw    (402.53,254.62) -- (370.38,220.88) ;
\draw [shift={(369,219.43)}, rotate = 46.38] [color={rgb, 255:red, 0; green, 0; blue, 0 }  ][line width=0.75]    (10.93,-3.29) .. controls (6.95,-1.4) and (3.31,-0.3) .. (0,0) .. controls (3.31,0.3) and (6.95,1.4) .. (10.93,3.29)   ;
\draw    (475.01,253.49) -- (511.5,225.32) ;
\draw [shift={(513.08,224.1)}, rotate = 142.34] [color={rgb, 255:red, 0; green, 0; blue, 0 }  ][line width=0.75]    (10.93,-3.29) .. controls (6.95,-1.4) and (3.31,-0.3) .. (0,0) .. controls (3.31,0.3) and (6.95,1.4) .. (10.93,3.29)   ;
\draw    (411,153.43) -- (433.16,201.62) ;
\draw [shift={(434,203.43)}, rotate = 245.3] [color={rgb, 255:red, 0; green, 0; blue, 0 }  ][line width=0.75]    (10.93,-3.29) .. controls (6.95,-1.4) and (3.31,-0.3) .. (0,0) .. controls (3.31,0.3) and (6.95,1.4) .. (10.93,3.29)   ;
\draw    (534.05,132.55) .. controls (552.75,134.2) and (559.65,118.67) .. (535.94,123.65) ;
\draw [shift={(534.05,124.07)}, rotate = 346.53] [color={rgb, 255:red, 0; green, 0; blue, 0 }  ][line width=0.75]    (10.93,-3.29) .. controls (6.95,-1.4) and (3.31,-0.3) .. (0,0) .. controls (3.31,0.3) and (6.95,1.4) .. (10.93,3.29)   ;
\draw    (464.01,106.56) .. controls (467.88,119.29) and (482.58,130.4) .. (508.65,130.31) ;
\draw [shift={(510.26,130.29)}, rotate = 178.8] [color={rgb, 255:red, 0; green, 0; blue, 0 }  ][line width=0.75]    (10.93,-3.29) .. controls (6.95,-1.4) and (3.31,-0.3) .. (0,0) .. controls (3.31,0.3) and (6.95,1.4) .. (10.93,3.29)   ;
\draw    (363,164) -- (395.65,148.12) ;
\draw [shift={(397.45,147.24)}, rotate = 154.07] [color={rgb, 255:red, 0; green, 0; blue, 0 }  ][line width=0.75]    (10.93,-3.29) .. controls (6.95,-1.4) and (3.31,-0.3) .. (0,0) .. controls (3.31,0.3) and (6.95,1.4) .. (10.93,3.29)   ;
\draw    (133,407.43) .. controls (141.6,413.16) and (145.53,426.33) .. (139.48,438.82) ;
\draw [shift={(138.56,440.58)}, rotate = 299.43] [color={rgb, 255:red, 0; green, 0; blue, 0 }  ][line width=0.75]    (10.93,-3.29) .. controls (6.95,-1.4) and (3.31,-0.3) .. (0,0) .. controls (3.31,0.3) and (6.95,1.4) .. (10.93,3.29)   ;
\draw    (128.12,443.41) .. controls (118.55,441.21) and (111.44,429.19) .. (115.58,409.12) ;
\draw [shift={(116,407.24)}, rotate = 103.27] [color={rgb, 255:red, 0; green, 0; blue, 0 }  ][line width=0.75]    (10.93,-3.29) .. controls (6.95,-1.4) and (3.31,-0.3) .. (0,0) .. controls (3.31,0.3) and (6.95,1.4) .. (10.93,3.29)   ;
\draw    (193,491.43) -- (232,492.38) ;
\draw [shift={(234,492.43)}, rotate = 181.4] [color={rgb, 255:red, 0; green, 0; blue, 0 }  ][line width=0.75]    (10.93,-3.29) .. controls (6.95,-1.4) and (3.31,-0.3) .. (0,0) .. controls (3.31,0.3) and (6.95,1.4) .. (10.93,3.29)   ;
\draw    (205.68,450.75) -- (184.28,480.77) ;
\draw [shift={(183.12,482.4)}, rotate = 305.49] [color={rgb, 255:red, 0; green, 0; blue, 0 }  ][line width=0.75]    (10.93,-3.29) .. controls (6.95,-1.4) and (3.31,-0.3) .. (0,0) .. controls (3.31,0.3) and (6.95,1.4) .. (10.93,3.29)   ;
\draw    (238,486.43) -- (212.2,452.03) ;
\draw [shift={(211,450.43)}, rotate = 53.13] [color={rgb, 255:red, 0; green, 0; blue, 0 }  ][line width=0.75]    (10.93,-3.29) .. controls (6.95,-1.4) and (3.31,-0.3) .. (0,0) .. controls (3.31,0.3) and (6.95,1.4) .. (10.93,3.29)   ;
\draw    (280.7,405.54) -- (288.69,372.45) ;
\draw [shift={(289.16,370.51)}, rotate = 103.58] [color={rgb, 255:red, 0; green, 0; blue, 0 }  ][line width=0.75]    (10.93,-3.29) .. controls (6.95,-1.4) and (3.31,-0.3) .. (0,0) .. controls (3.31,0.3) and (6.95,1.4) .. (10.93,3.29)   ;
\draw    (280.13,358.64) -- (242.38,336.48) ;
\draw [shift={(240.65,335.47)}, rotate = 30.41] [color={rgb, 255:red, 0; green, 0; blue, 0 }  ][line width=0.75]    (10.93,-3.29) .. controls (6.95,-1.4) and (3.31,-0.3) .. (0,0) .. controls (3.31,0.3) and (6.95,1.4) .. (10.93,3.29)   ;
\draw    (225.99,340.56) -- (187.54,372.16) ;
\draw [shift={(186,373.43)}, rotate = 320.57] [color={rgb, 255:red, 0; green, 0; blue, 0 }  ][line width=0.75]    (10.93,-3.29) .. controls (6.95,-1.4) and (3.31,-0.3) .. (0,0) .. controls (3.31,0.3) and (6.95,1.4) .. (10.93,3.29)   ;
\draw    (188,382.43) -- (268.12,414.41) ;
\draw [shift={(269.98,415.15)}, rotate = 201.76] [color={rgb, 255:red, 0; green, 0; blue, 0 }  ][line width=0.75]    (10.93,-3.29) .. controls (6.95,-1.4) and (3.31,-0.3) .. (0,0) .. controls (3.31,0.3) and (6.95,1.4) .. (10.93,3.29)   ;
\draw    (102,368.43) -- (117.06,396.67) ;
\draw [shift={(118,398.43)}, rotate = 241.93] [color={rgb, 255:red, 0; green, 0; blue, 0 }  ][line width=0.75]    (10.93,-3.29) .. controls (6.95,-1.4) and (3.31,-0.3) .. (0,0) .. controls (3.31,0.3) and (6.95,1.4) .. (10.93,3.29)   ;
\draw    (117.69,477.88) -- (129.18,455.64) ;
\draw [shift={(130.1,453.86)}, rotate = 117.32] [color={rgb, 255:red, 0; green, 0; blue, 0 }  ][line width=0.75]    (10.93,-3.29) .. controls (6.95,-1.4) and (3.31,-0.3) .. (0,0) .. controls (3.31,0.3) and (6.95,1.4) .. (10.93,3.29)   ;
\draw    (152.66,518.85) -- (173.29,496.86) ;
\draw [shift={(174.66,495.4)}, rotate = 133.17] [color={rgb, 255:red, 0; green, 0; blue, 0 }  ][line width=0.75]    (10.93,-3.29) .. controls (6.95,-1.4) and (3.31,-0.3) .. (0,0) .. controls (3.31,0.3) and (6.95,1.4) .. (10.93,3.29)   ;
\draw    (308,325.43) -- (287.51,307.74) ;
\draw [shift={(286,306.43)}, rotate = 40.82] [color={rgb, 255:red, 0; green, 0; blue, 0 }  ][line width=0.75]    (10.93,-3.29) .. controls (6.95,-1.4) and (3.31,-0.3) .. (0,0) .. controls (3.31,0.3) and (6.95,1.4) .. (10.93,3.29)   ;
\draw    (265,308.43) -- (240.27,327.2) ;
\draw [shift={(238.68,328.41)}, rotate = 322.81] [color={rgb, 255:red, 0; green, 0; blue, 0 }  ][line width=0.75]    (10.93,-3.29) .. controls (6.95,-1.4) and (3.31,-0.3) .. (0,0) .. controls (3.31,0.3) and (6.95,1.4) .. (10.93,3.29)   ;
\draw    (289.72,447.36) -- (282.88,420.48) ;
\draw [shift={(282.39,418.54)}, rotate = 75.73] [color={rgb, 255:red, 0; green, 0; blue, 0 }  ][line width=0.75]    (10.93,-3.29) .. controls (6.95,-1.4) and (3.31,-0.3) .. (0,0) .. controls (3.31,0.3) and (6.95,1.4) .. (10.93,3.29)   ;
\draw    (200.04,443.97) -- (148,445.38) ;
\draw [shift={(146,445.43)}, rotate = 358.45] [color={rgb, 255:red, 0; green, 0; blue, 0 }  ][line width=0.75]    (10.93,-3.29) .. controls (6.95,-1.4) and (3.31,-0.3) .. (0,0) .. controls (3.31,0.3) and (6.95,1.4) .. (10.93,3.29)   ;
\draw    (173.53,488.62) -- (141.38,454.88) ;
\draw [shift={(140,453.43)}, rotate = 46.38] [color={rgb, 255:red, 0; green, 0; blue, 0 }  ][line width=0.75]    (10.93,-3.29) .. controls (6.95,-1.4) and (3.31,-0.3) .. (0,0) .. controls (3.31,0.3) and (6.95,1.4) .. (10.93,3.29)   ;
\draw    (246.01,487.49) -- (282.5,459.32) ;
\draw [shift={(284.08,458.1)}, rotate = 142.34] [color={rgb, 255:red, 0; green, 0; blue, 0 }  ][line width=0.75]    (10.93,-3.29) .. controls (6.95,-1.4) and (3.31,-0.3) .. (0,0) .. controls (3.31,0.3) and (6.95,1.4) .. (10.93,3.29)   ;
\draw    (182,387.43) -- (204.16,435.62) ;
\draw [shift={(205,437.43)}, rotate = 245.3] [color={rgb, 255:red, 0; green, 0; blue, 0 }  ][line width=0.75]    (10.93,-3.29) .. controls (6.95,-1.4) and (3.31,-0.3) .. (0,0) .. controls (3.31,0.3) and (6.95,1.4) .. (10.93,3.29)   ;
\draw    (305.05,366.55) .. controls (323.75,368.2) and (330.65,352.67) .. (306.94,357.65) ;
\draw [shift={(305.05,358.07)}, rotate = 346.53] [color={rgb, 255:red, 0; green, 0; blue, 0 }  ][line width=0.75]    (10.93,-3.29) .. controls (6.95,-1.4) and (3.31,-0.3) .. (0,0) .. controls (3.31,0.3) and (6.95,1.4) .. (10.93,3.29)   ;
\draw    (235.01,340.56) .. controls (238.88,353.29) and (253.58,364.4) .. (279.65,364.31) ;
\draw [shift={(281.26,364.29)}, rotate = 178.8] [color={rgb, 255:red, 0; green, 0; blue, 0 }  ][line width=0.75]    (10.93,-3.29) .. controls (6.95,-1.4) and (3.31,-0.3) .. (0,0) .. controls (3.31,0.3) and (6.95,1.4) .. (10.93,3.29)   ;
\draw    (131,398) -- (166.63,382.06) ;
\draw [shift={(168.45,381.24)}, rotate = 155.9] [color={rgb, 255:red, 0; green, 0; blue, 0 }  ][line width=0.75]    (10.93,-3.29) .. controls (6.95,-1.4) and (3.31,-0.3) .. (0,0) .. controls (3.31,0.3) and (6.95,1.4) .. (10.93,3.29)   ;
\draw [line width=1.5]    (116,295) -- (123.76,390.01) ;
\draw [shift={(124,393)}, rotate = 265.33] [color={rgb, 255:red, 0; green, 0; blue, 0 }  ][line width=1.5]    (14.21,-4.28) .. controls (9.04,-1.82) and (4.3,-0.39) .. (0,0) .. controls (4.3,0.39) and (9.04,1.82) .. (14.21,4.28)   ;
\draw [line width=1.5]    (304,448) -- (464.01,268.24) ;
\draw [shift={(466,266)}, rotate = 131.67] [color={rgb, 255:red, 0; green, 0; blue, 0 }  ][line width=1.5]    (14.21,-4.28) .. controls (9.04,-1.82) and (4.3,-0.39) .. (0,0) .. controls (4.3,0.39) and (9.04,1.82) .. (14.21,4.28)   ;

\draw (261.35,69.56) node [anchor=north west][inner sep=0.75pt]  [font=\scriptsize] [align=left] {a};
\draw (209.58,66.95) node [anchor=north west][inner sep=0.75pt]  [font=\scriptsize] [align=left] {a};
\draw (155.13,104.82) node [anchor=north west][inner sep=0.75pt]  [font=\scriptsize] [align=left] {a};
\draw (180.08,160.46) node [anchor=north west][inner sep=0.75pt]  [font=\scriptsize] [align=left] {a};
\draw (221.99,95.77) node [anchor=north west][inner sep=0.75pt]  [font=\scriptsize] [align=left] {a};
\draw (247.93,153.03) node [anchor=north west][inner sep=0.75pt]  [font=\scriptsize] [align=left] {a,b};
\draw (251.76,191.89) node [anchor=north west][inner sep=0.75pt]  [font=\scriptsize] [align=left] {a};
\draw (56.42,137.46) node [anchor=north west][inner sep=0.75pt]  [font=\scriptsize] [align=left] {a};
\draw (104.98,176.19) node [anchor=north west][inner sep=0.75pt]  [font=\scriptsize] [align=left] {a};
\draw (56.19,185.32) node [anchor=north west][inner sep=0.75pt]  [font=\scriptsize] [align=left] {a,b};
\draw (188.54,222.71) node [anchor=north west][inner sep=0.75pt]  [font=\scriptsize] [align=left] {a};
\draw (146.16,217.14) node [anchor=north west][inner sep=0.75pt]  [font=\scriptsize] [align=left] {a};
\draw (165.03,253.18) node [anchor=north west][inner sep=0.75pt]  [font=\scriptsize] [align=left] {a};
\draw (126.66,272.12) node [anchor=north west][inner sep=0.75pt]  [font=\scriptsize] [align=left] {a};
\draw (72.75,221.38) node [anchor=north west][inner sep=0.75pt]  [font=\scriptsize] [align=left] {a};
\draw (128.36,191.71) node [anchor=north west][inner sep=0.75pt]  [font=\scriptsize] [align=left] {b};
\draw (107.11,233.01) node [anchor=north west][inner sep=0.75pt]  [font=\scriptsize] [align=left] {b};
\draw (228.04,237.79) node [anchor=north west][inner sep=0.75pt]  [font=\scriptsize] [align=left] {b};
\draw (141.03,166.89) node [anchor=north west][inner sep=0.75pt]  [font=\scriptsize] [align=left] {b};
\draw (285.83,117.12) node [anchor=north west][inner sep=0.75pt]  [font=\scriptsize] [align=left] {b};
\draw (199.82,116.94) node [anchor=north west][inner sep=0.75pt]  [font=\scriptsize] [align=left] {b};
\draw (98.29,136.02) node [anchor=north west][inner sep=0.75pt]  [font=\scriptsize] [align=left] {b};
\draw (53.17,116.38) node [anchor=north west][inner sep=0.75pt]  [font=\scriptsize] [align=left] {1a};
\draw (78.43,159.81) node [anchor=north west][inner sep=0.75pt]  [font=\scriptsize] [align=left] {2a};
\draw (90.26,201.45) node [anchor=north west][inner sep=0.75pt]  [font=\scriptsize] [align=left] {3a};
\draw (71.09,239.88) node [anchor=north west][inner sep=0.75pt]  [font=\scriptsize] [align=left] {4a};
\draw (105.06,280) node [anchor=north west][inner sep=0.75pt]  [font=\scriptsize] [align=left] {5a};
\draw (137.39,246.96) node [anchor=north west][inner sep=0.75pt]  [font=\scriptsize] [align=left] {6a};
\draw (163.81,201.32) node [anchor=north west][inner sep=0.75pt]  [font=\scriptsize] [align=left] {7a};
\draw (198.57,249.66) node [anchor=north west][inner sep=0.75pt]  [font=\scriptsize] [align=left] {9a};
\draw (131.57,134.46) node [anchor=north west][inner sep=0.75pt]  [font=\scriptsize] [align=left] {8a};
\draw (231.35,168.98) node [anchor=north west][inner sep=0.75pt]  [font=\scriptsize] [align=left] {10a};
\draw (242.5,209.36) node [anchor=north west][inner sep=0.75pt]  [font=\scriptsize] [align=left] {12a};
\draw (181.19,87.99) node [anchor=north west][inner sep=0.75pt]  [font=\scriptsize] [align=left] {11a};
\draw (244.13,118.64) node [anchor=north west][inner sep=0.75pt]  [font=\scriptsize] [align=left] {13a};
\draw (224.84,57.87) node [anchor=north west][inner sep=0.75pt]  [font=\scriptsize] [align=left] {14a};
\draw (269.08,87.47) node [anchor=north west][inner sep=0.75pt]  [font=\scriptsize] [align=left] {15a};
\draw (529.35,74.56) node [anchor=north west][inner sep=0.75pt]  [font=\scriptsize] [align=left] {a};
\draw (477.58,70.95) node [anchor=north west][inner sep=0.75pt]  [font=\scriptsize] [align=left] {a};
\draw (423.13,108.82) node [anchor=north west][inner sep=0.75pt]  [font=\scriptsize] [align=left] {a};
\draw (448.08,164.46) node [anchor=north west][inner sep=0.75pt]  [font=\scriptsize] [align=left] {a};
\draw (489.99,99.77) node [anchor=north west][inner sep=0.75pt]  [font=\scriptsize] [align=left] {a};
\draw (515.93,157.03) node [anchor=north west][inner sep=0.75pt]  [font=\scriptsize] [align=left] {a,b};
\draw (519.76,195.89) node [anchor=north west][inner sep=0.75pt]  [font=\scriptsize] [align=left] {a};
\draw (324.42,141.46) node [anchor=north west][inner sep=0.75pt]  [font=\scriptsize] [align=left] {a};
\draw (372.98,180.19) node [anchor=north west][inner sep=0.75pt]  [font=\scriptsize] [align=left] {a};
\draw (324.19,189.32) node [anchor=north west][inner sep=0.75pt]  [font=\scriptsize] [align=left] {a,b};
\draw (456.54,226.71) node [anchor=north west][inner sep=0.75pt]  [font=\scriptsize] [align=left] {a};
\draw (414.16,221.14) node [anchor=north west][inner sep=0.75pt]  [font=\scriptsize] [align=left] {a};
\draw (433.03,257.18) node [anchor=north west][inner sep=0.75pt]  [font=\scriptsize] [align=left] {a};
\draw (394.66,276.12) node [anchor=north west][inner sep=0.75pt]  [font=\scriptsize] [align=left] {a};
\draw (340.75,225.38) node [anchor=north west][inner sep=0.75pt]  [font=\scriptsize] [align=left] {a};
\draw (396.36,195.71) node [anchor=north west][inner sep=0.75pt]  [font=\scriptsize] [align=left] {b};
\draw (375.11,237.01) node [anchor=north west][inner sep=0.75pt]  [font=\scriptsize] [align=left] {b};
\draw (496.04,241.79) node [anchor=north west][inner sep=0.75pt]  [font=\scriptsize] [align=left] {b};
\draw (409.03,170.89) node [anchor=north west][inner sep=0.75pt]  [font=\scriptsize] [align=left] {b};
\draw (553.83,121.12) node [anchor=north west][inner sep=0.75pt]  [font=\scriptsize] [align=left] {b};
\draw (467.82,120.94) node [anchor=north west][inner sep=0.75pt]  [font=\scriptsize] [align=left] {b};
\draw (366.29,140.02) node [anchor=north west][inner sep=0.75pt]  [font=\scriptsize] [align=left] {b};
\draw (321.17,120.38) node [anchor=north west][inner sep=0.75pt]  [font=\scriptsize] [align=left] {1c};
\draw (346.43,163.81) node [anchor=north west][inner sep=0.75pt]  [font=\scriptsize] [align=left] {2c};
\draw (358.26,205.45) node [anchor=north west][inner sep=0.75pt]  [font=\scriptsize] [align=left] {3c};
\draw (339.09,243.88) node [anchor=north west][inner sep=0.75pt]  [font=\scriptsize] [align=left] {4c};
\draw (373.06,284) node [anchor=north west][inner sep=0.75pt]  [font=\scriptsize] [align=left] {5c};
\draw (405.39,250.96) node [anchor=north west][inner sep=0.75pt]  [font=\scriptsize] [align=left] {6c};
\draw (430,204.43) node [anchor=north west][inner sep=0.75pt]  [font=\scriptsize] [align=left] {7c};
\draw (464,253.43) node [anchor=north west][inner sep=0.75pt]  [font=\scriptsize] [align=left] {9c};
\draw (399.57,138.46) node [anchor=north west][inner sep=0.75pt]  [font=\scriptsize] [align=left] {8c};
\draw (499.35,172.98) node [anchor=north west][inner sep=0.75pt]  [font=\scriptsize] [align=left] {10c};
\draw (510.5,213.36) node [anchor=north west][inner sep=0.75pt]  [font=\scriptsize] [align=left] {12c};
\draw (449.19,91.99) node [anchor=north west][inner sep=0.75pt]  [font=\scriptsize] [align=left] {11c};
\draw (512.13,122.64) node [anchor=north west][inner sep=0.75pt]  [font=\scriptsize] [align=left] {13c};
\draw (492.84,61.87) node [anchor=north west][inner sep=0.75pt]  [font=\scriptsize] [align=left] {14c};
\draw (537.08,91.47) node [anchor=north west][inner sep=0.75pt]  [font=\scriptsize] [align=left] {15c};
\draw (299.35,306.56) node [anchor=north west][inner sep=0.75pt]  [font=\scriptsize] [align=left] {a};
\draw (248.58,304.95) node [anchor=north west][inner sep=0.75pt]  [font=\scriptsize] [align=left] {a};
\draw (194.13,342.82) node [anchor=north west][inner sep=0.75pt]  [font=\scriptsize] [align=left] {a};
\draw (219.08,398.46) node [anchor=north west][inner sep=0.75pt]  [font=\scriptsize] [align=left] {a};
\draw (260.99,333.77) node [anchor=north west][inner sep=0.75pt]  [font=\scriptsize] [align=left] {a};
\draw (286.93,391.03) node [anchor=north west][inner sep=0.75pt]  [font=\scriptsize] [align=left] {a,b};
\draw (290.76,429.89) node [anchor=north west][inner sep=0.75pt]  [font=\scriptsize] [align=left] {a};
\draw (95.42,375.46) node [anchor=north west][inner sep=0.75pt]  [font=\scriptsize] [align=left] {a};
\draw (143.98,414.19) node [anchor=north west][inner sep=0.75pt]  [font=\scriptsize] [align=left] {a};
\draw (95.19,423.32) node [anchor=north west][inner sep=0.75pt]  [font=\scriptsize] [align=left] {a,b};
\draw (227.54,460.71) node [anchor=north west][inner sep=0.75pt]  [font=\scriptsize] [align=left] {a};
\draw (185.16,455.14) node [anchor=north west][inner sep=0.75pt]  [font=\scriptsize] [align=left] {a};
\draw (204.03,491.18) node [anchor=north west][inner sep=0.75pt]  [font=\scriptsize] [align=left] {a};
\draw (165.66,510.12) node [anchor=north west][inner sep=0.75pt]  [font=\scriptsize] [align=left] {a};
\draw (111.75,459.38) node [anchor=north west][inner sep=0.75pt]  [font=\scriptsize] [align=left] {a};
\draw (167.36,429.71) node [anchor=north west][inner sep=0.75pt]  [font=\scriptsize] [align=left] {b};
\draw (146.11,471.01) node [anchor=north west][inner sep=0.75pt]  [font=\scriptsize] [align=left] {b};
\draw (267.04,475.79) node [anchor=north west][inner sep=0.75pt]  [font=\scriptsize] [align=left] {b};
\draw (180.03,404.89) node [anchor=north west][inner sep=0.75pt]  [font=\scriptsize] [align=left] {b};
\draw (324.83,355.12) node [anchor=north west][inner sep=0.75pt]  [font=\scriptsize] [align=left] {b};
\draw (238.82,354.94) node [anchor=north west][inner sep=0.75pt]  [font=\scriptsize] [align=left] {b};
\draw (137.29,374.02) node [anchor=north west][inner sep=0.75pt]  [font=\scriptsize] [align=left] {b};
\draw (92.17,354.38) node [anchor=north west][inner sep=0.75pt]  [font=\scriptsize] [align=left] {1b};
\draw (117.43,397.81) node [anchor=north west][inner sep=0.75pt]  [font=\scriptsize] [align=left] {2b};
\draw (129.26,439.45) node [anchor=north west][inner sep=0.75pt]  [font=\scriptsize] [align=left] {3b};
\draw (110.09,477.88) node [anchor=north west][inner sep=0.75pt]  [font=\scriptsize] [align=left] {4b};
\draw (144.06,518) node [anchor=north west][inner sep=0.75pt]  [font=\scriptsize] [align=left] {5b};
\draw (176.39,484.96) node [anchor=north west][inner sep=0.75pt]  [font=\scriptsize] [align=left] {6b};
\draw (202.81,439.32) node [anchor=north west][inner sep=0.75pt]  [font=\scriptsize] [align=left] {7b};
\draw (237.57,487.66) node [anchor=north west][inner sep=0.75pt]  [font=\scriptsize] [align=left] {9b};
\draw (170.57,372.46) node [anchor=north west][inner sep=0.75pt]  [font=\scriptsize] [align=left] {8b};
\draw (270.35,406.98) node [anchor=north west][inner sep=0.75pt]  [font=\scriptsize] [align=left] {10b};
\draw (281.5,447.36) node [anchor=north west][inner sep=0.75pt]  [font=\scriptsize] [align=left] {12b};
\draw (220.19,325.99) node [anchor=north west][inner sep=0.75pt]  [font=\scriptsize] [align=left] {11b};
\draw (283.13,356.64) node [anchor=north west][inner sep=0.75pt]  [font=\scriptsize] [align=left] {13b};
\draw (263.84,295.87) node [anchor=north west][inner sep=0.75pt]  [font=\scriptsize] [align=left] {14b};
\draw (308.08,325.47) node [anchor=north west][inner sep=0.75pt]  [font=\scriptsize] [align=left] {15b};
\draw (124,331) node [anchor=north west][inner sep=0.75pt]  [font=\scriptsize] [align=left] {0};
\draw (387,360) node [anchor=north west][inner sep=0.75pt]  [font=\scriptsize] [align=left] {1};

\end{tikzpicture}

%% file: third_step.tex
\tikzset{every picture/.style={line width=0.75pt}} 

\begin{tikzpicture}[x=0.75pt,y=0.75pt,yscale=-0.9,xscale=0.9]

\draw    (87,161.43) .. controls (95.6,167.16) and (99.53,180.33) .. (93.48,192.82) ;
\draw [shift={(92.56,194.58)}, rotate = 299.43] [color={rgb, 255:red, 0; green, 0; blue, 0 }  ][line width=0.75]    (10.93,-3.29) .. controls (6.95,-1.4) and (3.31,-0.3) .. (0,0) .. controls (3.31,0.3) and (6.95,1.4) .. (10.93,3.29)   ;
\draw    (82.12,197.41) .. controls (72.55,195.21) and (65.44,183.19) .. (69.58,163.12) ;
\draw [shift={(70,161.24)}, rotate = 103.27] [color={rgb, 255:red, 0; green, 0; blue, 0 }  ][line width=0.75]    (10.93,-3.29) .. controls (6.95,-1.4) and (3.31,-0.3) .. (0,0) .. controls (3.31,0.3) and (6.95,1.4) .. (10.93,3.29)   ;
\draw    (147,245.43) -- (186,246.38) ;
\draw [shift={(188,246.43)}, rotate = 181.4] [color={rgb, 255:red, 0; green, 0; blue, 0 }  ][line width=0.75]    (10.93,-3.29) .. controls (6.95,-1.4) and (3.31,-0.3) .. (0,0) .. controls (3.31,0.3) and (6.95,1.4) .. (10.93,3.29)   ;
\draw    (159.68,204.75) -- (138.28,234.77) ;
\draw [shift={(137.12,236.4)}, rotate = 305.49] [color={rgb, 255:red, 0; green, 0; blue, 0 }  ][line width=0.75]    (10.93,-3.29) .. controls (6.95,-1.4) and (3.31,-0.3) .. (0,0) .. controls (3.31,0.3) and (6.95,1.4) .. (10.93,3.29)   ;
\draw    (192,240.43) -- (166.2,206.03) ;
\draw [shift={(165,204.43)}, rotate = 53.13] [color={rgb, 255:red, 0; green, 0; blue, 0 }  ][line width=0.75]    (10.93,-3.29) .. controls (6.95,-1.4) and (3.31,-0.3) .. (0,0) .. controls (3.31,0.3) and (6.95,1.4) .. (10.93,3.29)   ;
\draw    (234.7,159.54) -- (242.69,126.45) ;
\draw [shift={(243.16,124.51)}, rotate = 103.58] [color={rgb, 255:red, 0; green, 0; blue, 0 }  ][line width=0.75]    (10.93,-3.29) .. controls (6.95,-1.4) and (3.31,-0.3) .. (0,0) .. controls (3.31,0.3) and (6.95,1.4) .. (10.93,3.29)   ;
\draw    (234.13,112.64) -- (196.38,90.48) ;
\draw [shift={(194.65,89.47)}, rotate = 30.41] [color={rgb, 255:red, 0; green, 0; blue, 0 }  ][line width=0.75]    (10.93,-3.29) .. controls (6.95,-1.4) and (3.31,-0.3) .. (0,0) .. controls (3.31,0.3) and (6.95,1.4) .. (10.93,3.29)   ;
\draw    (179.99,94.56) -- (141.54,126.16) ;
\draw [shift={(140,127.43)}, rotate = 320.57] [color={rgb, 255:red, 0; green, 0; blue, 0 }  ][line width=0.75]    (10.93,-3.29) .. controls (6.95,-1.4) and (3.31,-0.3) .. (0,0) .. controls (3.31,0.3) and (6.95,1.4) .. (10.93,3.29)   ;
\draw    (142,136.43) -- (222.12,168.41) ;
\draw [shift={(223.98,169.15)}, rotate = 201.76] [color={rgb, 255:red, 0; green, 0; blue, 0 }  ][line width=0.75]    (10.93,-3.29) .. controls (6.95,-1.4) and (3.31,-0.3) .. (0,0) .. controls (3.31,0.3) and (6.95,1.4) .. (10.93,3.29)   ;
\draw    (56,122.43) -- (71.06,150.67) ;
\draw [shift={(72,152.43)}, rotate = 241.93] [color={rgb, 255:red, 0; green, 0; blue, 0 }  ][line width=0.75]    (10.93,-3.29) .. controls (6.95,-1.4) and (3.31,-0.3) .. (0,0) .. controls (3.31,0.3) and (6.95,1.4) .. (10.93,3.29)   ;
\draw    (107,298) -- (91.37,213.97) ;
\draw [shift={(91,212)}, rotate = 79.46] [color={rgb, 255:red, 0; green, 0; blue, 0 }  ][line width=0.75]    (10.93,-3.29) .. controls (6.95,-1.4) and (3.31,-0.3) .. (0,0) .. controls (3.31,0.3) and (6.95,1.4) .. (10.93,3.29)   ;
\draw    (141,285) -- (137.25,254.98) ;
\draw [shift={(137,253)}, rotate = 82.87] [color={rgb, 255:red, 0; green, 0; blue, 0 }  ][line width=0.75]    (10.93,-3.29) .. controls (6.95,-1.4) and (3.31,-0.3) .. (0,0) .. controls (3.31,0.3) and (6.95,1.4) .. (10.93,3.29)   ;
\draw    (262,79.43) -- (241.51,61.74) ;
\draw [shift={(240,60.43)}, rotate = 40.82] [color={rgb, 255:red, 0; green, 0; blue, 0 }  ][line width=0.75]    (10.93,-3.29) .. controls (6.95,-1.4) and (3.31,-0.3) .. (0,0) .. controls (3.31,0.3) and (6.95,1.4) .. (10.93,3.29)   ;
\draw    (219,62.43) -- (194.27,81.2) ;
\draw [shift={(192.68,82.41)}, rotate = 322.81] [color={rgb, 255:red, 0; green, 0; blue, 0 }  ][line width=0.75]    (10.93,-3.29) .. controls (6.95,-1.4) and (3.31,-0.3) .. (0,0) .. controls (3.31,0.3) and (6.95,1.4) .. (10.93,3.29)   ;
\draw    (243.72,201.36) -- (236.88,174.48) ;
\draw [shift={(236.39,172.54)}, rotate = 75.73] [color={rgb, 255:red, 0; green, 0; blue, 0 }  ][line width=0.75]    (10.93,-3.29) .. controls (6.95,-1.4) and (3.31,-0.3) .. (0,0) .. controls (3.31,0.3) and (6.95,1.4) .. (10.93,3.29)   ;
\draw    (154.04,197.97) -- (102,199.38) ;
\draw [shift={(100,199.43)}, rotate = 358.45] [color={rgb, 255:red, 0; green, 0; blue, 0 }  ][line width=0.75]    (10.93,-3.29) .. controls (6.95,-1.4) and (3.31,-0.3) .. (0,0) .. controls (3.31,0.3) and (6.95,1.4) .. (10.93,3.29)   ;
\draw    (127.53,242.62) -- (95.38,208.88) ;
\draw [shift={(94,207.43)}, rotate = 46.38] [color={rgb, 255:red, 0; green, 0; blue, 0 }  ][line width=0.75]    (10.93,-3.29) .. controls (6.95,-1.4) and (3.31,-0.3) .. (0,0) .. controls (3.31,0.3) and (6.95,1.4) .. (10.93,3.29)   ;
\draw    (200.01,241.49) -- (236.5,213.32) ;
\draw [shift={(238.08,212.1)}, rotate = 142.34] [color={rgb, 255:red, 0; green, 0; blue, 0 }  ][line width=0.75]    (10.93,-3.29) .. controls (6.95,-1.4) and (3.31,-0.3) .. (0,0) .. controls (3.31,0.3) and (6.95,1.4) .. (10.93,3.29)   ;
\draw    (136,141.43) -- (158.16,189.62) ;
\draw [shift={(159,191.43)}, rotate = 245.3] [color={rgb, 255:red, 0; green, 0; blue, 0 }  ][line width=0.75]    (10.93,-3.29) .. controls (6.95,-1.4) and (3.31,-0.3) .. (0,0) .. controls (3.31,0.3) and (6.95,1.4) .. (10.93,3.29)   ;
\draw    (259.05,120.55) .. controls (277.75,122.2) and (284.65,106.67) .. (260.94,111.65) ;
\draw [shift={(259.05,112.07)}, rotate = 346.53] [color={rgb, 255:red, 0; green, 0; blue, 0 }  ][line width=0.75]    (10.93,-3.29) .. controls (6.95,-1.4) and (3.31,-0.3) .. (0,0) .. controls (3.31,0.3) and (6.95,1.4) .. (10.93,3.29)   ;
\draw    (189.01,94.56) .. controls (192.88,107.29) and (207.58,118.4) .. (233.65,118.31) ;
\draw [shift={(235.26,118.29)}, rotate = 178.8] [color={rgb, 255:red, 0; green, 0; blue, 0 }  ][line width=0.75]    (10.93,-3.29) .. controls (6.95,-1.4) and (3.31,-0.3) .. (0,0) .. controls (3.31,0.3) and (6.95,1.4) .. (10.93,3.29)   ;
\draw    (87,152) -- (120.64,136.1) ;
\draw [shift={(122.45,135.24)}, rotate = 154.7] [color={rgb, 255:red, 0; green, 0; blue, 0 }  ][line width=0.75]    (10.93,-3.29) .. controls (6.95,-1.4) and (3.31,-0.3) .. (0,0) .. controls (3.31,0.3) and (6.95,1.4) .. (10.93,3.29)   ;
\draw    (355,165.43) .. controls (363.6,171.16) and (367.53,184.33) .. (361.48,196.82) ;
\draw [shift={(360.56,198.58)}, rotate = 299.43] [color={rgb, 255:red, 0; green, 0; blue, 0 }  ][line width=0.75]    (10.93,-3.29) .. controls (6.95,-1.4) and (3.31,-0.3) .. (0,0) .. controls (3.31,0.3) and (6.95,1.4) .. (10.93,3.29)   ;
\draw    (350.12,201.41) .. controls (345.1,190.23) and (340.2,192.89) .. (338.12,166.87) ;
\draw [shift={(338,165.24)}, rotate = 85.87] [color={rgb, 255:red, 0; green, 0; blue, 0 }  ][line width=0.75]    (10.93,-3.29) .. controls (6.95,-1.4) and (3.31,-0.3) .. (0,0) .. controls (3.31,0.3) and (6.95,1.4) .. (10.93,3.29)   ;
\draw    (415,249.43) -- (454,250.38) ;
\draw [shift={(456,250.43)}, rotate = 181.4] [color={rgb, 255:red, 0; green, 0; blue, 0 }  ][line width=0.75]    (10.93,-3.29) .. controls (6.95,-1.4) and (3.31,-0.3) .. (0,0) .. controls (3.31,0.3) and (6.95,1.4) .. (10.93,3.29)   ;
\draw    (427.68,208.75) -- (406.28,238.77) ;
\draw [shift={(405.12,240.4)}, rotate = 305.49] [color={rgb, 255:red, 0; green, 0; blue, 0 }  ][line width=0.75]    (10.93,-3.29) .. controls (6.95,-1.4) and (3.31,-0.3) .. (0,0) .. controls (3.31,0.3) and (6.95,1.4) .. (10.93,3.29)   ;
\draw    (460,244.43) -- (434.2,210.03) ;
\draw [shift={(433,208.43)}, rotate = 53.13] [color={rgb, 255:red, 0; green, 0; blue, 0 }  ][line width=0.75]    (10.93,-3.29) .. controls (6.95,-1.4) and (3.31,-0.3) .. (0,0) .. controls (3.31,0.3) and (6.95,1.4) .. (10.93,3.29)   ;
\draw    (502.7,163.54) -- (510.69,130.45) ;
\draw [shift={(511.16,128.51)}, rotate = 103.58] [color={rgb, 255:red, 0; green, 0; blue, 0 }  ][line width=0.75]    (10.93,-3.29) .. controls (6.95,-1.4) and (3.31,-0.3) .. (0,0) .. controls (3.31,0.3) and (6.95,1.4) .. (10.93,3.29)   ;
\draw    (502.13,116.64) -- (464.38,94.48) ;
\draw [shift={(462.65,93.47)}, rotate = 30.41] [color={rgb, 255:red, 0; green, 0; blue, 0 }  ][line width=0.75]    (10.93,-3.29) .. controls (6.95,-1.4) and (3.31,-0.3) .. (0,0) .. controls (3.31,0.3) and (6.95,1.4) .. (10.93,3.29)   ;
\draw    (447.99,98.56) -- (409.54,130.16) ;
\draw [shift={(408,131.43)}, rotate = 320.57] [color={rgb, 255:red, 0; green, 0; blue, 0 }  ][line width=0.75]    (10.93,-3.29) .. controls (6.95,-1.4) and (3.31,-0.3) .. (0,0) .. controls (3.31,0.3) and (6.95,1.4) .. (10.93,3.29)   ;
\draw    (410,140.43) -- (490.12,172.41) ;
\draw [shift={(491.98,173.15)}, rotate = 201.76] [color={rgb, 255:red, 0; green, 0; blue, 0 }  ][line width=0.75]    (10.93,-3.29) .. controls (6.95,-1.4) and (3.31,-0.3) .. (0,0) .. controls (3.31,0.3) and (6.95,1.4) .. (10.93,3.29)   ;
\draw    (324,126.43) -- (339.06,154.67) ;
\draw [shift={(340,156.43)}, rotate = 241.93] [color={rgb, 255:red, 0; green, 0; blue, 0 }  ][line width=0.75]    (10.93,-3.29) .. controls (6.95,-1.4) and (3.31,-0.3) .. (0,0) .. controls (3.31,0.3) and (6.95,1.4) .. (10.93,3.29)   ;
\draw    (339.69,235.88) -- (351.18,213.64) ;
\draw [shift={(352.1,211.86)}, rotate = 117.32] [color={rgb, 255:red, 0; green, 0; blue, 0 }  ][line width=0.75]    (10.93,-3.29) .. controls (6.95,-1.4) and (3.31,-0.3) .. (0,0) .. controls (3.31,0.3) and (6.95,1.4) .. (10.93,3.29)   ;
\draw    (374.66,276.85) -- (395.29,254.86) ;
\draw [shift={(396.66,253.4)}, rotate = 133.17] [color={rgb, 255:red, 0; green, 0; blue, 0 }  ][line width=0.75]    (10.93,-3.29) .. controls (6.95,-1.4) and (3.31,-0.3) .. (0,0) .. controls (3.31,0.3) and (6.95,1.4) .. (10.93,3.29)   ;
\draw    (530,83.43) -- (509.51,65.74) ;
\draw [shift={(508,64.43)}, rotate = 40.82] [color={rgb, 255:red, 0; green, 0; blue, 0 }  ][line width=0.75]    (10.93,-3.29) .. controls (6.95,-1.4) and (3.31,-0.3) .. (0,0) .. controls (3.31,0.3) and (6.95,1.4) .. (10.93,3.29)   ;
\draw    (487,66.43) -- (462.27,85.2) ;
\draw [shift={(460.68,86.41)}, rotate = 322.81] [color={rgb, 255:red, 0; green, 0; blue, 0 }  ][line width=0.75]    (10.93,-3.29) .. controls (6.95,-1.4) and (3.31,-0.3) .. (0,0) .. controls (3.31,0.3) and (6.95,1.4) .. (10.93,3.29)   ;
\draw    (511.72,205.36) -- (504.88,178.48) ;
\draw [shift={(504.39,176.54)}, rotate = 75.73] [color={rgb, 255:red, 0; green, 0; blue, 0 }  ][line width=0.75]    (10.93,-3.29) .. controls (6.95,-1.4) and (3.31,-0.3) .. (0,0) .. controls (3.31,0.3) and (6.95,1.4) .. (10.93,3.29)   ;
\draw    (422.04,201.97) -- (370,203.38) ;
\draw [shift={(368,203.43)}, rotate = 358.45] [color={rgb, 255:red, 0; green, 0; blue, 0 }  ][line width=0.75]    (10.93,-3.29) .. controls (6.95,-1.4) and (3.31,-0.3) .. (0,0) .. controls (3.31,0.3) and (6.95,1.4) .. (10.93,3.29)   ;
\draw    (395.53,246.62) -- (363.38,212.88) ;
\draw [shift={(362,211.43)}, rotate = 46.38] [color={rgb, 255:red, 0; green, 0; blue, 0 }  ][line width=0.75]    (10.93,-3.29) .. controls (6.95,-1.4) and (3.31,-0.3) .. (0,0) .. controls (3.31,0.3) and (6.95,1.4) .. (10.93,3.29)   ;
\draw    (468.01,245.49) -- (504.5,217.32) ;
\draw [shift={(506.08,216.1)}, rotate = 142.34] [color={rgb, 255:red, 0; green, 0; blue, 0 }  ][line width=0.75]    (10.93,-3.29) .. controls (6.95,-1.4) and (3.31,-0.3) .. (0,0) .. controls (3.31,0.3) and (6.95,1.4) .. (10.93,3.29)   ;
\draw    (404,145.43) -- (426.16,193.62) ;
\draw [shift={(427,195.43)}, rotate = 245.3] [color={rgb, 255:red, 0; green, 0; blue, 0 }  ][line width=0.75]    (10.93,-3.29) .. controls (6.95,-1.4) and (3.31,-0.3) .. (0,0) .. controls (3.31,0.3) and (6.95,1.4) .. (10.93,3.29)   ;
\draw    (527.05,124.55) .. controls (545.75,126.2) and (552.65,110.67) .. (528.94,115.65) ;
\draw [shift={(527.05,116.07)}, rotate = 346.53] [color={rgb, 255:red, 0; green, 0; blue, 0 }  ][line width=0.75]    (10.93,-3.29) .. controls (6.95,-1.4) and (3.31,-0.3) .. (0,0) .. controls (3.31,0.3) and (6.95,1.4) .. (10.93,3.29)   ;
\draw    (457.01,98.56) .. controls (460.88,111.29) and (475.58,122.4) .. (501.65,122.31) ;
\draw [shift={(503.26,122.29)}, rotate = 178.8] [color={rgb, 255:red, 0; green, 0; blue, 0 }  ][line width=0.75]    (10.93,-3.29) .. controls (6.95,-1.4) and (3.31,-0.3) .. (0,0) .. controls (3.31,0.3) and (6.95,1.4) .. (10.93,3.29)   ;
\draw    (356,156) -- (388.65,140.12) ;
\draw [shift={(390.45,139.24)}, rotate = 154.07] [color={rgb, 255:red, 0; green, 0; blue, 0 }  ][line width=0.75]    (10.93,-3.29) .. controls (6.95,-1.4) and (3.31,-0.3) .. (0,0) .. controls (3.31,0.3) and (6.95,1.4) .. (10.93,3.29)   ;
\draw    (126,399.43) .. controls (134.6,405.16) and (138.53,418.33) .. (132.48,430.82) ;
\draw [shift={(131.56,432.58)}, rotate = 299.43] [color={rgb, 255:red, 0; green, 0; blue, 0 }  ][line width=0.75]    (10.93,-3.29) .. controls (6.95,-1.4) and (3.31,-0.3) .. (0,0) .. controls (3.31,0.3) and (6.95,1.4) .. (10.93,3.29)   ;
\draw    (121.12,435.41) .. controls (111.55,433.21) and (104.44,421.19) .. (108.58,401.12) ;
\draw [shift={(109,399.24)}, rotate = 103.27] [color={rgb, 255:red, 0; green, 0; blue, 0 }  ][line width=0.75]    (10.93,-3.29) .. controls (6.95,-1.4) and (3.31,-0.3) .. (0,0) .. controls (3.31,0.3) and (6.95,1.4) .. (10.93,3.29)   ;
\draw    (186,483.43) -- (225,484.38) ;
\draw [shift={(227,484.43)}, rotate = 181.4] [color={rgb, 255:red, 0; green, 0; blue, 0 }  ][line width=0.75]    (10.93,-3.29) .. controls (6.95,-1.4) and (3.31,-0.3) .. (0,0) .. controls (3.31,0.3) and (6.95,1.4) .. (10.93,3.29)   ;
\draw    (198.68,442.75) -- (177.28,472.77) ;
\draw [shift={(176.12,474.4)}, rotate = 305.49] [color={rgb, 255:red, 0; green, 0; blue, 0 }  ][line width=0.75]    (10.93,-3.29) .. controls (6.95,-1.4) and (3.31,-0.3) .. (0,0) .. controls (3.31,0.3) and (6.95,1.4) .. (10.93,3.29)   ;
\draw    (231,478.43) -- (205.2,444.03) ;
\draw [shift={(204,442.43)}, rotate = 53.13] [color={rgb, 255:red, 0; green, 0; blue, 0 }  ][line width=0.75]    (10.93,-3.29) .. controls (6.95,-1.4) and (3.31,-0.3) .. (0,0) .. controls (3.31,0.3) and (6.95,1.4) .. (10.93,3.29)   ;
\draw    (273.7,397.54) -- (281.69,364.45) ;
\draw [shift={(282.16,362.51)}, rotate = 103.58] [color={rgb, 255:red, 0; green, 0; blue, 0 }  ][line width=0.75]    (10.93,-3.29) .. controls (6.95,-1.4) and (3.31,-0.3) .. (0,0) .. controls (3.31,0.3) and (6.95,1.4) .. (10.93,3.29)   ;
\draw    (273.13,350.64) -- (235.38,328.48) ;
\draw [shift={(233.65,327.47)}, rotate = 30.41] [color={rgb, 255:red, 0; green, 0; blue, 0 }  ][line width=0.75]    (10.93,-3.29) .. controls (6.95,-1.4) and (3.31,-0.3) .. (0,0) .. controls (3.31,0.3) and (6.95,1.4) .. (10.93,3.29)   ;
\draw    (218.99,332.56) -- (180.54,364.16) ;
\draw [shift={(179,365.43)}, rotate = 320.57] [color={rgb, 255:red, 0; green, 0; blue, 0 }  ][line width=0.75]    (10.93,-3.29) .. controls (6.95,-1.4) and (3.31,-0.3) .. (0,0) .. controls (3.31,0.3) and (6.95,1.4) .. (10.93,3.29)   ;
\draw    (181,374.43) -- (261.12,406.41) ;
\draw [shift={(262.98,407.15)}, rotate = 201.76] [color={rgb, 255:red, 0; green, 0; blue, 0 }  ][line width=0.75]    (10.93,-3.29) .. controls (6.95,-1.4) and (3.31,-0.3) .. (0,0) .. controls (3.31,0.3) and (6.95,1.4) .. (10.93,3.29)   ;
\draw    (95,360.43) -- (110.06,388.67) ;
\draw [shift={(111,390.43)}, rotate = 241.93] [color={rgb, 255:red, 0; green, 0; blue, 0 }  ][line width=0.75]    (10.93,-3.29) .. controls (6.95,-1.4) and (3.31,-0.3) .. (0,0) .. controls (3.31,0.3) and (6.95,1.4) .. (10.93,3.29)   ;
\draw    (110.69,469.88) -- (122.18,447.64) ;
\draw [shift={(123.1,445.86)}, rotate = 117.32] [color={rgb, 255:red, 0; green, 0; blue, 0 }  ][line width=0.75]    (10.93,-3.29) .. controls (6.95,-1.4) and (3.31,-0.3) .. (0,0) .. controls (3.31,0.3) and (6.95,1.4) .. (10.93,3.29)   ;
\draw    (145.66,510.85) -- (166.29,488.86) ;
\draw [shift={(167.66,487.4)}, rotate = 133.17] [color={rgb, 255:red, 0; green, 0; blue, 0 }  ][line width=0.75]    (10.93,-3.29) .. controls (6.95,-1.4) and (3.31,-0.3) .. (0,0) .. controls (3.31,0.3) and (6.95,1.4) .. (10.93,3.29)   ;
\draw    (301,317.43) -- (280.51,299.74) ;
\draw [shift={(279,298.43)}, rotate = 40.82] [color={rgb, 255:red, 0; green, 0; blue, 0 }  ][line width=0.75]    (10.93,-3.29) .. controls (6.95,-1.4) and (3.31,-0.3) .. (0,0) .. controls (3.31,0.3) and (6.95,1.4) .. (10.93,3.29)   ;
\draw    (258,300.43) -- (233.27,319.2) ;
\draw [shift={(231.68,320.41)}, rotate = 322.81] [color={rgb, 255:red, 0; green, 0; blue, 0 }  ][line width=0.75]    (10.93,-3.29) .. controls (6.95,-1.4) and (3.31,-0.3) .. (0,0) .. controls (3.31,0.3) and (6.95,1.4) .. (10.93,3.29)   ;
\draw    (282.72,439.36) -- (275.88,412.48) ;
\draw [shift={(275.39,410.54)}, rotate = 75.73] [color={rgb, 255:red, 0; green, 0; blue, 0 }  ][line width=0.75]    (10.93,-3.29) .. controls (6.95,-1.4) and (3.31,-0.3) .. (0,0) .. controls (3.31,0.3) and (6.95,1.4) .. (10.93,3.29)   ;
\draw    (193.04,435.97) -- (141,437.38) ;
\draw [shift={(139,437.43)}, rotate = 358.45] [color={rgb, 255:red, 0; green, 0; blue, 0 }  ][line width=0.75]    (10.93,-3.29) .. controls (6.95,-1.4) and (3.31,-0.3) .. (0,0) .. controls (3.31,0.3) and (6.95,1.4) .. (10.93,3.29)   ;
\draw    (166.53,480.62) -- (134.38,446.88) ;
\draw [shift={(133,445.43)}, rotate = 46.38] [color={rgb, 255:red, 0; green, 0; blue, 0 }  ][line width=0.75]    (10.93,-3.29) .. controls (6.95,-1.4) and (3.31,-0.3) .. (0,0) .. controls (3.31,0.3) and (6.95,1.4) .. (10.93,3.29)   ;
\draw    (239.01,479.49) -- (275.5,451.32) ;
\draw [shift={(277.08,450.1)}, rotate = 142.34] [color={rgb, 255:red, 0; green, 0; blue, 0 }  ][line width=0.75]    (10.93,-3.29) .. controls (6.95,-1.4) and (3.31,-0.3) .. (0,0) .. controls (3.31,0.3) and (6.95,1.4) .. (10.93,3.29)   ;
\draw    (175,379.43) -- (197.16,427.62) ;
\draw [shift={(198,429.43)}, rotate = 245.3] [color={rgb, 255:red, 0; green, 0; blue, 0 }  ][line width=0.75]    (10.93,-3.29) .. controls (6.95,-1.4) and (3.31,-0.3) .. (0,0) .. controls (3.31,0.3) and (6.95,1.4) .. (10.93,3.29)   ;
\draw    (298.05,358.55) .. controls (316.75,360.2) and (323.65,344.67) .. (299.94,349.65) ;
\draw [shift={(298.05,350.07)}, rotate = 346.53] [color={rgb, 255:red, 0; green, 0; blue, 0 }  ][line width=0.75]    (10.93,-3.29) .. controls (6.95,-1.4) and (3.31,-0.3) .. (0,0) .. controls (3.31,0.3) and (6.95,1.4) .. (10.93,3.29)   ;
\draw    (228.01,332.56) .. controls (231.88,345.29) and (246.58,356.4) .. (272.65,356.31) ;
\draw [shift={(274.26,356.29)}, rotate = 178.8] [color={rgb, 255:red, 0; green, 0; blue, 0 }  ][line width=0.75]    (10.93,-3.29) .. controls (6.95,-1.4) and (3.31,-0.3) .. (0,0) .. controls (3.31,0.3) and (6.95,1.4) .. (10.93,3.29)   ;
\draw    (124,390) -- (159.63,374.06) ;
\draw [shift={(161.45,373.24)}, rotate = 155.9] [color={rgb, 255:red, 0; green, 0; blue, 0 }  ][line width=0.75]    (10.93,-3.29) .. controls (6.95,-1.4) and (3.31,-0.3) .. (0,0) .. controls (3.31,0.3) and (6.95,1.4) .. (10.93,3.29)   ;
\draw [line width=0.75]    (139,301) -- (117.51,383.07) ;
\draw [shift={(117,385)}, rotate = 284.68] [color={rgb, 255:red, 0; green, 0; blue, 0 }  ][line width=0.75]    (10.93,-3.29) .. controls (6.95,-1.4) and (3.31,-0.3) .. (0,0) .. controls (3.31,0.3) and (6.95,1.4) .. (10.93,3.29)   ;
\draw [line width=0.75]    (297,440) -- (457.67,259.49) ;
\draw [shift={(459,258)}, rotate = 131.67] [color={rgb, 255:red, 0; green, 0; blue, 0 }  ][line width=0.75]    (10.93,-3.29) .. controls (6.95,-1.4) and (3.31,-0.3) .. (0,0) .. controls (3.31,0.3) and (6.95,1.4) .. (10.93,3.29)   ;
\draw [line width=1.5]    (42,116) .. controls (14.14,201.57) and (14,316.84) .. (103.64,391.87) ;
\draw [shift={(105,393)}, rotate = 219.49] [color={rgb, 255:red, 0; green, 0; blue, 0 }  ][line width=1.5]    (14.21,-4.28) .. controls (9.04,-1.82) and (4.3,-0.39) .. (0,0) .. controls (4.3,0.39) and (9.04,1.82) .. (14.21,4.28)   ;
\draw [line width=1.5]    (85,342) .. controls (63.33,325.26) and (50.39,249.33) .. (85.37,208.83) ;
\draw [shift={(87,207)}, rotate = 132.77] [color={rgb, 255:red, 0; green, 0; blue, 0 }  ][line width=1.5]    (14.21,-4.28) .. controls (9.04,-1.82) and (4.3,-0.39) .. (0,0) .. controls (4.3,0.39) and (9.04,1.82) .. (14.21,4.28)   ;
\draw [line width=1.5]    (118,309) .. controls (241.75,379.29) and (264.55,225.14) .. (335.82,166.97) ;
\draw [shift={(338,165.24)}, rotate = 142.13] [color={rgb, 255:red, 0; green, 0; blue, 0 }  ][line width=1.5]    (14.21,-4.28) .. controls (9.04,-1.82) and (4.3,-0.39) .. (0,0) .. controls (4.3,0.39) and (9.04,1.82) .. (14.21,4.28)   ;
\draw [line width=1.5]    (286,86) -- (483.03,57.43) ;
\draw [shift={(486,57)}, rotate = 171.75] [color={rgb, 255:red, 0; green, 0; blue, 0 }  ][line width=1.5]    (14.21,-4.28) .. controls (9.04,-1.82) and (4.3,-0.39) .. (0,0) .. controls (4.3,0.39) and (9.04,1.82) .. (14.21,4.28)   ;
\draw [line width=1.5]    (314,121) .. controls (263.56,267.47) and (239.47,255.87) .. (209.77,245.91) ;
\draw [shift={(207,245)}, rotate = 17.88] [color={rgb, 255:red, 0; green, 0; blue, 0 }  ][line width=1.5]    (14.21,-4.28) .. controls (9.04,-1.82) and (4.3,-0.39) .. (0,0) .. controls (4.3,0.39) and (9.04,1.82) .. (14.21,4.28)   ;
\draw [line width=1.5]    (257,200) .. controls (268.38,135.74) and (292.75,84.51) .. (434.85,86.96) ;
\draw [shift={(437,87)}, rotate = 181.19] [color={rgb, 255:red, 0; green, 0; blue, 0 }  ][line width=1.5]    (14.21,-4.28) .. controls (9.04,-1.82) and (4.3,-0.39) .. (0,0) .. controls (4.3,0.39) and (9.04,1.82) .. (14.21,4.28)   ;
\draw [line width=1.5]    (554,90) .. controls (574,272) and (472,418) .. (247,489) ;
\draw [shift={(247,489)}, rotate = 342.49] [color={rgb, 255:red, 0; green, 0; blue, 0 }  ][line width=1.5]    (14.21,-4.28) .. controls (9.04,-1.82) and (4.3,-0.39) .. (0,0) .. controls (4.3,0.39) and (9.04,1.82) .. (14.21,4.28)   ;
\draw [line width=1.5]    (322,320) .. controls (370.51,323.96) and (393.54,328.9) .. (405.64,260.11) ;
\draw [shift={(406,258)}, rotate = 99.59] [color={rgb, 255:red, 0; green, 0; blue, 0 }  ][line width=1.5]    (14.21,-4.28) .. controls (9.04,-1.82) and (4.3,-0.39) .. (0,0) .. controls (4.3,0.39) and (9.04,1.82) .. (14.21,4.28)   ;

\draw (254.35,61.56) node [anchor=north west][inner sep=0.75pt]  [font=\scriptsize] [align=left] {a};
\draw (202.58,58.95) node [anchor=north west][inner sep=0.75pt]  [font=\scriptsize] [align=left] {a};
\draw (148.13,96.82) node [anchor=north west][inner sep=0.75pt]  [font=\scriptsize] [align=left] {a};
\draw (173.08,152.46) node [anchor=north west][inner sep=0.75pt]  [font=\scriptsize] [align=left] {a};
\draw (214.99,87.77) node [anchor=north west][inner sep=0.75pt]  [font=\scriptsize] [align=left] {a};
\draw (240.93,145.03) node [anchor=north west][inner sep=0.75pt]  [font=\scriptsize] [align=left] {a,b};
\draw (244.76,183.89) node [anchor=north west][inner sep=0.75pt]  [font=\scriptsize] [align=left] {a};
\draw (49.42,129.46) node [anchor=north west][inner sep=0.75pt]  [font=\scriptsize] [align=left] {a};
\draw (97.98,168.19) node [anchor=north west][inner sep=0.75pt]  [font=\scriptsize] [align=left] {a};
\draw (49.19,177.32) node [anchor=north west][inner sep=0.75pt]  [font=\scriptsize] [align=left] {a,b};
\draw (181.54,214.71) node [anchor=north west][inner sep=0.75pt]  [font=\scriptsize] [align=left] {a};
\draw (139.16,209.14) node [anchor=north west][inner sep=0.75pt]  [font=\scriptsize] [align=left] {a};
\draw (158.03,245.18) node [anchor=north west][inner sep=0.75pt]  [font=\scriptsize] [align=left] {a};
\draw (141.83,265.2) node [anchor=north west][inner sep=0.75pt]  [font=\scriptsize] [align=left] {a};
\draw (84.75,264.38) node [anchor=north west][inner sep=0.75pt]  [font=\scriptsize] [align=left] {a};
\draw (121.36,183.71) node [anchor=north west][inner sep=0.75pt]  [font=\scriptsize] [align=left] {b};
\draw (100.11,225.01) node [anchor=north west][inner sep=0.75pt]  [font=\scriptsize] [align=left] {b};
\draw (221.04,229.79) node [anchor=north west][inner sep=0.75pt]  [font=\scriptsize] [align=left] {b};
\draw (134.03,158.89) node [anchor=north west][inner sep=0.75pt]  [font=\scriptsize] [align=left] {b};
\draw (278.83,109.12) node [anchor=north west][inner sep=0.75pt]  [font=\scriptsize] [align=left] {b};
\draw (192.82,108.94) node [anchor=north west][inner sep=0.75pt]  [font=\scriptsize] [align=left] {b};
\draw (91.29,128.02) node [anchor=north west][inner sep=0.75pt]  [font=\scriptsize] [align=left] {b};
\draw (46.17,108.38) node [anchor=north west][inner sep=0.75pt]  [font=\scriptsize] [align=left] {1a};
\draw (71.43,151.81) node [anchor=north west][inner sep=0.75pt]  [font=\scriptsize] [align=left] {2a};
\draw (83.26,193.45) node [anchor=north west][inner sep=0.75pt]  [font=\scriptsize] [align=left] {3a};
\draw (100.09,302.88) node [anchor=north west][inner sep=0.75pt]  [font=\scriptsize] [align=left] {4a};
\draw (137.06,285) node [anchor=north west][inner sep=0.75pt]  [font=\scriptsize] [align=left] {5a};
\draw (130.39,238.96) node [anchor=north west][inner sep=0.75pt]  [font=\scriptsize] [align=left] {6a};
\draw (156.81,193.32) node [anchor=north west][inner sep=0.75pt]  [font=\scriptsize] [align=left] {7a};
\draw (191.57,241.66) node [anchor=north west][inner sep=0.75pt]  [font=\scriptsize] [align=left] {9a};
\draw (124.57,126.46) node [anchor=north west][inner sep=0.75pt]  [font=\scriptsize] [align=left] {8a};
\draw (224.35,160.98) node [anchor=north west][inner sep=0.75pt]  [font=\scriptsize] [align=left] {10a};
\draw (235.5,201.36) node [anchor=north west][inner sep=0.75pt]  [font=\scriptsize] [align=left] {12a};
\draw (174.19,79.99) node [anchor=north west][inner sep=0.75pt]  [font=\scriptsize] [align=left] {11a};
\draw (237.13,110.64) node [anchor=north west][inner sep=0.75pt]  [font=\scriptsize] [align=left] {13a};
\draw (217.84,49.87) node [anchor=north west][inner sep=0.75pt]  [font=\scriptsize] [align=left] {14a};
\draw (262.08,79.47) node [anchor=north west][inner sep=0.75pt]  [font=\scriptsize] [align=left] {15a};
\draw (522.35,66.56) node [anchor=north west][inner sep=0.75pt]  [font=\scriptsize] [align=left] {a};
\draw (470.58,62.95) node [anchor=north west][inner sep=0.75pt]  [font=\scriptsize] [align=left] {a};
\draw (416.13,100.82) node [anchor=north west][inner sep=0.75pt]  [font=\scriptsize] [align=left] {a};
\draw (441.08,156.46) node [anchor=north west][inner sep=0.75pt]  [font=\scriptsize] [align=left] {a};
\draw (482.99,91.77) node [anchor=north west][inner sep=0.75pt]  [font=\scriptsize] [align=left] {a};
\draw (508.93,149.03) node [anchor=north west][inner sep=0.75pt]  [font=\scriptsize] [align=left] {a,b};
\draw (512.76,187.89) node [anchor=north west][inner sep=0.75pt]  [font=\scriptsize] [align=left] {a};
\draw (318.42,137.46) node [anchor=north west][inner sep=0.75pt]  [font=\scriptsize] [align=left] {a};
\draw (365.98,172.19) node [anchor=north west][inner sep=0.75pt]  [font=\scriptsize] [align=left] {a};
\draw (322.19,185.32) node [anchor=north west][inner sep=0.75pt]  [font=\scriptsize] [align=left] {a,b};
\draw (449.54,218.71) node [anchor=north west][inner sep=0.75pt]  [font=\scriptsize] [align=left] {a};
\draw (407.16,213.14) node [anchor=north west][inner sep=0.75pt]  [font=\scriptsize] [align=left] {a};
\draw (426.03,249.18) node [anchor=north west][inner sep=0.75pt]  [font=\scriptsize] [align=left] {a};
\draw (387.66,268.12) node [anchor=north west][inner sep=0.75pt]  [font=\scriptsize] [align=left] {a};
\draw (333.75,217.38) node [anchor=north west][inner sep=0.75pt]  [font=\scriptsize] [align=left] {a};
\draw (389.36,187.71) node [anchor=north west][inner sep=0.75pt]  [font=\scriptsize] [align=left] {b};
\draw (368.11,229.01) node [anchor=north west][inner sep=0.75pt]  [font=\scriptsize] [align=left] {b};
\draw (489.04,233.79) node [anchor=north west][inner sep=0.75pt]  [font=\scriptsize] [align=left] {b};
\draw (402.03,162.89) node [anchor=north west][inner sep=0.75pt]  [font=\scriptsize] [align=left] {b};
\draw (546.83,113.12) node [anchor=north west][inner sep=0.75pt]  [font=\scriptsize] [align=left] {b};
\draw (460.82,112.94) node [anchor=north west][inner sep=0.75pt]  [font=\scriptsize] [align=left] {b};
\draw (359.29,132.02) node [anchor=north west][inner sep=0.75pt]  [font=\scriptsize] [align=left] {b};
\draw (314.17,112.38) node [anchor=north west][inner sep=0.75pt]  [font=\scriptsize] [align=left] {1c};
\draw (339.43,155.81) node [anchor=north west][inner sep=0.75pt]  [font=\scriptsize] [align=left] {2c};
\draw (351.26,197.45) node [anchor=north west][inner sep=0.75pt]  [font=\scriptsize] [align=left] {3c};
\draw (332.09,235.88) node [anchor=north west][inner sep=0.75pt]  [font=\scriptsize] [align=left] {4c};
\draw (366.06,276) node [anchor=north west][inner sep=0.75pt]  [font=\scriptsize] [align=left] {5c};
\draw (398.39,242.96) node [anchor=north west][inner sep=0.75pt]  [font=\scriptsize] [align=left] {6c};
\draw (423,196.43) node [anchor=north west][inner sep=0.75pt]  [font=\scriptsize] [align=left] {7c};
\draw (457,245.43) node [anchor=north west][inner sep=0.75pt]  [font=\scriptsize] [align=left] {9c};
\draw (392.57,130.46) node [anchor=north west][inner sep=0.75pt]  [font=\scriptsize] [align=left] {8c};
\draw (492.35,164.98) node [anchor=north west][inner sep=0.75pt]  [font=\scriptsize] [align=left] {10c};
\draw (503.5,205.36) node [anchor=north west][inner sep=0.75pt]  [font=\scriptsize] [align=left] {12c};
\draw (442.19,83.99) node [anchor=north west][inner sep=0.75pt]  [font=\scriptsize] [align=left] {11c};
\draw (505.13,114.64) node [anchor=north west][inner sep=0.75pt]  [font=\scriptsize] [align=left] {13c};
\draw (485.84,53.87) node [anchor=north west][inner sep=0.75pt]  [font=\scriptsize] [align=left] {14c};
\draw (530.08,83.47) node [anchor=north west][inner sep=0.75pt]  [font=\scriptsize] [align=left] {15c};
\draw (292.35,298.56) node [anchor=north west][inner sep=0.75pt]  [font=\scriptsize] [align=left] {a};
\draw (241.58,296.95) node [anchor=north west][inner sep=0.75pt]  [font=\scriptsize] [align=left] {a};
\draw (187.13,334.82) node [anchor=north west][inner sep=0.75pt]  [font=\scriptsize] [align=left] {a};
\draw (212.08,390.46) node [anchor=north west][inner sep=0.75pt]  [font=\scriptsize] [align=left] {a};
\draw (253.99,325.77) node [anchor=north west][inner sep=0.75pt]  [font=\scriptsize] [align=left] {a};
\draw (279.93,383.03) node [anchor=north west][inner sep=0.75pt]  [font=\scriptsize] [align=left] {a,b};
\draw (283.76,421.89) node [anchor=north west][inner sep=0.75pt]  [font=\scriptsize] [align=left] {a};
\draw (88.42,367.46) node [anchor=north west][inner sep=0.75pt]  [font=\scriptsize] [align=left] {a};
\draw (136.98,406.19) node [anchor=north west][inner sep=0.75pt]  [font=\scriptsize] [align=left] {a};
\draw (88.19,415.32) node [anchor=north west][inner sep=0.75pt]  [font=\scriptsize] [align=left] {a,b};
\draw (220.54,452.71) node [anchor=north west][inner sep=0.75pt]  [font=\scriptsize] [align=left] {a};
\draw (178.16,447.14) node [anchor=north west][inner sep=0.75pt]  [font=\scriptsize] [align=left] {a};
\draw (197.03,483.18) node [anchor=north west][inner sep=0.75pt]  [font=\scriptsize] [align=left] {a};
\draw (158.66,502.12) node [anchor=north west][inner sep=0.75pt]  [font=\scriptsize] [align=left] {a};
\draw (104.75,451.38) node [anchor=north west][inner sep=0.75pt]  [font=\scriptsize] [align=left] {a};
\draw (160.36,421.71) node [anchor=north west][inner sep=0.75pt]  [font=\scriptsize] [align=left] {b};
\draw (139.11,463.01) node [anchor=north west][inner sep=0.75pt]  [font=\scriptsize] [align=left] {b};
\draw (260.04,467.79) node [anchor=north west][inner sep=0.75pt]  [font=\scriptsize] [align=left] {b};
\draw (173.03,396.89) node [anchor=north west][inner sep=0.75pt]  [font=\scriptsize] [align=left] {b};
\draw (317.83,347.12) node [anchor=north west][inner sep=0.75pt]  [font=\scriptsize] [align=left] {b};
\draw (231.82,346.94) node [anchor=north west][inner sep=0.75pt]  [font=\scriptsize] [align=left] {b};
\draw (130.29,366.02) node [anchor=north west][inner sep=0.75pt]  [font=\scriptsize] [align=left] {b};
\draw (85.17,346.38) node [anchor=north west][inner sep=0.75pt]  [font=\scriptsize] [align=left] {1b};
\draw (110.43,389.81) node [anchor=north west][inner sep=0.75pt]  [font=\scriptsize] [align=left] {2b};
\draw (122.26,431.45) node [anchor=north west][inner sep=0.75pt]  [font=\scriptsize] [align=left] {3b};
\draw (103.09,469.88) node [anchor=north west][inner sep=0.75pt]  [font=\scriptsize] [align=left] {4b};
\draw (137.06,510) node [anchor=north west][inner sep=0.75pt]  [font=\scriptsize] [align=left] {5b};
\draw (169.39,476.96) node [anchor=north west][inner sep=0.75pt]  [font=\scriptsize] [align=left] {6b};
\draw (195.81,431.32) node [anchor=north west][inner sep=0.75pt]  [font=\scriptsize] [align=left] {7b};
\draw (230.57,479.66) node [anchor=north west][inner sep=0.75pt]  [font=\scriptsize] [align=left] {9b};
\draw (163.57,364.46) node [anchor=north west][inner sep=0.75pt]  [font=\scriptsize] [align=left] {8b};
\draw (263.35,398.98) node [anchor=north west][inner sep=0.75pt]  [font=\scriptsize] [align=left] {10b};
\draw (274.5,439.36) node [anchor=north west][inner sep=0.75pt]  [font=\scriptsize] [align=left] {12b};
\draw (213.19,317.99) node [anchor=north west][inner sep=0.75pt]  [font=\scriptsize] [align=left] {11b};
\draw (276.13,348.64) node [anchor=north west][inner sep=0.75pt]  [font=\scriptsize] [align=left] {13b};
\draw (256.84,287.87) node [anchor=north west][inner sep=0.75pt]  [font=\scriptsize] [align=left] {14b};
\draw (301.08,317.47) node [anchor=north west][inner sep=0.75pt]  [font=\scriptsize] [align=left] {15b};
\draw (118,330) node [anchor=north west][inner sep=0.75pt]  [font=\scriptsize] [align=left] {0};
\draw (380,352) node [anchor=north west][inner sep=0.75pt]  [font=\scriptsize] [align=left] {1};
\draw (19,270) node [anchor=north west][inner sep=0.75pt]  [font=\scriptsize] [align=left] {b};
\draw (46,274) node [anchor=north west][inner sep=0.75pt]  [font=\scriptsize] [align=left] {b};
\draw (275,251) node [anchor=north west][inner sep=0.75pt]  [font=\scriptsize] [align=left] {b};
\draw (371,57) node [anchor=north west][inner sep=0.75pt]  [font=\scriptsize] [align=left] {b};
\draw (278,210) node [anchor=north west][inner sep=0.75pt]  [font=\scriptsize] [align=left] {b};
\draw (303,94) node [anchor=north west][inner sep=0.75pt]  [font=\scriptsize] [align=left] {b};
\draw (501,329) node [anchor=north west][inner sep=0.75pt]  [font=\scriptsize] [align=left] {b};
\draw (380,317) node [anchor=north west][inner sep=0.75pt]  [font=\scriptsize] [align=left] {b};

\end{tikzpicture}

%% file: fourth_step.tex
\tikzset{every picture/.style={line width=0.75pt}} 

\begin{tikzpicture}[x=0.75pt,y=0.75pt,yscale=-0.9,xscale=0.9]

\draw    (85,133.43) .. controls (93.6,139.16) and (97.53,152.33) .. (91.48,164.82) ;
\draw [shift={(90.56,166.58)}, rotate = 299.43] [color={rgb, 255:red, 0; green, 0; blue, 0 }  ][line width=0.75]    (10.93,-3.29) .. controls (6.95,-1.4) and (3.31,-0.3) .. (0,0) .. controls (3.31,0.3) and (6.95,1.4) .. (10.93,3.29)   ;
\draw    (80.12,169.41) .. controls (70.55,167.21) and (63.44,155.19) .. (67.58,135.12) ;
\draw [shift={(68,133.24)}, rotate = 103.27] [color={rgb, 255:red, 0; green, 0; blue, 0 }  ][line width=0.75]    (10.93,-3.29) .. controls (6.95,-1.4) and (3.31,-0.3) .. (0,0) .. controls (3.31,0.3) and (6.95,1.4) .. (10.93,3.29)   ;
\draw    (145,217.43) -- (184,218.38) ;
\draw [shift={(186,218.43)}, rotate = 181.4] [color={rgb, 255:red, 0; green, 0; blue, 0 }  ][line width=0.75]    (10.93,-3.29) .. controls (6.95,-1.4) and (3.31,-0.3) .. (0,0) .. controls (3.31,0.3) and (6.95,1.4) .. (10.93,3.29)   ;
\draw    (157.68,176.75) -- (136.28,206.77) ;
\draw [shift={(135.12,208.4)}, rotate = 305.49] [color={rgb, 255:red, 0; green, 0; blue, 0 }  ][line width=0.75]    (10.93,-3.29) .. controls (6.95,-1.4) and (3.31,-0.3) .. (0,0) .. controls (3.31,0.3) and (6.95,1.4) .. (10.93,3.29)   ;
\draw    (190,212.43) -- (164.2,178.03) ;
\draw [shift={(163,176.43)}, rotate = 53.13] [color={rgb, 255:red, 0; green, 0; blue, 0 }  ][line width=0.75]    (10.93,-3.29) .. controls (6.95,-1.4) and (3.31,-0.3) .. (0,0) .. controls (3.31,0.3) and (6.95,1.4) .. (10.93,3.29)   ;
\draw    (232.7,131.54) -- (240.69,98.45) ;
\draw [shift={(241.16,96.51)}, rotate = 103.58] [color={rgb, 255:red, 0; green, 0; blue, 0 }  ][line width=0.75]    (10.93,-3.29) .. controls (6.95,-1.4) and (3.31,-0.3) .. (0,0) .. controls (3.31,0.3) and (6.95,1.4) .. (10.93,3.29)   ;
\draw    (232.13,84.64) -- (194.38,62.48) ;
\draw [shift={(192.65,61.47)}, rotate = 30.41] [color={rgb, 255:red, 0; green, 0; blue, 0 }  ][line width=0.75]    (10.93,-3.29) .. controls (6.95,-1.4) and (3.31,-0.3) .. (0,0) .. controls (3.31,0.3) and (6.95,1.4) .. (10.93,3.29)   ;
\draw    (177.99,66.56) -- (139.54,98.16) ;
\draw [shift={(138,99.43)}, rotate = 320.57] [color={rgb, 255:red, 0; green, 0; blue, 0 }  ][line width=0.75]    (10.93,-3.29) .. controls (6.95,-1.4) and (3.31,-0.3) .. (0,0) .. controls (3.31,0.3) and (6.95,1.4) .. (10.93,3.29)   ;
\draw    (140,108.43) -- (220.12,140.41) ;
\draw [shift={(221.98,141.15)}, rotate = 201.76] [color={rgb, 255:red, 0; green, 0; blue, 0 }  ][line width=0.75]    (10.93,-3.29) .. controls (6.95,-1.4) and (3.31,-0.3) .. (0,0) .. controls (3.31,0.3) and (6.95,1.4) .. (10.93,3.29)   ;
\draw    (54,94.43) -- (69.06,122.67) ;
\draw [shift={(70,124.43)}, rotate = 241.93] [color={rgb, 255:red, 0; green, 0; blue, 0 }  ][line width=0.75]    (10.93,-3.29) .. controls (6.95,-1.4) and (3.31,-0.3) .. (0,0) .. controls (3.31,0.3) and (6.95,1.4) .. (10.93,3.29)   ;
\draw    (105,270) -- (89.37,185.97) ;
\draw [shift={(89,184)}, rotate = 79.46] [color={rgb, 255:red, 0; green, 0; blue, 0 }  ][line width=0.75]    (10.93,-3.29) .. controls (6.95,-1.4) and (3.31,-0.3) .. (0,0) .. controls (3.31,0.3) and (6.95,1.4) .. (10.93,3.29)   ;
\draw    (139,257) -- (135.25,226.98) ;
\draw [shift={(135,225)}, rotate = 82.87] [color={rgb, 255:red, 0; green, 0; blue, 0 }  ][line width=0.75]    (10.93,-3.29) .. controls (6.95,-1.4) and (3.31,-0.3) .. (0,0) .. controls (3.31,0.3) and (6.95,1.4) .. (10.93,3.29)   ;
\draw    (260,51.43) -- (239.51,33.74) ;
\draw [shift={(238,32.43)}, rotate = 40.82] [color={rgb, 255:red, 0; green, 0; blue, 0 }  ][line width=0.75]    (10.93,-3.29) .. controls (6.95,-1.4) and (3.31,-0.3) .. (0,0) .. controls (3.31,0.3) and (6.95,1.4) .. (10.93,3.29)   ;
\draw    (217,34.43) -- (192.27,53.2) ;
\draw [shift={(190.68,54.41)}, rotate = 322.81] [color={rgb, 255:red, 0; green, 0; blue, 0 }  ][line width=0.75]    (10.93,-3.29) .. controls (6.95,-1.4) and (3.31,-0.3) .. (0,0) .. controls (3.31,0.3) and (6.95,1.4) .. (10.93,3.29)   ;
\draw    (241.72,173.36) -- (234.88,146.48) ;
\draw [shift={(234.39,144.54)}, rotate = 75.73] [color={rgb, 255:red, 0; green, 0; blue, 0 }  ][line width=0.75]    (10.93,-3.29) .. controls (6.95,-1.4) and (3.31,-0.3) .. (0,0) .. controls (3.31,0.3) and (6.95,1.4) .. (10.93,3.29)   ;
\draw    (152.04,169.97) -- (100,171.38) ;
\draw [shift={(98,171.43)}, rotate = 358.45] [color={rgb, 255:red, 0; green, 0; blue, 0 }  ][line width=0.75]    (10.93,-3.29) .. controls (6.95,-1.4) and (3.31,-0.3) .. (0,0) .. controls (3.31,0.3) and (6.95,1.4) .. (10.93,3.29)   ;
\draw    (125.53,214.62) -- (93.38,180.88) ;
\draw [shift={(92,179.43)}, rotate = 46.38] [color={rgb, 255:red, 0; green, 0; blue, 0 }  ][line width=0.75]    (10.93,-3.29) .. controls (6.95,-1.4) and (3.31,-0.3) .. (0,0) .. controls (3.31,0.3) and (6.95,1.4) .. (10.93,3.29)   ;
\draw    (198.01,213.49) -- (234.5,185.32) ;
\draw [shift={(236.08,184.1)}, rotate = 142.34] [color={rgb, 255:red, 0; green, 0; blue, 0 }  ][line width=0.75]    (10.93,-3.29) .. controls (6.95,-1.4) and (3.31,-0.3) .. (0,0) .. controls (3.31,0.3) and (6.95,1.4) .. (10.93,3.29)   ;
\draw    (134,113.43) -- (156.16,161.62) ;
\draw [shift={(157,163.43)}, rotate = 245.3] [color={rgb, 255:red, 0; green, 0; blue, 0 }  ][line width=0.75]    (10.93,-3.29) .. controls (6.95,-1.4) and (3.31,-0.3) .. (0,0) .. controls (3.31,0.3) and (6.95,1.4) .. (10.93,3.29)   ;
\draw    (257.05,92.55) .. controls (275.75,94.2) and (282.65,78.67) .. (258.94,83.65) ;
\draw [shift={(257.05,84.07)}, rotate = 346.53] [color={rgb, 255:red, 0; green, 0; blue, 0 }  ][line width=0.75]    (10.93,-3.29) .. controls (6.95,-1.4) and (3.31,-0.3) .. (0,0) .. controls (3.31,0.3) and (6.95,1.4) .. (10.93,3.29)   ;
\draw    (187.01,66.56) .. controls (190.88,79.29) and (205.58,90.4) .. (231.65,90.31) ;
\draw [shift={(233.26,90.29)}, rotate = 178.8] [color={rgb, 255:red, 0; green, 0; blue, 0 }  ][line width=0.75]    (10.93,-3.29) .. controls (6.95,-1.4) and (3.31,-0.3) .. (0,0) .. controls (3.31,0.3) and (6.95,1.4) .. (10.93,3.29)   ;
\draw    (85,124) -- (118.64,108.1) ;
\draw [shift={(120.45,107.24)}, rotate = 154.7] [color={rgb, 255:red, 0; green, 0; blue, 0 }  ][line width=0.75]    (10.93,-3.29) .. controls (6.95,-1.4) and (3.31,-0.3) .. (0,0) .. controls (3.31,0.3) and (6.95,1.4) .. (10.93,3.29)   ;
\draw    (353,137.43) .. controls (361.6,143.16) and (365.53,156.33) .. (359.48,168.82) ;
\draw [shift={(358.56,170.58)}, rotate = 299.43] [color={rgb, 255:red, 0; green, 0; blue, 0 }  ][line width=0.75]    (10.93,-3.29) .. controls (6.95,-1.4) and (3.31,-0.3) .. (0,0) .. controls (3.31,0.3) and (6.95,1.4) .. (10.93,3.29)   ;
\draw    (348.12,173.41) .. controls (343.1,162.23) and (338.2,164.89) .. (336.12,138.87) ;
\draw [shift={(336,137.24)}, rotate = 85.87] [color={rgb, 255:red, 0; green, 0; blue, 0 }  ][line width=0.75]    (10.93,-3.29) .. controls (6.95,-1.4) and (3.31,-0.3) .. (0,0) .. controls (3.31,0.3) and (6.95,1.4) .. (10.93,3.29)   ;
\draw    (413,221.43) -- (452,222.38) ;
\draw [shift={(454,222.43)}, rotate = 181.4] [color={rgb, 255:red, 0; green, 0; blue, 0 }  ][line width=0.75]    (10.93,-3.29) .. controls (6.95,-1.4) and (3.31,-0.3) .. (0,0) .. controls (3.31,0.3) and (6.95,1.4) .. (10.93,3.29)   ;
\draw    (425.68,180.75) -- (404.28,210.77) ;
\draw [shift={(403.12,212.4)}, rotate = 305.49] [color={rgb, 255:red, 0; green, 0; blue, 0 }  ][line width=0.75]    (10.93,-3.29) .. controls (6.95,-1.4) and (3.31,-0.3) .. (0,0) .. controls (3.31,0.3) and (6.95,1.4) .. (10.93,3.29)   ;
\draw    (458,216.43) -- (432.2,182.03) ;
\draw [shift={(431,180.43)}, rotate = 53.13] [color={rgb, 255:red, 0; green, 0; blue, 0 }  ][line width=0.75]    (10.93,-3.29) .. controls (6.95,-1.4) and (3.31,-0.3) .. (0,0) .. controls (3.31,0.3) and (6.95,1.4) .. (10.93,3.29)   ;
\draw    (500.7,135.54) -- (508.69,102.45) ;
\draw [shift={(509.16,100.51)}, rotate = 103.58] [color={rgb, 255:red, 0; green, 0; blue, 0 }  ][line width=0.75]    (10.93,-3.29) .. controls (6.95,-1.4) and (3.31,-0.3) .. (0,0) .. controls (3.31,0.3) and (6.95,1.4) .. (10.93,3.29)   ;
\draw    (500.13,88.64) -- (462.38,66.48) ;
\draw [shift={(460.65,65.47)}, rotate = 30.41] [color={rgb, 255:red, 0; green, 0; blue, 0 }  ][line width=0.75]    (10.93,-3.29) .. controls (6.95,-1.4) and (3.31,-0.3) .. (0,0) .. controls (3.31,0.3) and (6.95,1.4) .. (10.93,3.29)   ;
\draw    (445.99,70.56) -- (407.54,102.16) ;
\draw [shift={(406,103.43)}, rotate = 320.57] [color={rgb, 255:red, 0; green, 0; blue, 0 }  ][line width=0.75]    (10.93,-3.29) .. controls (6.95,-1.4) and (3.31,-0.3) .. (0,0) .. controls (3.31,0.3) and (6.95,1.4) .. (10.93,3.29)   ;
\draw    (408,109.43) -- (485.13,139.28) ;
\draw [shift={(487,140)}, rotate = 201.15] [color={rgb, 255:red, 0; green, 0; blue, 0 }  ][line width=0.75]    (10.93,-3.29) .. controls (6.95,-1.4) and (3.31,-0.3) .. (0,0) .. controls (3.31,0.3) and (6.95,1.4) .. (10.93,3.29)   ;
\draw    (322,98.43) -- (337.06,126.67) ;
\draw [shift={(338,128.43)}, rotate = 241.93] [color={rgb, 255:red, 0; green, 0; blue, 0 }  ][line width=0.75]    (10.93,-3.29) .. controls (6.95,-1.4) and (3.31,-0.3) .. (0,0) .. controls (3.31,0.3) and (6.95,1.4) .. (10.93,3.29)   ;
\draw    (337.69,207.88) -- (349.18,185.64) ;
\draw [shift={(350.1,183.86)}, rotate = 117.32] [color={rgb, 255:red, 0; green, 0; blue, 0 }  ][line width=0.75]    (10.93,-3.29) .. controls (6.95,-1.4) and (3.31,-0.3) .. (0,0) .. controls (3.31,0.3) and (6.95,1.4) .. (10.93,3.29)   ;
\draw    (372.66,248.85) -- (393.29,226.86) ;
\draw [shift={(394.66,225.4)}, rotate = 133.17] [color={rgb, 255:red, 0; green, 0; blue, 0 }  ][line width=0.75]    (10.93,-3.29) .. controls (6.95,-1.4) and (3.31,-0.3) .. (0,0) .. controls (3.31,0.3) and (6.95,1.4) .. (10.93,3.29)   ;
\draw    (528,55.43) -- (507.51,37.74) ;
\draw [shift={(506,36.43)}, rotate = 40.82] [color={rgb, 255:red, 0; green, 0; blue, 0 }  ][line width=0.75]    (10.93,-3.29) .. controls (6.95,-1.4) and (3.31,-0.3) .. (0,0) .. controls (3.31,0.3) and (6.95,1.4) .. (10.93,3.29)   ;
\draw    (485,38.43) -- (460.27,57.2) ;
\draw [shift={(458.68,58.41)}, rotate = 322.81] [color={rgb, 255:red, 0; green, 0; blue, 0 }  ][line width=0.75]    (10.93,-3.29) .. controls (6.95,-1.4) and (3.31,-0.3) .. (0,0) .. controls (3.31,0.3) and (6.95,1.4) .. (10.93,3.29)   ;
\draw    (509.72,177.36) -- (502.88,150.48) ;
\draw [shift={(502.39,148.54)}, rotate = 75.73] [color={rgb, 255:red, 0; green, 0; blue, 0 }  ][line width=0.75]    (10.93,-3.29) .. controls (6.95,-1.4) and (3.31,-0.3) .. (0,0) .. controls (3.31,0.3) and (6.95,1.4) .. (10.93,3.29)   ;
\draw    (420.04,173.97) -- (368,175.38) ;
\draw [shift={(366,175.43)}, rotate = 358.45] [color={rgb, 255:red, 0; green, 0; blue, 0 }  ][line width=0.75]    (10.93,-3.29) .. controls (6.95,-1.4) and (3.31,-0.3) .. (0,0) .. controls (3.31,0.3) and (6.95,1.4) .. (10.93,3.29)   ;
\draw    (393.53,218.62) -- (361.38,184.88) ;
\draw [shift={(360,183.43)}, rotate = 46.38] [color={rgb, 255:red, 0; green, 0; blue, 0 }  ][line width=0.75]    (10.93,-3.29) .. controls (6.95,-1.4) and (3.31,-0.3) .. (0,0) .. controls (3.31,0.3) and (6.95,1.4) .. (10.93,3.29)   ;
\draw    (466.01,217.49) -- (502.5,189.32) ;
\draw [shift={(504.08,188.1)}, rotate = 142.34] [color={rgb, 255:red, 0; green, 0; blue, 0 }  ][line width=0.75]    (10.93,-3.29) .. controls (6.95,-1.4) and (3.31,-0.3) .. (0,0) .. controls (3.31,0.3) and (6.95,1.4) .. (10.93,3.29)   ;
\draw    (402,117.43) -- (424.16,165.62) ;
\draw [shift={(425,167.43)}, rotate = 245.3] [color={rgb, 255:red, 0; green, 0; blue, 0 }  ][line width=0.75]    (10.93,-3.29) .. controls (6.95,-1.4) and (3.31,-0.3) .. (0,0) .. controls (3.31,0.3) and (6.95,1.4) .. (10.93,3.29)   ;
\draw    (525.05,96.55) .. controls (543.75,98.2) and (550.65,82.67) .. (526.94,87.65) ;
\draw [shift={(525.05,88.07)}, rotate = 346.53] [color={rgb, 255:red, 0; green, 0; blue, 0 }  ][line width=0.75]    (10.93,-3.29) .. controls (6.95,-1.4) and (3.31,-0.3) .. (0,0) .. controls (3.31,0.3) and (6.95,1.4) .. (10.93,3.29)   ;
\draw    (455.01,70.56) .. controls (458.88,83.29) and (473.58,94.4) .. (499.65,94.31) ;
\draw [shift={(501.26,94.29)}, rotate = 178.8] [color={rgb, 255:red, 0; green, 0; blue, 0 }  ][line width=0.75]    (10.93,-3.29) .. controls (6.95,-1.4) and (3.31,-0.3) .. (0,0) .. controls (3.31,0.3) and (6.95,1.4) .. (10.93,3.29)   ;
\draw    (354,128) -- (386.65,112.12) ;
\draw [shift={(388.45,111.24)}, rotate = 154.07] [color={rgb, 255:red, 0; green, 0; blue, 0 }  ][line width=0.75]    (10.93,-3.29) .. controls (6.95,-1.4) and (3.31,-0.3) .. (0,0) .. controls (3.31,0.3) and (6.95,1.4) .. (10.93,3.29)   ;
\draw    (124,371.43) .. controls (128.78,381.52) and (128.99,390.19) .. (129.49,402.77) ;
\draw [shift={(129.56,404.58)}, rotate = 267.64] [color={rgb, 255:red, 0; green, 0; blue, 0 }  ][line width=0.75]    (10.93,-3.29) .. controls (6.95,-1.4) and (3.31,-0.3) .. (0,0) .. controls (3.31,0.3) and (6.95,1.4) .. (10.93,3.29)   ;
\draw    (119.12,407.41) .. controls (109.55,405.21) and (102.44,393.19) .. (106.58,373.12) ;
\draw [shift={(107,371.24)}, rotate = 103.27] [color={rgb, 255:red, 0; green, 0; blue, 0 }  ][line width=0.75]    (10.93,-3.29) .. controls (6.95,-1.4) and (3.31,-0.3) .. (0,0) .. controls (3.31,0.3) and (6.95,1.4) .. (10.93,3.29)   ;
\draw    (184,455.43) -- (223,456.38) ;
\draw [shift={(225,456.43)}, rotate = 181.4] [color={rgb, 255:red, 0; green, 0; blue, 0 }  ][line width=0.75]    (10.93,-3.29) .. controls (6.95,-1.4) and (3.31,-0.3) .. (0,0) .. controls (3.31,0.3) and (6.95,1.4) .. (10.93,3.29)   ;
\draw    (196.68,414.75) -- (175.28,444.77) ;
\draw [shift={(174.12,446.4)}, rotate = 305.49] [color={rgb, 255:red, 0; green, 0; blue, 0 }  ][line width=0.75]    (10.93,-3.29) .. controls (6.95,-1.4) and (3.31,-0.3) .. (0,0) .. controls (3.31,0.3) and (6.95,1.4) .. (10.93,3.29)   ;
\draw    (229,450.43) -- (203.2,416.03) ;
\draw [shift={(202,414.43)}, rotate = 53.13] [color={rgb, 255:red, 0; green, 0; blue, 0 }  ][line width=0.75]    (10.93,-3.29) .. controls (6.95,-1.4) and (3.31,-0.3) .. (0,0) .. controls (3.31,0.3) and (6.95,1.4) .. (10.93,3.29)   ;
\draw    (271.7,369.54) -- (279.69,336.45) ;
\draw [shift={(280.16,334.51)}, rotate = 103.58] [color={rgb, 255:red, 0; green, 0; blue, 0 }  ][line width=0.75]    (10.93,-3.29) .. controls (6.95,-1.4) and (3.31,-0.3) .. (0,0) .. controls (3.31,0.3) and (6.95,1.4) .. (10.93,3.29)   ;
\draw    (271.13,322.64) -- (233.38,300.48) ;
\draw [shift={(231.65,299.47)}, rotate = 30.41] [color={rgb, 255:red, 0; green, 0; blue, 0 }  ][line width=0.75]    (10.93,-3.29) .. controls (6.95,-1.4) and (3.31,-0.3) .. (0,0) .. controls (3.31,0.3) and (6.95,1.4) .. (10.93,3.29)   ;
\draw    (216.99,304.56) -- (178.54,336.16) ;
\draw [shift={(177,337.43)}, rotate = 320.57] [color={rgb, 255:red, 0; green, 0; blue, 0 }  ][line width=0.75]    (10.93,-3.29) .. controls (6.95,-1.4) and (3.31,-0.3) .. (0,0) .. controls (3.31,0.3) and (6.95,1.4) .. (10.93,3.29)   ;
\draw    (179,346.43) -- (259.12,378.41) ;
\draw [shift={(260.98,379.15)}, rotate = 201.76] [color={rgb, 255:red, 0; green, 0; blue, 0 }  ][line width=0.75]    (10.93,-3.29) .. controls (6.95,-1.4) and (3.31,-0.3) .. (0,0) .. controls (3.31,0.3) and (6.95,1.4) .. (10.93,3.29)   ;
\draw    (93,332.43) -- (108.06,360.67) ;
\draw [shift={(109,362.43)}, rotate = 241.93] [color={rgb, 255:red, 0; green, 0; blue, 0 }  ][line width=0.75]    (10.93,-3.29) .. controls (6.95,-1.4) and (3.31,-0.3) .. (0,0) .. controls (3.31,0.3) and (6.95,1.4) .. (10.93,3.29)   ;
\draw    (108.69,441.88) -- (120.18,419.64) ;
\draw [shift={(121.1,417.86)}, rotate = 117.32] [color={rgb, 255:red, 0; green, 0; blue, 0 }  ][line width=0.75]    (10.93,-3.29) .. controls (6.95,-1.4) and (3.31,-0.3) .. (0,0) .. controls (3.31,0.3) and (6.95,1.4) .. (10.93,3.29)   ;
\draw    (143.66,482.85) -- (164.29,460.86) ;
\draw [shift={(165.66,459.4)}, rotate = 133.17] [color={rgb, 255:red, 0; green, 0; blue, 0 }  ][line width=0.75]    (10.93,-3.29) .. controls (6.95,-1.4) and (3.31,-0.3) .. (0,0) .. controls (3.31,0.3) and (6.95,1.4) .. (10.93,3.29)   ;
\draw    (299,289.43) -- (278.51,271.74) ;
\draw [shift={(277,270.43)}, rotate = 40.82] [color={rgb, 255:red, 0; green, 0; blue, 0 }  ][line width=0.75]    (10.93,-3.29) .. controls (6.95,-1.4) and (3.31,-0.3) .. (0,0) .. controls (3.31,0.3) and (6.95,1.4) .. (10.93,3.29)   ;
\draw    (256,272.43) -- (231.27,291.2) ;
\draw [shift={(229.68,292.41)}, rotate = 322.81] [color={rgb, 255:red, 0; green, 0; blue, 0 }  ][line width=0.75]    (10.93,-3.29) .. controls (6.95,-1.4) and (3.31,-0.3) .. (0,0) .. controls (3.31,0.3) and (6.95,1.4) .. (10.93,3.29)   ;
\draw    (280.72,411.36) -- (273.88,384.48) ;
\draw [shift={(273.39,382.54)}, rotate = 75.73] [color={rgb, 255:red, 0; green, 0; blue, 0 }  ][line width=0.75]    (10.93,-3.29) .. controls (6.95,-1.4) and (3.31,-0.3) .. (0,0) .. controls (3.31,0.3) and (6.95,1.4) .. (10.93,3.29)   ;
\draw    (191.04,407.97) -- (139,409.38) ;
\draw [shift={(137,409.43)}, rotate = 358.45] [color={rgb, 255:red, 0; green, 0; blue, 0 }  ][line width=0.75]    (10.93,-3.29) .. controls (6.95,-1.4) and (3.31,-0.3) .. (0,0) .. controls (3.31,0.3) and (6.95,1.4) .. (10.93,3.29)   ;
\draw    (164.53,452.62) -- (132.38,418.88) ;
\draw [shift={(131,417.43)}, rotate = 46.38] [color={rgb, 255:red, 0; green, 0; blue, 0 }  ][line width=0.75]    (10.93,-3.29) .. controls (6.95,-1.4) and (3.31,-0.3) .. (0,0) .. controls (3.31,0.3) and (6.95,1.4) .. (10.93,3.29)   ;
\draw    (237.01,451.49) -- (273.5,423.32) ;
\draw [shift={(275.08,422.1)}, rotate = 142.34] [color={rgb, 255:red, 0; green, 0; blue, 0 }  ][line width=0.75]    (10.93,-3.29) .. controls (6.95,-1.4) and (3.31,-0.3) .. (0,0) .. controls (3.31,0.3) and (6.95,1.4) .. (10.93,3.29)   ;
\draw    (173,351.43) -- (195.16,399.62) ;
\draw [shift={(196,401.43)}, rotate = 245.3] [color={rgb, 255:red, 0; green, 0; blue, 0 }  ][line width=0.75]    (10.93,-3.29) .. controls (6.95,-1.4) and (3.31,-0.3) .. (0,0) .. controls (3.31,0.3) and (6.95,1.4) .. (10.93,3.29)   ;
\draw    (296.05,330.55) .. controls (314.75,332.2) and (321.65,316.67) .. (297.94,321.65) ;
\draw [shift={(296.05,322.07)}, rotate = 346.53] [color={rgb, 255:red, 0; green, 0; blue, 0 }  ][line width=0.75]    (10.93,-3.29) .. controls (6.95,-1.4) and (3.31,-0.3) .. (0,0) .. controls (3.31,0.3) and (6.95,1.4) .. (10.93,3.29)   ;
\draw    (226.01,304.56) .. controls (229.88,317.29) and (244.58,328.4) .. (270.65,328.31) ;
\draw [shift={(272.26,328.29)}, rotate = 178.8] [color={rgb, 255:red, 0; green, 0; blue, 0 }  ][line width=0.75]    (10.93,-3.29) .. controls (6.95,-1.4) and (3.31,-0.3) .. (0,0) .. controls (3.31,0.3) and (6.95,1.4) .. (10.93,3.29)   ;
\draw    (122,362) -- (157.63,346.06) ;
\draw [shift={(159.45,345.24)}, rotate = 155.9] [color={rgb, 255:red, 0; green, 0; blue, 0 }  ][line width=0.75]    (10.93,-3.29) .. controls (6.95,-1.4) and (3.31,-0.3) .. (0,0) .. controls (3.31,0.3) and (6.95,1.4) .. (10.93,3.29)   ;
\draw [line width=0.75]    (137,273) -- (115.51,355.07) ;
\draw [shift={(115,357)}, rotate = 284.68] [color={rgb, 255:red, 0; green, 0; blue, 0 }  ][line width=0.75]    (10.93,-3.29) .. controls (6.95,-1.4) and (3.31,-0.3) .. (0,0) .. controls (3.31,0.3) and (6.95,1.4) .. (10.93,3.29)   ;
\draw [line width=0.75]    (295,412) -- (455.67,231.49) ;
\draw [shift={(457,230)}, rotate = 131.67] [color={rgb, 255:red, 0; green, 0; blue, 0 }  ][line width=0.75]    (10.93,-3.29) .. controls (6.95,-1.4) and (3.31,-0.3) .. (0,0) .. controls (3.31,0.3) and (6.95,1.4) .. (10.93,3.29)   ;
\draw [line width=0.75]    (40,88) .. controls (12.14,173.57) and (12,288.84) .. (101.64,363.87) ;
\draw [shift={(103,365)}, rotate = 219.49] [color={rgb, 255:red, 0; green, 0; blue, 0 }  ][line width=0.75]    (10.93,-3.29) .. controls (6.95,-1.4) and (3.31,-0.3) .. (0,0) .. controls (3.31,0.3) and (6.95,1.4) .. (10.93,3.29)   ;
\draw [line width=0.75]    (83,314) .. controls (61.22,297.17) and (48.26,220.55) .. (83.9,180.21) ;
\draw [shift={(85,179)}, rotate = 132.77] [color={rgb, 255:red, 0; green, 0; blue, 0 }  ][line width=0.75]    (10.93,-3.29) .. controls (6.95,-1.4) and (3.31,-0.3) .. (0,0) .. controls (3.31,0.3) and (6.95,1.4) .. (10.93,3.29)   ;
\draw [line width=0.75]    (116,281) .. controls (240.38,351.65) and (262.78,195.57) .. (334.91,138.1) ;
\draw [shift={(336,137.24)}, rotate = 142.13] [color={rgb, 255:red, 0; green, 0; blue, 0 }  ][line width=0.75]    (10.93,-3.29) .. controls (6.95,-1.4) and (3.31,-0.3) .. (0,0) .. controls (3.31,0.3) and (6.95,1.4) .. (10.93,3.29)   ;
\draw [line width=0.75]    (284,58) -- (482.02,29.29) ;
\draw [shift={(484,29)}, rotate = 171.75] [color={rgb, 255:red, 0; green, 0; blue, 0 }  ][line width=0.75]    (10.93,-3.29) .. controls (6.95,-1.4) and (3.31,-0.3) .. (0,0) .. controls (3.31,0.3) and (6.95,1.4) .. (10.93,3.29)   ;
\draw [line width=0.75]    (312,93) .. controls (261.04,240.98) and (236.97,227.61) .. (206.85,217.61) ;
\draw [shift={(205,217)}, rotate = 17.88] [color={rgb, 255:red, 0; green, 0; blue, 0 }  ][line width=0.75]    (10.93,-3.29) .. controls (6.95,-1.4) and (3.31,-0.3) .. (0,0) .. controls (3.31,0.3) and (6.95,1.4) .. (10.93,3.29)   ;
\draw [line width=0.75]    (255,172) .. controls (266.44,107.42) and (291,56) .. (435,59) ;
\draw [shift={(435,59)}, rotate = 181.19] [color={rgb, 255:red, 0; green, 0; blue, 0 }  ][line width=0.75]    (10.93,-3.29) .. controls (6.95,-1.4) and (3.31,-0.3) .. (0,0) .. controls (3.31,0.3) and (6.95,1.4) .. (10.93,3.29)   ;
\draw [line width=0.75]    (552,62) .. controls (572,244) and (470,390) .. (245,461) ;
\draw [shift={(245,461)}, rotate = 342.49] [color={rgb, 255:red, 0; green, 0; blue, 0 }  ][line width=0.75]    (10.93,-3.29) .. controls (6.95,-1.4) and (3.31,-0.3) .. (0,0) .. controls (3.31,0.3) and (6.95,1.4) .. (10.93,3.29)   ;
\draw [line width=0.75]    (320,292) .. controls (368.76,295.98) and (391.77,300.95) .. (403.82,231.06) ;
\draw [shift={(404,230)}, rotate = 99.59] [color={rgb, 255:red, 0; green, 0; blue, 0 }  ][line width=0.75]    (10.93,-3.29) .. controls (6.95,-1.4) and (3.31,-0.3) .. (0,0) .. controls (3.31,0.3) and (6.95,1.4) .. (10.93,3.29)   ;
\draw [line width=1.5]    (197,207) -- (225.64,150.67) ;
\draw [shift={(227,148)}, rotate = 116.95] [color={rgb, 255:red, 0; green, 0; blue, 0 }  ][line width=1.5]    (14.21,-4.28) .. controls (9.04,-1.82) and (4.3,-0.39) .. (0,0) .. controls (4.3,0.39) and (9.04,1.82) .. (14.21,4.28)   ;
\draw [line width=1.5]    (211,408) -- (262.24,386.18) ;
\draw [shift={(265,385)}, rotate = 156.93] [color={rgb, 255:red, 0; green, 0; blue, 0 }  ][line width=1.5]    (14.21,-4.28) .. controls (9.04,-1.82) and (4.3,-0.39) .. (0,0) .. controls (4.3,0.39) and (9.04,1.82) .. (14.21,4.28)   ;
\draw [line width=1.5]    (167,351) .. controls (164.11,367.41) and (155.62,384.74) .. (137.96,401.21) ;
\draw [shift={(136,403)}, rotate = 318.18] [color={rgb, 255:red, 0; green, 0; blue, 0 }  ][line width=1.5]    (14.21,-4.28) .. controls (9.04,-1.82) and (4.3,-0.39) .. (0,0) .. controls (4.3,0.39) and (9.04,1.82) .. (14.21,4.28)   ;
\draw [line width=1.5]    (502,180) -- (439.99,175.23) ;
\draw [shift={(437,175)}, rotate = 4.4] [color={rgb, 255:red, 0; green, 0; blue, 0 }  ][line width=1.5]    (14.21,-4.28) .. controls (9.04,-1.82) and (4.3,-0.39) .. (0,0) .. controls (4.3,0.39) and (9.04,1.82) .. (14.21,4.28)   ;

\draw (252.35,33.56) node [anchor=north west][inner sep=0.75pt]  [font=\scriptsize] [align=left] {a};
\draw (200.58,30.95) node [anchor=north west][inner sep=0.75pt]  [font=\scriptsize] [align=left] {a};
\draw (146.13,68.82) node [anchor=north west][inner sep=0.75pt]  [font=\scriptsize] [align=left] {a};
\draw (171.08,124.46) node [anchor=north west][inner sep=0.75pt]  [font=\scriptsize] [align=left] {a};
\draw (212.99,59.77) node [anchor=north west][inner sep=0.75pt]  [font=\scriptsize] [align=left] {a};
\draw (238.93,117.03) node [anchor=north west][inner sep=0.75pt]  [font=\scriptsize] [align=left] {a,b};
\draw (242.76,155.89) node [anchor=north west][inner sep=0.75pt]  [font=\scriptsize] [align=left] {a};
\draw (47.42,101.46) node [anchor=north west][inner sep=0.75pt]  [font=\scriptsize] [align=left] {a};
\draw (95.98,140.19) node [anchor=north west][inner sep=0.75pt]  [font=\scriptsize] [align=left] {a};
\draw (47.19,149.32) node [anchor=north west][inner sep=0.75pt]  [font=\scriptsize] [align=left] {a,b};
\draw (179.54,186.71) node [anchor=north west][inner sep=0.75pt]  [font=\scriptsize] [align=left] {a};
\draw (137.16,181.14) node [anchor=north west][inner sep=0.75pt]  [font=\scriptsize] [align=left] {a};
\draw (156.03,217.18) node [anchor=north west][inner sep=0.75pt]  [font=\scriptsize] [align=left] {a};
\draw (139.83,237.2) node [anchor=north west][inner sep=0.75pt]  [font=\scriptsize] [align=left] {a};
\draw (82.75,236.38) node [anchor=north west][inner sep=0.75pt]  [font=\scriptsize] [align=left] {a};
\draw (119.36,155.71) node [anchor=north west][inner sep=0.75pt]  [font=\scriptsize] [align=left] {b};
\draw (98.11,197.01) node [anchor=north west][inner sep=0.75pt]  [font=\scriptsize] [align=left] {b};
\draw (219.04,201.79) node [anchor=north west][inner sep=0.75pt]  [font=\scriptsize] [align=left] {b};
\draw (132.03,130.89) node [anchor=north west][inner sep=0.75pt]  [font=\scriptsize] [align=left] {b};
\draw (276.83,81.12) node [anchor=north west][inner sep=0.75pt]  [font=\scriptsize] [align=left] {b};
\draw (190.82,80.94) node [anchor=north west][inner sep=0.75pt]  [font=\scriptsize] [align=left] {b};
\draw (89.29,100.02) node [anchor=north west][inner sep=0.75pt]  [font=\scriptsize] [align=left] {b};
\draw (44.17,80.38) node [anchor=north west][inner sep=0.75pt]  [font=\scriptsize] [align=left] {1a};
\draw (69.43,123.81) node [anchor=north west][inner sep=0.75pt]  [font=\scriptsize] [align=left] {2a};
\draw (81.26,165.45) node [anchor=north west][inner sep=0.75pt]  [font=\scriptsize] [align=left] {3a};
\draw (98.09,274.88) node [anchor=north west][inner sep=0.75pt]  [font=\scriptsize] [align=left] {4a};
\draw (135.06,257) node [anchor=north west][inner sep=0.75pt]  [font=\scriptsize] [align=left] {5a};
\draw (128.39,210.96) node [anchor=north west][inner sep=0.75pt]  [font=\scriptsize] [align=left] {6a};
\draw (154.81,165.32) node [anchor=north west][inner sep=0.75pt]  [font=\scriptsize] [align=left] {7a};
\draw (189.57,213.66) node [anchor=north west][inner sep=0.75pt]  [font=\scriptsize] [align=left] {9a};
\draw (122.57,98.46) node [anchor=north west][inner sep=0.75pt]  [font=\scriptsize] [align=left] {8a};
\draw (222.35,132.98) node [anchor=north west][inner sep=0.75pt]  [font=\scriptsize] [align=left] {10a};
\draw (233.5,173.36) node [anchor=north west][inner sep=0.75pt]  [font=\scriptsize] [align=left] {12a};
\draw (172.19,51.99) node [anchor=north west][inner sep=0.75pt]  [font=\scriptsize] [align=left] {11a};
\draw (235.13,82.64) node [anchor=north west][inner sep=0.75pt]  [font=\scriptsize] [align=left] {13a};
\draw (215.84,21.87) node [anchor=north west][inner sep=0.75pt]  [font=\scriptsize] [align=left] {14a};
\draw (260.08,51.47) node [anchor=north west][inner sep=0.75pt]  [font=\scriptsize] [align=left] {15a};
\draw (520.35,38.56) node [anchor=north west][inner sep=0.75pt]  [font=\scriptsize] [align=left] {a};
\draw (468.58,34.95) node [anchor=north west][inner sep=0.75pt]  [font=\scriptsize] [align=left] {a};
\draw (414.13,72.82) node [anchor=north west][inner sep=0.75pt]  [font=\scriptsize] [align=left] {a};
\draw (439.08,128.46) node [anchor=north west][inner sep=0.75pt]  [font=\scriptsize] [align=left] {a};
\draw (480.99,63.77) node [anchor=north west][inner sep=0.75pt]  [font=\scriptsize] [align=left] {a};
\draw (506.93,121.03) node [anchor=north west][inner sep=0.75pt]  [font=\scriptsize] [align=left] {a,b};
\draw (510.76,159.89) node [anchor=north west][inner sep=0.75pt]  [font=\scriptsize] [align=left] {a};
\draw (316.42,109.46) node [anchor=north west][inner sep=0.75pt]  [font=\scriptsize] [align=left] {a};
\draw (363.98,144.19) node [anchor=north west][inner sep=0.75pt]  [font=\scriptsize] [align=left] {a};
\draw (320.19,157.32) node [anchor=north west][inner sep=0.75pt]  [font=\scriptsize] [align=left] {a,b};
\draw (447.54,190.71) node [anchor=north west][inner sep=0.75pt]  [font=\scriptsize] [align=left] {a};
\draw (405.16,185.14) node [anchor=north west][inner sep=0.75pt]  [font=\scriptsize] [align=left] {a};
\draw (424.03,221.18) node [anchor=north west][inner sep=0.75pt]  [font=\scriptsize] [align=left] {a};
\draw (385.66,240.12) node [anchor=north west][inner sep=0.75pt]  [font=\scriptsize] [align=left] {a};
\draw (331.75,189.38) node [anchor=north west][inner sep=0.75pt]  [font=\scriptsize] [align=left] {a};
\draw (387.36,159.71) node [anchor=north west][inner sep=0.75pt]  [font=\scriptsize] [align=left] {b};
\draw (366.11,201.01) node [anchor=north west][inner sep=0.75pt]  [font=\scriptsize] [align=left] {b};
\draw (487.04,205.79) node [anchor=north west][inner sep=0.75pt]  [font=\scriptsize] [align=left] {b};
\draw (400.03,134.89) node [anchor=north west][inner sep=0.75pt]  [font=\scriptsize] [align=left] {b};
\draw (544.83,85.12) node [anchor=north west][inner sep=0.75pt]  [font=\scriptsize] [align=left] {b};
\draw (458.82,84.94) node [anchor=north west][inner sep=0.75pt]  [font=\scriptsize] [align=left] {b};
\draw (357.29,104.02) node [anchor=north west][inner sep=0.75pt]  [font=\scriptsize] [align=left] {b};
\draw (312.17,84.38) node [anchor=north west][inner sep=0.75pt]  [font=\scriptsize] [align=left] {1c};
\draw (337.43,127.81) node [anchor=north west][inner sep=0.75pt]  [font=\scriptsize] [align=left] {2c};
\draw (349.26,169.45) node [anchor=north west][inner sep=0.75pt]  [font=\scriptsize] [align=left] {3c};
\draw (330.09,207.88) node [anchor=north west][inner sep=0.75pt]  [font=\scriptsize] [align=left] {4c};
\draw (364.06,248) node [anchor=north west][inner sep=0.75pt]  [font=\scriptsize] [align=left] {5c};
\draw (396.39,214.96) node [anchor=north west][inner sep=0.75pt]  [font=\scriptsize] [align=left] {6c};
\draw (421,168.43) node [anchor=north west][inner sep=0.75pt]  [font=\scriptsize] [align=left] {7c};
\draw (455,217.43) node [anchor=north west][inner sep=0.75pt]  [font=\scriptsize] [align=left] {9c};
\draw (390.57,102.46) node [anchor=north west][inner sep=0.75pt]  [font=\scriptsize] [align=left] {8c};
\draw (490.35,136.98) node [anchor=north west][inner sep=0.75pt]  [font=\scriptsize] [align=left] {10c};
\draw (501.5,177.36) node [anchor=north west][inner sep=0.75pt]  [font=\scriptsize] [align=left] {12c};
\draw (440.19,55.99) node [anchor=north west][inner sep=0.75pt]  [font=\scriptsize] [align=left] {11c};
\draw (503.13,86.64) node [anchor=north west][inner sep=0.75pt]  [font=\scriptsize] [align=left] {13c};
\draw (483.84,25.87) node [anchor=north west][inner sep=0.75pt]  [font=\scriptsize] [align=left] {14c};
\draw (528.08,55.47) node [anchor=north west][inner sep=0.75pt]  [font=\scriptsize] [align=left] {15c};
\draw (290.35,270.56) node [anchor=north west][inner sep=0.75pt]  [font=\scriptsize] [align=left] {a};
\draw (239.58,268.95) node [anchor=north west][inner sep=0.75pt]  [font=\scriptsize] [align=left] {a};
\draw (185.13,306.82) node [anchor=north west][inner sep=0.75pt]  [font=\scriptsize] [align=left] {a};
\draw (210.08,362.46) node [anchor=north west][inner sep=0.75pt]  [font=\scriptsize] [align=left] {a};
\draw (251.99,297.77) node [anchor=north west][inner sep=0.75pt]  [font=\scriptsize] [align=left] {a};
\draw (277.93,355.03) node [anchor=north west][inner sep=0.75pt]  [font=\scriptsize] [align=left] {a,b};
\draw (281.76,393.89) node [anchor=north west][inner sep=0.75pt]  [font=\scriptsize] [align=left] {a};
\draw (86.42,339.46) node [anchor=north west][inner sep=0.75pt]  [font=\scriptsize] [align=left] {a};
\draw (131.98,378.19) node [anchor=north west][inner sep=0.75pt]  [font=\scriptsize] [align=left] {a};
\draw (86.19,387.32) node [anchor=north west][inner sep=0.75pt]  [font=\scriptsize] [align=left] {a,b};
\draw (218.54,424.71) node [anchor=north west][inner sep=0.75pt]  [font=\scriptsize] [align=left] {a};
\draw (176.16,419.14) node [anchor=north west][inner sep=0.75pt]  [font=\scriptsize] [align=left] {a};
\draw (195.03,455.18) node [anchor=north west][inner sep=0.75pt]  [font=\scriptsize] [align=left] {a};
\draw (156.66,474.12) node [anchor=north west][inner sep=0.75pt]  [font=\scriptsize] [align=left] {a};
\draw (102.75,423.38) node [anchor=north west][inner sep=0.75pt]  [font=\scriptsize] [align=left] {a};
\draw (166.02,411.7) node [anchor=north west][inner sep=0.75pt]  [font=\scriptsize] [align=left] {b};
\draw (137.11,435.01) node [anchor=north west][inner sep=0.75pt]  [font=\scriptsize] [align=left] {b};
\draw (258.04,439.79) node [anchor=north west][inner sep=0.75pt]  [font=\scriptsize] [align=left] {b};
\draw (171.03,368.89) node [anchor=north west][inner sep=0.75pt]  [font=\scriptsize] [align=left] {b};
\draw (315.83,319.12) node [anchor=north west][inner sep=0.75pt]  [font=\scriptsize] [align=left] {b};
\draw (229.82,318.94) node [anchor=north west][inner sep=0.75pt]  [font=\scriptsize] [align=left] {b};
\draw (128.29,338.02) node [anchor=north west][inner sep=0.75pt]  [font=\scriptsize] [align=left] {b};
\draw (83.17,318.38) node [anchor=north west][inner sep=0.75pt]  [font=\scriptsize] [align=left] {1b};
\draw (108.43,361.81) node [anchor=north west][inner sep=0.75pt]  [font=\scriptsize] [align=left] {2b};
\draw (120.26,403.45) node [anchor=north west][inner sep=0.75pt]  [font=\scriptsize] [align=left] {3b};
\draw (101.09,441.88) node [anchor=north west][inner sep=0.75pt]  [font=\scriptsize] [align=left] {4b};
\draw (135.06,482) node [anchor=north west][inner sep=0.75pt]  [font=\scriptsize] [align=left] {5b};
\draw (167.39,448.96) node [anchor=north west][inner sep=0.75pt]  [font=\scriptsize] [align=left] {6b};
\draw (193.81,403.32) node [anchor=north west][inner sep=0.75pt]  [font=\scriptsize] [align=left] {7b};
\draw (228.57,451.66) node [anchor=north west][inner sep=0.75pt]  [font=\scriptsize] [align=left] {9b};
\draw (161.57,336.46) node [anchor=north west][inner sep=0.75pt]  [font=\scriptsize] [align=left] {8b};
\draw (261.35,370.98) node [anchor=north west][inner sep=0.75pt]  [font=\scriptsize] [align=left] {10b};
\draw (272.5,411.36) node [anchor=north west][inner sep=0.75pt]  [font=\scriptsize] [align=left] {12b};
\draw (211.19,289.99) node [anchor=north west][inner sep=0.75pt]  [font=\scriptsize] [align=left] {11b};
\draw (274.13,320.64) node [anchor=north west][inner sep=0.75pt]  [font=\scriptsize] [align=left] {13b};
\draw (254.84,259.87) node [anchor=north west][inner sep=0.75pt]  [font=\scriptsize] [align=left] {14b};
\draw (299.08,289.47) node [anchor=north west][inner sep=0.75pt]  [font=\scriptsize] [align=left] {15b};
\draw (116,302) node [anchor=north west][inner sep=0.75pt]  [font=\scriptsize] [align=left] {0};
\draw (378,324) node [anchor=north west][inner sep=0.75pt]  [font=\scriptsize] [align=left] {1};
\draw (17,242) node [anchor=north west][inner sep=0.75pt]  [font=\scriptsize] [align=left] {b};
\draw (49,240) node [anchor=north west][inner sep=0.75pt]  [font=\scriptsize] [align=left] {b};
\draw (273,223) node [anchor=north west][inner sep=0.75pt]  [font=\scriptsize] [align=left] {b};
\draw (369,29) node [anchor=north west][inner sep=0.75pt]  [font=\scriptsize] [align=left] {b};
\draw (276,182) node [anchor=north west][inner sep=0.75pt]  [font=\scriptsize] [align=left] {b};
\draw (301,66) node [anchor=north west][inner sep=0.75pt]  [font=\scriptsize] [align=left] {b};
\draw (499,301) node [anchor=north west][inner sep=0.75pt]  [font=\scriptsize] [align=left] {b};
\draw (378,289) node [anchor=north west][inner sep=0.75pt]  [font=\scriptsize] [align=left] {b};
\draw (198,170) node [anchor=north west][inner sep=0.75pt]  [font=\scriptsize] [align=left] {0};
\draw (225,387) node [anchor=north west][inner sep=0.75pt]  [font=\scriptsize] [align=left] {1};
\draw (155,380) node [anchor=north west][inner sep=0.75pt]  [font=\scriptsize] [align=left] {0};
\draw (468,165) node [anchor=north west][inner sep=0.75pt]  [font=\scriptsize] [align=left] {0};

\end{tikzpicture}

%% file: extended_step.tex
\tikzset{every picture/.style={line width=0.75pt}} 

\begin{tikzpicture}[x=0.75pt,y=0.75pt,yscale=-0.9,xscale=0.9]

\draw    (53,122.43) .. controls (61.6,128.16) and (65.53,141.33) .. (59.48,153.82) ;
\draw [shift={(58.56,155.58)}, rotate = 299.43] [color={rgb, 255:red, 0; green, 0; blue, 0 }  ][line width=0.75]    (10.93,-3.29) .. controls (6.95,-1.4) and (3.31,-0.3) .. (0,0) .. controls (3.31,0.3) and (6.95,1.4) .. (10.93,3.29)   ;
\draw    (48.12,158.41) .. controls (38.55,156.21) and (31.44,144.19) .. (35.58,124.12) ;
\draw [shift={(36,122.24)}, rotate = 103.27] [color={rgb, 255:red, 0; green, 0; blue, 0 }  ][line width=0.75]    (10.93,-3.29) .. controls (6.95,-1.4) and (3.31,-0.3) .. (0,0) .. controls (3.31,0.3) and (6.95,1.4) .. (10.93,3.29)   ;
\draw    (113,206.43) -- (152,207.38) ;
\draw [shift={(154,207.43)}, rotate = 181.4] [color={rgb, 255:red, 0; green, 0; blue, 0 }  ][line width=0.75]    (10.93,-3.29) .. controls (6.95,-1.4) and (3.31,-0.3) .. (0,0) .. controls (3.31,0.3) and (6.95,1.4) .. (10.93,3.29)   ;
\draw    (125.68,165.75) -- (104.28,195.77) ;
\draw [shift={(103.12,197.4)}, rotate = 305.49] [color={rgb, 255:red, 0; green, 0; blue, 0 }  ][line width=0.75]    (10.93,-3.29) .. controls (6.95,-1.4) and (3.31,-0.3) .. (0,0) .. controls (3.31,0.3) and (6.95,1.4) .. (10.93,3.29)   ;
\draw    (158,201.43) -- (132.2,167.03) ;
\draw [shift={(131,165.43)}, rotate = 53.13] [color={rgb, 255:red, 0; green, 0; blue, 0 }  ][line width=0.75]    (10.93,-3.29) .. controls (6.95,-1.4) and (3.31,-0.3) .. (0,0) .. controls (3.31,0.3) and (6.95,1.4) .. (10.93,3.29)   ;
\draw    (200.7,120.54) -- (208.69,87.45) ;
\draw [shift={(209.16,85.51)}, rotate = 103.58] [color={rgb, 255:red, 0; green, 0; blue, 0 }  ][line width=0.75]    (10.93,-3.29) .. controls (6.95,-1.4) and (3.31,-0.3) .. (0,0) .. controls (3.31,0.3) and (6.95,1.4) .. (10.93,3.29)   ;
\draw    (200.13,73.64) -- (162.38,51.48) ;
\draw [shift={(160.65,50.47)}, rotate = 30.41] [color={rgb, 255:red, 0; green, 0; blue, 0 }  ][line width=0.75]    (10.93,-3.29) .. controls (6.95,-1.4) and (3.31,-0.3) .. (0,0) .. controls (3.31,0.3) and (6.95,1.4) .. (10.93,3.29)   ;
\draw    (145.99,55.56) -- (107.54,87.16) ;
\draw [shift={(106,88.43)}, rotate = 320.57] [color={rgb, 255:red, 0; green, 0; blue, 0 }  ][line width=0.75]    (10.93,-3.29) .. controls (6.95,-1.4) and (3.31,-0.3) .. (0,0) .. controls (3.31,0.3) and (6.95,1.4) .. (10.93,3.29)   ;
\draw    (108,97.43) -- (188.12,129.41) ;
\draw [shift={(189.98,130.15)}, rotate = 201.76] [color={rgb, 255:red, 0; green, 0; blue, 0 }  ][line width=0.75]    (10.93,-3.29) .. controls (6.95,-1.4) and (3.31,-0.3) .. (0,0) .. controls (3.31,0.3) and (6.95,1.4) .. (10.93,3.29)   ;
\draw    (22,83.43) -- (37.06,111.67) ;
\draw [shift={(38,113.43)}, rotate = 241.93] [color={rgb, 255:red, 0; green, 0; blue, 0 }  ][line width=0.75]    (10.93,-3.29) .. controls (6.95,-1.4) and (3.31,-0.3) .. (0,0) .. controls (3.31,0.3) and (6.95,1.4) .. (10.93,3.29)   ;
\draw    (37.69,192.88) -- (49.18,170.64) ;
\draw [shift={(50.1,168.86)}, rotate = 117.32] [color={rgb, 255:red, 0; green, 0; blue, 0 }  ][line width=0.75]    (10.93,-3.29) .. controls (6.95,-1.4) and (3.31,-0.3) .. (0,0) .. controls (3.31,0.3) and (6.95,1.4) .. (10.93,3.29)   ;
\draw    (72.66,233.85) -- (93.29,211.86) ;
\draw [shift={(94.66,210.4)}, rotate = 133.17] [color={rgb, 255:red, 0; green, 0; blue, 0 }  ][line width=0.75]    (10.93,-3.29) .. controls (6.95,-1.4) and (3.31,-0.3) .. (0,0) .. controls (3.31,0.3) and (6.95,1.4) .. (10.93,3.29)   ;
\draw    (228,40.43) -- (207.51,22.74) ;
\draw [shift={(206,21.43)}, rotate = 40.82] [color={rgb, 255:red, 0; green, 0; blue, 0 }  ][line width=0.75]    (10.93,-3.29) .. controls (6.95,-1.4) and (3.31,-0.3) .. (0,0) .. controls (3.31,0.3) and (6.95,1.4) .. (10.93,3.29)   ;
\draw    (185,23.43) -- (160.27,42.2) ;
\draw [shift={(158.68,43.41)}, rotate = 322.81] [color={rgb, 255:red, 0; green, 0; blue, 0 }  ][line width=0.75]    (10.93,-3.29) .. controls (6.95,-1.4) and (3.31,-0.3) .. (0,0) .. controls (3.31,0.3) and (6.95,1.4) .. (10.93,3.29)   ;
\draw    (209.72,162.36) -- (202.88,135.48) ;
\draw [shift={(202.39,133.54)}, rotate = 75.73] [color={rgb, 255:red, 0; green, 0; blue, 0 }  ][line width=0.75]    (10.93,-3.29) .. controls (6.95,-1.4) and (3.31,-0.3) .. (0,0) .. controls (3.31,0.3) and (6.95,1.4) .. (10.93,3.29)   ;
\draw    (120.04,158.97) -- (68,160.38) ;
\draw [shift={(66,160.43)}, rotate = 358.45] [color={rgb, 255:red, 0; green, 0; blue, 0 }  ][line width=0.75]    (10.93,-3.29) .. controls (6.95,-1.4) and (3.31,-0.3) .. (0,0) .. controls (3.31,0.3) and (6.95,1.4) .. (10.93,3.29)   ;
\draw    (93.53,203.62) -- (61.38,169.88) ;
\draw [shift={(60,168.43)}, rotate = 46.38] [color={rgb, 255:red, 0; green, 0; blue, 0 }  ][line width=0.75]    (10.93,-3.29) .. controls (6.95,-1.4) and (3.31,-0.3) .. (0,0) .. controls (3.31,0.3) and (6.95,1.4) .. (10.93,3.29)   ;
\draw    (166.01,202.49) -- (202.5,174.32) ;
\draw [shift={(204.08,173.1)}, rotate = 142.34] [color={rgb, 255:red, 0; green, 0; blue, 0 }  ][line width=0.75]    (10.93,-3.29) .. controls (6.95,-1.4) and (3.31,-0.3) .. (0,0) .. controls (3.31,0.3) and (6.95,1.4) .. (10.93,3.29)   ;
\draw    (102,102.43) -- (124.16,150.62) ;
\draw [shift={(125,152.43)}, rotate = 245.3] [color={rgb, 255:red, 0; green, 0; blue, 0 }  ][line width=0.75]    (10.93,-3.29) .. controls (6.95,-1.4) and (3.31,-0.3) .. (0,0) .. controls (3.31,0.3) and (6.95,1.4) .. (10.93,3.29)   ;
\draw    (225.05,81.55) .. controls (243.75,83.2) and (250.65,67.67) .. (226.94,72.65) ;
\draw [shift={(225.05,73.07)}, rotate = 346.53] [color={rgb, 255:red, 0; green, 0; blue, 0 }  ][line width=0.75]    (10.93,-3.29) .. controls (6.95,-1.4) and (3.31,-0.3) .. (0,0) .. controls (3.31,0.3) and (6.95,1.4) .. (10.93,3.29)   ;
\draw    (155.01,55.56) .. controls (158.88,68.29) and (173.58,79.4) .. (199.65,79.31) ;
\draw [shift={(201.26,79.29)}, rotate = 178.8] [color={rgb, 255:red, 0; green, 0; blue, 0 }  ][line width=0.75]    (10.93,-3.29) .. controls (6.95,-1.4) and (3.31,-0.3) .. (0,0) .. controls (3.31,0.3) and (6.95,1.4) .. (10.93,3.29)   ;
\draw    (43.89,112.07) -- (86.57,96.91) ;
\draw [shift={(88.45,96.24)}, rotate = 160.45] [color={rgb, 255:red, 0; green, 0; blue, 0 }  ][line width=0.75]    (10.93,-3.29) .. controls (6.95,-1.4) and (3.31,-0.3) .. (0,0) .. controls (3.31,0.3) and (6.95,1.4) .. (10.93,3.29)   ;
\draw    (321,126.43) .. controls (329.6,132.16) and (333.53,145.33) .. (327.48,157.82) ;
\draw [shift={(326.56,159.58)}, rotate = 299.43] [color={rgb, 255:red, 0; green, 0; blue, 0 }  ][line width=0.75]    (10.93,-3.29) .. controls (6.95,-1.4) and (3.31,-0.3) .. (0,0) .. controls (3.31,0.3) and (6.95,1.4) .. (10.93,3.29)   ;
\draw    (316.12,162.41) .. controls (306.55,160.21) and (299.44,148.19) .. (303.58,128.12) ;
\draw [shift={(304,126.24)}, rotate = 103.27] [color={rgb, 255:red, 0; green, 0; blue, 0 }  ][line width=0.75]    (10.93,-3.29) .. controls (6.95,-1.4) and (3.31,-0.3) .. (0,0) .. controls (3.31,0.3) and (6.95,1.4) .. (10.93,3.29)   ;
\draw    (381,210.43) -- (420,211.38) ;
\draw [shift={(422,211.43)}, rotate = 181.4] [color={rgb, 255:red, 0; green, 0; blue, 0 }  ][line width=0.75]    (10.93,-3.29) .. controls (6.95,-1.4) and (3.31,-0.3) .. (0,0) .. controls (3.31,0.3) and (6.95,1.4) .. (10.93,3.29)   ;
\draw    (393.68,169.75) -- (372.28,199.77) ;
\draw [shift={(371.12,201.4)}, rotate = 305.49] [color={rgb, 255:red, 0; green, 0; blue, 0 }  ][line width=0.75]    (10.93,-3.29) .. controls (6.95,-1.4) and (3.31,-0.3) .. (0,0) .. controls (3.31,0.3) and (6.95,1.4) .. (10.93,3.29)   ;
\draw    (426,205.43) -- (400.2,171.03) ;
\draw [shift={(399,169.43)}, rotate = 53.13] [color={rgb, 255:red, 0; green, 0; blue, 0 }  ][line width=0.75]    (10.93,-3.29) .. controls (6.95,-1.4) and (3.31,-0.3) .. (0,0) .. controls (3.31,0.3) and (6.95,1.4) .. (10.93,3.29)   ;
\draw    (468.7,124.54) -- (476.69,91.45) ;
\draw [shift={(477.16,89.51)}, rotate = 103.58] [color={rgb, 255:red, 0; green, 0; blue, 0 }  ][line width=0.75]    (10.93,-3.29) .. controls (6.95,-1.4) and (3.31,-0.3) .. (0,0) .. controls (3.31,0.3) and (6.95,1.4) .. (10.93,3.29)   ;
\draw    (468.13,77.64) -- (430.38,55.48) ;
\draw [shift={(428.65,54.47)}, rotate = 30.41] [color={rgb, 255:red, 0; green, 0; blue, 0 }  ][line width=0.75]    (10.93,-3.29) .. controls (6.95,-1.4) and (3.31,-0.3) .. (0,0) .. controls (3.31,0.3) and (6.95,1.4) .. (10.93,3.29)   ;
\draw    (413.99,59.56) -- (375.54,91.16) ;
\draw [shift={(374,92.43)}, rotate = 320.57] [color={rgb, 255:red, 0; green, 0; blue, 0 }  ][line width=0.75]    (10.93,-3.29) .. controls (6.95,-1.4) and (3.31,-0.3) .. (0,0) .. controls (3.31,0.3) and (6.95,1.4) .. (10.93,3.29)   ;
\draw    (376,101.43) -- (456.12,133.41) ;
\draw [shift={(457.98,134.15)}, rotate = 201.76] [color={rgb, 255:red, 0; green, 0; blue, 0 }  ][line width=0.75]    (10.93,-3.29) .. controls (6.95,-1.4) and (3.31,-0.3) .. (0,0) .. controls (3.31,0.3) and (6.95,1.4) .. (10.93,3.29)   ;
\draw    (290,87.43) -- (305.06,115.67) ;
\draw [shift={(306,117.43)}, rotate = 241.93] [color={rgb, 255:red, 0; green, 0; blue, 0 }  ][line width=0.75]    (10.93,-3.29) .. controls (6.95,-1.4) and (3.31,-0.3) .. (0,0) .. controls (3.31,0.3) and (6.95,1.4) .. (10.93,3.29)   ;
\draw    (277,197) -- (316.37,173.87) ;
\draw [shift={(318.1,172.86)}, rotate = 149.57] [color={rgb, 255:red, 0; green, 0; blue, 0 }  ][line width=0.75]    (10.93,-3.29) .. controls (6.95,-1.4) and (3.31,-0.3) .. (0,0) .. controls (3.31,0.3) and (6.95,1.4) .. (10.93,3.29)   ;
\draw    (386,258) -- (363.6,216.16) ;
\draw [shift={(362.66,214.4)}, rotate = 61.84] [color={rgb, 255:red, 0; green, 0; blue, 0 }  ][line width=0.75]    (10.93,-3.29) .. controls (6.95,-1.4) and (3.31,-0.3) .. (0,0) .. controls (3.31,0.3) and (6.95,1.4) .. (10.93,3.29)   ;
\draw    (496,44.43) -- (475.51,26.74) ;
\draw [shift={(474,25.43)}, rotate = 40.82] [color={rgb, 255:red, 0; green, 0; blue, 0 }  ][line width=0.75]    (10.93,-3.29) .. controls (6.95,-1.4) and (3.31,-0.3) .. (0,0) .. controls (3.31,0.3) and (6.95,1.4) .. (10.93,3.29)   ;
\draw    (453,27.43) -- (428.27,46.2) ;
\draw [shift={(426.68,47.41)}, rotate = 322.81] [color={rgb, 255:red, 0; green, 0; blue, 0 }  ][line width=0.75]    (10.93,-3.29) .. controls (6.95,-1.4) and (3.31,-0.3) .. (0,0) .. controls (3.31,0.3) and (6.95,1.4) .. (10.93,3.29)   ;
\draw    (477.72,166.36) -- (470.88,139.48) ;
\draw [shift={(470.39,137.54)}, rotate = 75.73] [color={rgb, 255:red, 0; green, 0; blue, 0 }  ][line width=0.75]    (10.93,-3.29) .. controls (6.95,-1.4) and (3.31,-0.3) .. (0,0) .. controls (3.31,0.3) and (6.95,1.4) .. (10.93,3.29)   ;
\draw    (388.04,162.97) -- (336,164.38) ;
\draw [shift={(334,164.43)}, rotate = 358.45] [color={rgb, 255:red, 0; green, 0; blue, 0 }  ][line width=0.75]    (10.93,-3.29) .. controls (6.95,-1.4) and (3.31,-0.3) .. (0,0) .. controls (3.31,0.3) and (6.95,1.4) .. (10.93,3.29)   ;
\draw    (361.53,207.62) -- (329.38,173.88) ;
\draw [shift={(328,172.43)}, rotate = 46.38] [color={rgb, 255:red, 0; green, 0; blue, 0 }  ][line width=0.75]    (10.93,-3.29) .. controls (6.95,-1.4) and (3.31,-0.3) .. (0,0) .. controls (3.31,0.3) and (6.95,1.4) .. (10.93,3.29)   ;
\draw    (434.01,206.49) -- (470.5,178.32) ;
\draw [shift={(472.08,177.1)}, rotate = 142.34] [color={rgb, 255:red, 0; green, 0; blue, 0 }  ][line width=0.75]    (10.93,-3.29) .. controls (6.95,-1.4) and (3.31,-0.3) .. (0,0) .. controls (3.31,0.3) and (6.95,1.4) .. (10.93,3.29)   ;
\draw    (370,106.43) -- (392.16,154.62) ;
\draw [shift={(393,156.43)}, rotate = 245.3] [color={rgb, 255:red, 0; green, 0; blue, 0 }  ][line width=0.75]    (10.93,-3.29) .. controls (6.95,-1.4) and (3.31,-0.3) .. (0,0) .. controls (3.31,0.3) and (6.95,1.4) .. (10.93,3.29)   ;
\draw    (493.05,85.55) .. controls (511.75,87.2) and (518.65,71.67) .. (494.94,76.65) ;
\draw [shift={(493.05,77.07)}, rotate = 346.53] [color={rgb, 255:red, 0; green, 0; blue, 0 }  ][line width=0.75]    (10.93,-3.29) .. controls (6.95,-1.4) and (3.31,-0.3) .. (0,0) .. controls (3.31,0.3) and (6.95,1.4) .. (10.93,3.29)   ;
\draw    (423.01,59.56) .. controls (426.88,72.29) and (441.58,83.4) .. (467.65,83.31) ;
\draw [shift={(469.26,83.29)}, rotate = 178.8] [color={rgb, 255:red, 0; green, 0; blue, 0 }  ][line width=0.75]    (10.93,-3.29) .. controls (6.95,-1.4) and (3.31,-0.3) .. (0,0) .. controls (3.31,0.3) and (6.95,1.4) .. (10.93,3.29)   ;
\draw    (311.89,116.07) -- (354.57,100.91) ;
\draw [shift={(356.45,100.24)}, rotate = 160.45] [color={rgb, 255:red, 0; green, 0; blue, 0 }  ][line width=0.75]    (10.93,-3.29) .. controls (6.95,-1.4) and (3.31,-0.3) .. (0,0) .. controls (3.31,0.3) and (6.95,1.4) .. (10.93,3.29)   ;
\draw    (92,360.43) .. controls (100.6,366.16) and (104.53,379.33) .. (98.48,391.82) ;
\draw [shift={(97.56,393.58)}, rotate = 299.43] [color={rgb, 255:red, 0; green, 0; blue, 0 }  ][line width=0.75]    (10.93,-3.29) .. controls (6.95,-1.4) and (3.31,-0.3) .. (0,0) .. controls (3.31,0.3) and (6.95,1.4) .. (10.93,3.29)   ;
\draw    (87.12,396.41) .. controls (77.55,394.21) and (70.44,382.19) .. (74.58,362.12) ;
\draw [shift={(75,360.24)}, rotate = 103.27] [color={rgb, 255:red, 0; green, 0; blue, 0 }  ][line width=0.75]    (10.93,-3.29) .. controls (6.95,-1.4) and (3.31,-0.3) .. (0,0) .. controls (3.31,0.3) and (6.95,1.4) .. (10.93,3.29)   ;
\draw    (152,444.43) -- (191,445.38) ;
\draw [shift={(193,445.43)}, rotate = 181.4] [color={rgb, 255:red, 0; green, 0; blue, 0 }  ][line width=0.75]    (10.93,-3.29) .. controls (6.95,-1.4) and (3.31,-0.3) .. (0,0) .. controls (3.31,0.3) and (6.95,1.4) .. (10.93,3.29)   ;
\draw    (164.68,403.75) -- (143.28,433.77) ;
\draw [shift={(142.12,435.4)}, rotate = 305.49] [color={rgb, 255:red, 0; green, 0; blue, 0 }  ][line width=0.75]    (10.93,-3.29) .. controls (6.95,-1.4) and (3.31,-0.3) .. (0,0) .. controls (3.31,0.3) and (6.95,1.4) .. (10.93,3.29)   ;
\draw    (197,439.43) -- (171.2,405.03) ;
\draw [shift={(170,403.43)}, rotate = 53.13] [color={rgb, 255:red, 0; green, 0; blue, 0 }  ][line width=0.75]    (10.93,-3.29) .. controls (6.95,-1.4) and (3.31,-0.3) .. (0,0) .. controls (3.31,0.3) and (6.95,1.4) .. (10.93,3.29)   ;
\draw    (239.7,358.54) -- (247.69,325.45) ;
\draw [shift={(248.16,323.51)}, rotate = 103.58] [color={rgb, 255:red, 0; green, 0; blue, 0 }  ][line width=0.75]    (10.93,-3.29) .. controls (6.95,-1.4) and (3.31,-0.3) .. (0,0) .. controls (3.31,0.3) and (6.95,1.4) .. (10.93,3.29)   ;
\draw    (239.13,311.64) -- (201.38,289.48) ;
\draw [shift={(199.65,288.47)}, rotate = 30.41] [color={rgb, 255:red, 0; green, 0; blue, 0 }  ][line width=0.75]    (10.93,-3.29) .. controls (6.95,-1.4) and (3.31,-0.3) .. (0,0) .. controls (3.31,0.3) and (6.95,1.4) .. (10.93,3.29)   ;
\draw    (184.99,293.56) -- (146.54,325.16) ;
\draw [shift={(145,326.43)}, rotate = 320.57] [color={rgb, 255:red, 0; green, 0; blue, 0 }  ][line width=0.75]    (10.93,-3.29) .. controls (6.95,-1.4) and (3.31,-0.3) .. (0,0) .. controls (3.31,0.3) and (6.95,1.4) .. (10.93,3.29)   ;
\draw    (147,335.43) -- (227.12,367.41) ;
\draw [shift={(228.98,368.15)}, rotate = 201.76] [color={rgb, 255:red, 0; green, 0; blue, 0 }  ][line width=0.75]    (10.93,-3.29) .. controls (6.95,-1.4) and (3.31,-0.3) .. (0,0) .. controls (3.31,0.3) and (6.95,1.4) .. (10.93,3.29)   ;
\draw    (61,321.43) -- (76.06,349.67) ;
\draw [shift={(77,351.43)}, rotate = 241.93] [color={rgb, 255:red, 0; green, 0; blue, 0 }  ][line width=0.75]    (10.93,-3.29) .. controls (6.95,-1.4) and (3.31,-0.3) .. (0,0) .. controls (3.31,0.3) and (6.95,1.4) .. (10.93,3.29)   ;
\draw    (76.69,430.88) -- (88.18,408.64) ;
\draw [shift={(89.1,406.86)}, rotate = 117.32] [color={rgb, 255:red, 0; green, 0; blue, 0 }  ][line width=0.75]    (10.93,-3.29) .. controls (6.95,-1.4) and (3.31,-0.3) .. (0,0) .. controls (3.31,0.3) and (6.95,1.4) .. (10.93,3.29)   ;
\draw    (111.66,471.85) -- (132.29,449.86) ;
\draw [shift={(133.66,448.4)}, rotate = 133.17] [color={rgb, 255:red, 0; green, 0; blue, 0 }  ][line width=0.75]    (10.93,-3.29) .. controls (6.95,-1.4) and (3.31,-0.3) .. (0,0) .. controls (3.31,0.3) and (6.95,1.4) .. (10.93,3.29)   ;
\draw    (267,278.43) -- (246.51,260.74) ;
\draw [shift={(245,259.43)}, rotate = 40.82] [color={rgb, 255:red, 0; green, 0; blue, 0 }  ][line width=0.75]    (10.93,-3.29) .. controls (6.95,-1.4) and (3.31,-0.3) .. (0,0) .. controls (3.31,0.3) and (6.95,1.4) .. (10.93,3.29)   ;
\draw    (224,261.43) -- (199.27,280.2) ;
\draw [shift={(197.68,281.41)}, rotate = 322.81] [color={rgb, 255:red, 0; green, 0; blue, 0 }  ][line width=0.75]    (10.93,-3.29) .. controls (6.95,-1.4) and (3.31,-0.3) .. (0,0) .. controls (3.31,0.3) and (6.95,1.4) .. (10.93,3.29)   ;
\draw    (248.72,400.36) -- (241.88,373.48) ;
\draw [shift={(241.39,371.54)}, rotate = 75.73] [color={rgb, 255:red, 0; green, 0; blue, 0 }  ][line width=0.75]    (10.93,-3.29) .. controls (6.95,-1.4) and (3.31,-0.3) .. (0,0) .. controls (3.31,0.3) and (6.95,1.4) .. (10.93,3.29)   ;
\draw    (159.04,396.97) -- (107,398.38) ;
\draw [shift={(105,398.43)}, rotate = 358.45] [color={rgb, 255:red, 0; green, 0; blue, 0 }  ][line width=0.75]    (10.93,-3.29) .. controls (6.95,-1.4) and (3.31,-0.3) .. (0,0) .. controls (3.31,0.3) and (6.95,1.4) .. (10.93,3.29)   ;
\draw    (132.53,441.62) -- (100.38,407.88) ;
\draw [shift={(99,406.43)}, rotate = 46.38] [color={rgb, 255:red, 0; green, 0; blue, 0 }  ][line width=0.75]    (10.93,-3.29) .. controls (6.95,-1.4) and (3.31,-0.3) .. (0,0) .. controls (3.31,0.3) and (6.95,1.4) .. (10.93,3.29)   ;
\draw    (205.01,440.49) -- (241.5,412.32) ;
\draw [shift={(243.08,411.1)}, rotate = 142.34] [color={rgb, 255:red, 0; green, 0; blue, 0 }  ][line width=0.75]    (10.93,-3.29) .. controls (6.95,-1.4) and (3.31,-0.3) .. (0,0) .. controls (3.31,0.3) and (6.95,1.4) .. (10.93,3.29)   ;
\draw    (141,340.43) -- (163.16,388.62) ;
\draw [shift={(164,390.43)}, rotate = 245.3] [color={rgb, 255:red, 0; green, 0; blue, 0 }  ][line width=0.75]    (10.93,-3.29) .. controls (6.95,-1.4) and (3.31,-0.3) .. (0,0) .. controls (3.31,0.3) and (6.95,1.4) .. (10.93,3.29)   ;
\draw    (264.05,319.55) .. controls (282.75,321.2) and (289.65,305.67) .. (265.94,310.65) ;
\draw [shift={(264.05,311.07)}, rotate = 346.53] [color={rgb, 255:red, 0; green, 0; blue, 0 }  ][line width=0.75]    (10.93,-3.29) .. controls (6.95,-1.4) and (3.31,-0.3) .. (0,0) .. controls (3.31,0.3) and (6.95,1.4) .. (10.93,3.29)   ;
\draw    (194.01,293.56) .. controls (197.88,306.29) and (212.58,317.4) .. (238.65,317.31) ;
\draw [shift={(240.26,317.29)}, rotate = 178.8] [color={rgb, 255:red, 0; green, 0; blue, 0 }  ][line width=0.75]    (10.93,-3.29) .. controls (6.95,-1.4) and (3.31,-0.3) .. (0,0) .. controls (3.31,0.3) and (6.95,1.4) .. (10.93,3.29)   ;
\draw    (82.89,350.07) -- (125.57,334.91) ;
\draw [shift={(127.45,334.24)}, rotate = 160.45] [color={rgb, 255:red, 0; green, 0; blue, 0 }  ][line width=0.75]    (10.93,-3.29) .. controls (6.95,-1.4) and (3.31,-0.3) .. (0,0) .. controls (3.31,0.3) and (6.95,1.4) .. (10.93,3.29)   ;
\draw [line width=1.5]    (355,333) -- (405,333.94) ;
\draw [shift={(408,334)}, rotate = 181.08] [color={rgb, 255:red, 0; green, 0; blue, 0 }  ][line width=1.5]    (14.21,-4.28) .. controls (9.04,-1.82) and (4.3,-0.39) .. (0,0) .. controls (4.3,0.39) and (9.04,1.82) .. (14.21,4.28)   ;
\draw [line width=1.5]    (411,327) -- (390.85,301.36) ;
\draw [shift={(389,299)}, rotate = 51.84] [color={rgb, 255:red, 0; green, 0; blue, 0 }  ][line width=1.5]    (14.21,-4.28) .. controls (9.04,-1.82) and (4.3,-0.39) .. (0,0) .. controls (4.3,0.39) and (9.04,1.82) .. (14.21,4.28)   ;
\draw [line width=1.5]    (374,301) -- (353.98,323.75) ;
\draw [shift={(352,326)}, rotate = 311.35] [color={rgb, 255:red, 0; green, 0; blue, 0 }  ][line width=1.5]    (14.21,-4.28) .. controls (9.04,-1.82) and (4.3,-0.39) .. (0,0) .. controls (4.3,0.39) and (9.04,1.82) .. (14.21,4.28)   ;
\draw [line width=1.5]    (438,370) -- (423.34,340.68) ;
\draw [shift={(422,338)}, rotate = 63.43] [color={rgb, 255:red, 0; green, 0; blue, 0 }  ][line width=1.5]    (14.21,-4.28) .. controls (9.04,-1.82) and (4.3,-0.39) .. (0,0) .. controls (4.3,0.39) and (9.04,1.82) .. (14.21,4.28)   ;
\draw [line width=1.5]    (65,52) -- (88.22,83.58) ;
\draw [shift={(90,86)}, rotate = 233.67] [color={rgb, 255:red, 0; green, 0; blue, 0 }  ][line width=1.5]    (14.21,-4.28) .. controls (9.04,-1.82) and (4.3,-0.39) .. (0,0) .. controls (4.3,0.39) and (9.04,1.82) .. (14.21,4.28)   ;
\draw [line width=1.5]    (106,289) -- (129.22,320.58) ;
\draw [shift={(131,323)}, rotate = 233.67] [color={rgb, 255:red, 0; green, 0; blue, 0 }  ][line width=1.5]    (14.21,-4.28) .. controls (9.04,-1.82) and (4.3,-0.39) .. (0,0) .. controls (4.3,0.39) and (9.04,1.82) .. (14.21,4.28)   ;
\draw [line width=1.5]    (471,272) -- (437.68,222.49) ;
\draw [shift={(436,220)}, rotate = 56.06] [color={rgb, 255:red, 0; green, 0; blue, 0 }  ][line width=1.5]    (14.21,-4.28) .. controls (9.04,-1.82) and (4.3,-0.39) .. (0,0) .. controls (4.3,0.39) and (9.04,1.82) .. (14.21,4.28)   ;
\draw [line width=1.5]    (342,325) .. controls (305.74,323.04) and (292.53,222.16) .. (318.37,180.48) ;
\draw [shift={(320,178)}, rotate = 124.99] [color={rgb, 255:red, 0; green, 0; blue, 0 }  ][line width=1.5]    (14.21,-4.28) .. controls (9.04,-1.82) and (4.3,-0.39) .. (0,0) .. controls (4.3,0.39) and (9.04,1.82) .. (14.21,4.28)   ;
\draw [line width=1.5]    (374,290) .. controls (339.35,273.17) and (318.42,247.52) .. (327.71,174.66) ;
\draw [shift={(328,172.43)}, rotate = 97.64] [color={rgb, 255:red, 0; green, 0; blue, 0 }  ][line width=1.5]    (14.21,-4.28) .. controls (9.04,-1.82) and (4.3,-0.39) .. (0,0) .. controls (4.3,0.39) and (9.04,1.82) .. (14.21,4.28)   ;
\draw [line width=1.5]    (410,340) -- (264.76,401.82) ;
\draw [shift={(262,403)}, rotate = 336.94] [color={rgb, 255:red, 0; green, 0; blue, 0 }  ][line width=1.5]    (14.21,-4.28) .. controls (9.04,-1.82) and (4.3,-0.39) .. (0,0) .. controls (4.3,0.39) and (9.04,1.82) .. (14.21,4.28)   ;

\draw (220.35,22.56) node [anchor=north west][inner sep=0.75pt]  [font=\scriptsize] [align=left] {a};
\draw (168.58,19.95) node [anchor=north west][inner sep=0.75pt]  [font=\scriptsize] [align=left] {a};
\draw (114.13,57.82) node [anchor=north west][inner sep=0.75pt]  [font=\scriptsize] [align=left] {a};
\draw (139.08,113.46) node [anchor=north west][inner sep=0.75pt]  [font=\scriptsize] [align=left] {a};
\draw (180.99,48.77) node [anchor=north west][inner sep=0.75pt]  [font=\scriptsize] [align=left] {a};
\draw (206.93,106.03) node [anchor=north west][inner sep=0.75pt]  [font=\scriptsize] [align=left] {a,b};
\draw (210.76,144.89) node [anchor=north west][inner sep=0.75pt]  [font=\scriptsize] [align=left] {a};
\draw (15.42,90.46) node [anchor=north west][inner sep=0.75pt]  [font=\scriptsize] [align=left] {a};
\draw (63.98,129.19) node [anchor=north west][inner sep=0.75pt]  [font=\scriptsize] [align=left] {a};
\draw (15.19,138.32) node [anchor=north west][inner sep=0.75pt]  [font=\scriptsize] [align=left] {a,b};
\draw (147.54,175.71) node [anchor=north west][inner sep=0.75pt]  [font=\scriptsize] [align=left] {a};
\draw (105.16,170.14) node [anchor=north west][inner sep=0.75pt]  [font=\scriptsize] [align=left] {a};
\draw (124.03,206.18) node [anchor=north west][inner sep=0.75pt]  [font=\scriptsize] [align=left] {a};
\draw (85.66,225.12) node [anchor=north west][inner sep=0.75pt]  [font=\scriptsize] [align=left] {a};
\draw (31.75,174.38) node [anchor=north west][inner sep=0.75pt]  [font=\scriptsize] [align=left] {a};
\draw (87.36,144.71) node [anchor=north west][inner sep=0.75pt]  [font=\scriptsize] [align=left] {b};
\draw (66.11,186.01) node [anchor=north west][inner sep=0.75pt]  [font=\scriptsize] [align=left] {b};
\draw (187.04,190.79) node [anchor=north west][inner sep=0.75pt]  [font=\scriptsize] [align=left] {b};
\draw (100.03,119.89) node [anchor=north west][inner sep=0.75pt]  [font=\scriptsize] [align=left] {b};
\draw (244.83,70.12) node [anchor=north west][inner sep=0.75pt]  [font=\scriptsize] [align=left] {b};
\draw (158.82,69.94) node [anchor=north west][inner sep=0.75pt]  [font=\scriptsize] [align=left] {b};
\draw (57.29,89.02) node [anchor=north west][inner sep=0.75pt]  [font=\scriptsize] [align=left] {b};
\draw (12.17,69.38) node [anchor=north west][inner sep=0.75pt]  [font=\scriptsize] [align=left] {1a};
\draw (37.43,112.81) node [anchor=north west][inner sep=0.75pt]  [font=\scriptsize] [align=left] {2a};
\draw (49.26,154.45) node [anchor=north west][inner sep=0.75pt]  [font=\scriptsize] [align=left] {3a};
\draw (30.09,192.88) node [anchor=north west][inner sep=0.75pt]  [font=\scriptsize] [align=left] {4a};
\draw (64.06,233) node [anchor=north west][inner sep=0.75pt]  [font=\scriptsize] [align=left] {5a};
\draw (96.39,199.96) node [anchor=north west][inner sep=0.75pt]  [font=\scriptsize] [align=left] {6a};
\draw (122.81,154.32) node [anchor=north west][inner sep=0.75pt]  [font=\scriptsize] [align=left] {7a};
\draw (157.57,202.66) node [anchor=north west][inner sep=0.75pt]  [font=\scriptsize] [align=left] {9a};
\draw (90.57,87.46) node [anchor=north west][inner sep=0.75pt]  [font=\scriptsize] [align=left] {8a};
\draw (190.35,121.98) node [anchor=north west][inner sep=0.75pt]  [font=\scriptsize] [align=left] {10a};
\draw (201.5,162.36) node [anchor=north west][inner sep=0.75pt]  [font=\scriptsize] [align=left] {12a};
\draw (140.19,40.99) node [anchor=north west][inner sep=0.75pt]  [font=\scriptsize] [align=left] {11a};
\draw (203.13,71.64) node [anchor=north west][inner sep=0.75pt]  [font=\scriptsize] [align=left] {13a};
\draw (183.84,10.87) node [anchor=north west][inner sep=0.75pt]  [font=\scriptsize] [align=left] {14a};
\draw (228.08,40.47) node [anchor=north west][inner sep=0.75pt]  [font=\scriptsize] [align=left] {15a};
\draw (488.35,27.56) node [anchor=north west][inner sep=0.75pt]  [font=\scriptsize] [align=left] {a};
\draw (436.58,23.95) node [anchor=north west][inner sep=0.75pt]  [font=\scriptsize] [align=left] {a};
\draw (382.13,61.82) node [anchor=north west][inner sep=0.75pt]  [font=\scriptsize] [align=left] {a};
\draw (407.08,117.46) node [anchor=north west][inner sep=0.75pt]  [font=\scriptsize] [align=left] {a};
\draw (448.99,52.77) node [anchor=north west][inner sep=0.75pt]  [font=\scriptsize] [align=left] {a};
\draw (474.93,110.03) node [anchor=north west][inner sep=0.75pt]  [font=\scriptsize] [align=left] {a,b};
\draw (478.76,148.89) node [anchor=north west][inner sep=0.75pt]  [font=\scriptsize] [align=left] {a};
\draw (283.42,94.46) node [anchor=north west][inner sep=0.75pt]  [font=\scriptsize] [align=left] {a};
\draw (331.98,133.19) node [anchor=north west][inner sep=0.75pt]  [font=\scriptsize] [align=left] {a};
\draw (283.19,142.32) node [anchor=north west][inner sep=0.75pt]  [font=\scriptsize] [align=left] {a,b};
\draw (415.54,179.71) node [anchor=north west][inner sep=0.75pt]  [font=\scriptsize] [align=left] {a};
\draw (373.16,174.14) node [anchor=north west][inner sep=0.75pt]  [font=\scriptsize] [align=left] {a};
\draw (392.03,210.18) node [anchor=north west][inner sep=0.75pt]  [font=\scriptsize] [align=left] {a};
\draw (380.66,235.12) node [anchor=north west][inner sep=0.75pt]  [font=\scriptsize] [align=left] {a};
\draw (287.75,175.38) node [anchor=north west][inner sep=0.75pt]  [font=\scriptsize] [align=left] {a};
\draw (355.36,148.71) node [anchor=north west][inner sep=0.75pt]  [font=\scriptsize] [align=left] {b};
\draw (334.11,190.01) node [anchor=north west][inner sep=0.75pt]  [font=\scriptsize] [align=left] {b};
\draw (455.04,194.79) node [anchor=north west][inner sep=0.75pt]  [font=\scriptsize] [align=left] {b};
\draw (368.03,123.89) node [anchor=north west][inner sep=0.75pt]  [font=\scriptsize] [align=left] {b};
\draw (512.83,74.12) node [anchor=north west][inner sep=0.75pt]  [font=\scriptsize] [align=left] {b};
\draw (426.82,73.94) node [anchor=north west][inner sep=0.75pt]  [font=\scriptsize] [align=left] {b};
\draw (325.29,93.02) node [anchor=north west][inner sep=0.75pt]  [font=\scriptsize] [align=left] {b};
\draw (280.17,73.38) node [anchor=north west][inner sep=0.75pt]  [font=\scriptsize] [align=left] {1c};
\draw (305.43,116.81) node [anchor=north west][inner sep=0.75pt]  [font=\scriptsize] [align=left] {2c};
\draw (317.26,158.45) node [anchor=north west][inner sep=0.75pt]  [font=\scriptsize] [align=left] {3c};
\draw (265.09,194.88) node [anchor=north west][inner sep=0.75pt]  [font=\scriptsize] [align=left] {4c};
\draw (382.06,260) node [anchor=north west][inner sep=0.75pt]  [font=\scriptsize] [align=left] {5c};
\draw (364.39,203.96) node [anchor=north west][inner sep=0.75pt]  [font=\scriptsize] [align=left] {6c};
\draw (389,157.43) node [anchor=north west][inner sep=0.75pt]  [font=\scriptsize] [align=left] {7c};
\draw (423,206.43) node [anchor=north west][inner sep=0.75pt]  [font=\scriptsize] [align=left] {9c};
\draw (358.57,91.46) node [anchor=north west][inner sep=0.75pt]  [font=\scriptsize] [align=left] {8c};
\draw (458.35,125.98) node [anchor=north west][inner sep=0.75pt]  [font=\scriptsize] [align=left] {10c};
\draw (469.5,166.36) node [anchor=north west][inner sep=0.75pt]  [font=\scriptsize] [align=left] {12c};
\draw (408.19,44.99) node [anchor=north west][inner sep=0.75pt]  [font=\scriptsize] [align=left] {11c};
\draw (471.13,75.64) node [anchor=north west][inner sep=0.75pt]  [font=\scriptsize] [align=left] {13c};
\draw (451.84,14.87) node [anchor=north west][inner sep=0.75pt]  [font=\scriptsize] [align=left] {14c};
\draw (496.08,44.47) node [anchor=north west][inner sep=0.75pt]  [font=\scriptsize] [align=left] {15c};
\draw (258.35,259.56) node [anchor=north west][inner sep=0.75pt]  [font=\scriptsize] [align=left] {a};
\draw (207.58,257.95) node [anchor=north west][inner sep=0.75pt]  [font=\scriptsize] [align=left] {a};
\draw (153.13,295.82) node [anchor=north west][inner sep=0.75pt]  [font=\scriptsize] [align=left] {a};
\draw (178.08,351.46) node [anchor=north west][inner sep=0.75pt]  [font=\scriptsize] [align=left] {a};
\draw (219.99,286.77) node [anchor=north west][inner sep=0.75pt]  [font=\scriptsize] [align=left] {a};
\draw (245.93,344.03) node [anchor=north west][inner sep=0.75pt]  [font=\scriptsize] [align=left] {a,b};
\draw (249.76,382.89) node [anchor=north west][inner sep=0.75pt]  [font=\scriptsize] [align=left] {a};
\draw (54.42,328.46) node [anchor=north west][inner sep=0.75pt]  [font=\scriptsize] [align=left] {a};
\draw (102.98,367.19) node [anchor=north west][inner sep=0.75pt]  [font=\scriptsize] [align=left] {a};
\draw (54.19,376.32) node [anchor=north west][inner sep=0.75pt]  [font=\scriptsize] [align=left] {a,b};
\draw (186.54,413.71) node [anchor=north west][inner sep=0.75pt]  [font=\scriptsize] [align=left] {a};
\draw (144.16,408.14) node [anchor=north west][inner sep=0.75pt]  [font=\scriptsize] [align=left] {a};
\draw (163.03,444.18) node [anchor=north west][inner sep=0.75pt]  [font=\scriptsize] [align=left] {a};
\draw (124.66,463.12) node [anchor=north west][inner sep=0.75pt]  [font=\scriptsize] [align=left] {a};
\draw (70.75,412.38) node [anchor=north west][inner sep=0.75pt]  [font=\scriptsize] [align=left] {a};
\draw (126.36,382.71) node [anchor=north west][inner sep=0.75pt]  [font=\scriptsize] [align=left] {b};
\draw (105.11,424.01) node [anchor=north west][inner sep=0.75pt]  [font=\scriptsize] [align=left] {b};
\draw (226.04,428.79) node [anchor=north west][inner sep=0.75pt]  [font=\scriptsize] [align=left] {b};
\draw (139.03,357.89) node [anchor=north west][inner sep=0.75pt]  [font=\scriptsize] [align=left] {b};
\draw (283.83,308.12) node [anchor=north west][inner sep=0.75pt]  [font=\scriptsize] [align=left] {b};
\draw (197.82,307.94) node [anchor=north west][inner sep=0.75pt]  [font=\scriptsize] [align=left] {b};
\draw (96.29,327.02) node [anchor=north west][inner sep=0.75pt]  [font=\scriptsize] [align=left] {b};
\draw (51.17,307.38) node [anchor=north west][inner sep=0.75pt]  [font=\scriptsize] [align=left] {1b};
\draw (76.43,350.81) node [anchor=north west][inner sep=0.75pt]  [font=\scriptsize] [align=left] {2b};
\draw (88.26,392.45) node [anchor=north west][inner sep=0.75pt]  [font=\scriptsize] [align=left] {3b};
\draw (69.09,430.88) node [anchor=north west][inner sep=0.75pt]  [font=\scriptsize] [align=left] {4b};
\draw (103.06,471) node [anchor=north west][inner sep=0.75pt]  [font=\scriptsize] [align=left] {5b};
\draw (135.39,437.96) node [anchor=north west][inner sep=0.75pt]  [font=\scriptsize] [align=left] {6b};
\draw (161.81,392.32) node [anchor=north west][inner sep=0.75pt]  [font=\scriptsize] [align=left] {7b};
\draw (196.57,440.66) node [anchor=north west][inner sep=0.75pt]  [font=\scriptsize] [align=left] {9b};
\draw (129.57,325.46) node [anchor=north west][inner sep=0.75pt]  [font=\scriptsize] [align=left] {8b};
\draw (229.35,359.98) node [anchor=north west][inner sep=0.75pt]  [font=\scriptsize] [align=left] {10b};
\draw (240.5,400.36) node [anchor=north west][inner sep=0.75pt]  [font=\scriptsize] [align=left] {12b};
\draw (179.19,278.99) node [anchor=north west][inner sep=0.75pt]  [font=\scriptsize] [align=left] {11b};
\draw (242.13,309.64) node [anchor=north west][inner sep=0.75pt]  [font=\scriptsize] [align=left] {13b};
\draw (222.84,248.87) node [anchor=north west][inner sep=0.75pt]  [font=\scriptsize] [align=left] {14b};
\draw (267.08,278.47) node [anchor=north west][inner sep=0.75pt]  [font=\scriptsize] [align=left] {15b};
\draw (342,328) node [anchor=north west][inner sep=0.75pt]  [font=\scriptsize] [align=left] {7};
\draw (377,291) node [anchor=north west][inner sep=0.75pt]  [font=\scriptsize] [align=left] {6};
\draw (412,328) node [anchor=north west][inner sep=0.75pt]  [font=\scriptsize] [align=left] {5};
\draw (403.08,306.46) node [anchor=north west][inner sep=0.75pt]  [font=\scriptsize] [align=left] {a};
\draw (357.08,299.46) node [anchor=north west][inner sep=0.75pt]  [font=\scriptsize] [align=left] {a};
\draw (374.08,336.46) node [anchor=north west][inner sep=0.75pt]  [font=\scriptsize] [align=left] {a};
\draw (436,371) node [anchor=north west][inner sep=0.75pt]  [font=\scriptsize] [align=left] {4};
\draw (436.08,351.46) node [anchor=north west][inner sep=0.75pt]  [font=\scriptsize] [align=left] {a};
\draw (58,40) node [anchor=north west][inner sep=0.75pt]  [font=\scriptsize] [align=left] {1};
\draw (77.08,53.46) node [anchor=north west][inner sep=0.75pt]  [font=\scriptsize] [align=left] {a};
\draw (99,277) node [anchor=north west][inner sep=0.75pt]  [font=\scriptsize] [align=left] {2};
\draw (117.08,291.46) node [anchor=north west][inner sep=0.75pt]  [font=\scriptsize] [align=left] {a};
\draw (473,275) node [anchor=north west][inner sep=0.75pt]  [font=\scriptsize] [align=left] {1};
\draw (456.08,239.46) node [anchor=north west][inner sep=0.75pt]  [font=\scriptsize] [align=left] {a};
\draw (292.83,251.12) node [anchor=north west][inner sep=0.75pt]  [font=\scriptsize] [align=left] {b};
\draw (321.83,250.12) node [anchor=north west][inner sep=0.75pt]  [font=\scriptsize] [align=left] {b};
\draw (342.83,371.12) node [anchor=north west][inner sep=0.75pt]  [font=\scriptsize] [align=left] {b};

\end{tikzpicture}